\documentclass{amsart}%
\usepackage{amsfonts}
\usepackage{amsmath}
\usepackage{amssymb}
\usepackage{graphicx}
\usepackage[latin1]{inputenc}
\usepackage[english]{babel}%
\newtheorem{theorem}{Theorem}[section]
\theoremstyle{plain}

\newtheorem{claim}{Claim}

\newtheorem{corollary}[theorem]{Corollary}

\newtheorem{definition}[theorem]{Definition}

\newtheorem{lemma}[theorem]{Lemma}

\newtheorem{proposition}[theorem]{Proposition}
\newtheorem{remark}[theorem]{Remark}

\numberwithin{equation}{section}
\begin{document}
\title[Parabolic Equations and Markov Processes on Adeles]{Parabolic Type Equations and Markov Stochastic Processes on Adeles}
\author[S. M. Torba]{Sergii M. Torba}
\address[S. Torba and W. Zúñiga-Galindo]{Department of Mathematics, CINVESTAV del IPN,
Unidad Querétaro,\\
Libramiento Norponiente No. 2000, Fracc. Real de Juriquilla,\\
Querétaro, Qro. C.P. 76230 MEXICO}
\email[S. Torba]{storba@math.cinvestav.edu.mx}
\author{W. A. Zúñiga-Galindo}
\email[W. Zúñiga-Galindo]{wazuniga@math.cinvestav.edu.mx}
\thanks{Research of the first named author was partially supported by SNSF,
Switzerland (JRP IZ73Z0 of SCOPES 2009--2012)}
\thanks{The second author was partially supported by CONACYT under Grant \# 127794.}
\dedicatory{ }
\begin{abstract}
In this paper we study the Cauchy problem for new classes of parabolic type
pseudodifferential equations over the rings of finite adeles and adeles. We
show that the adelic topology is metrizable and give an explicit metric. We
find explicit representations of the fundamental solutions (the heat kernels).
These fundamental solutions are transition functions of Markov processes which are adelic analogues of the Archimedean Brownian motion. We show that the Cauchy problems for these equations are well-posed and find explicit representations of the evolution semigroup and formulas for
the solutions of homogeneous and non-homogeneous equations.

\end{abstract}
\maketitle

\section{Introduction}

During the last twenty years the interest on stochastic models on $p$-adics
and adeles has been increasing mainly because these models are convenient for
describing phenomena whose space of states display a hierarchical structure.
All these developments have been motivated by a conjecture in statistical
physics asserting that the non exponential relaxation of several models
describing complex systems, such as glasses and proteins, is a consequence of
a hierarchical structure of the state space which can in turn be put in
connection with $p$-adic structures. The pioneering work of Avetisov
\textit{et al.} on $p$-adic techniques for describing spontaneous symmetry
breaking in the models of spin glasses and relaxation processes in complex
systems gives a very strong motivation for developing a theory of parabolic
type pseudodifferential equations and their corresponding stochastic processes
on $p$-adics and adeles, see \cite{A-B}, \cite{Av-1}, \cite{Av-2},
\cite{Av-3}, \cite{Av-4}, \cite{Av-5}, \cite{Av-6}, \cite{Av-7}, \cite{B},
\cite{Blair}, \cite[and references therein]{Dra-Kh-K-V}, \cite{Ga-Zu}
\cite{H-S-S-S}, \cite{K-M}, \cite[and references therein]{Koch}, \cite{K-A-1},
\cite{K-A-2}, \cite{P-S}, \cite{R-Zu}, \cite{Va1}, \cite[and references
therein]{V-V-Z}, \cite{Ya}, \cite{Zu}, among others.

Another two motivations for studying pseudodifferential equations on adeles
are the following. In \cite{Haran} Haran established a connection between
explicit formulas for the Riemann zeta function and adelic pseudodifferential
operators, see also \cite{C}. In \cite{Manin}\ Manin posed the conjecture that the physical space
is adelic, which can be considered as an extension of the Volovich conjecture
on the non Archimedean nature of the physical space at the Planck scale
\cite{Vol2}, \cite{Vo}, \cite{Va2}. This conjecture conducts naturally to
consider models involving partial differential equations on adelic spaces.
Some preliminary results such as studying pseudodifferential operators are
presented in \cite[and references therein]{Dra}, \cite{D-R-K}, \cite{Kh-Ra},
\cite{R-R}.

In this article we work exclusively with complex valued functions on adeles,
this due to the fact that most of the physical models that motivate our theory
require `real valued probabilities'. However, recently new models of complex
systems using `$p$-adic valued probabilities' have emerged, see e.g.
\cite{K-M-1}, \cite{K-M-2}. The use of complex-valued functions allow us to
take advantage of the classical harmonic and functional analysis. However,
the classical derivative is not defined for complex-valued functions on
adeles, implying the consideration of pseudodifferential operators. For the
sake of simplicity we formulated all our results for finite adeles and adeles
on $\mathbb{Q}$, however all the results are still valid if the field of
rational numbers $\mathbb{Q}$ is replaced by a global field, i.e. by an
algebraic number field, or by the function field of an algebraic curve over a
finite field. We study the Cauchy problem for parabolic type
pseudodifferential equations over the rings of finite adeles and adeles
involving a natural generalization of the Taibleson operator \cite{A-K-S},
\cite{R-Zu}. The considered pseudodifferential operator is not a
straightforward generalization of the Taibleson operator and is natural only
from the point of view of its connection with the adelic topology and the
Fourier transform. By so far we are unaware of any similar results. Other
adelic pseudodifferential operators with different symbols have been studied
in \cite{Haran}, \cite{Kh-Ra}, \cite{Kh-K-Sh}, \cite{R-R}.

The article is organized as follows. In Section \ref{Section2} we summarize
some well-known results on $p$-adic and adelic analysis. In Sections
\ref{Sect3}, \ref{Sect4} we introduce metric structures on the rings of finite
adeles and adeles, see Propositions \ref{Prop1}, \ref{pro1a}. These metric
structures induce the adelic topology and are naturally connected with the Fourier transform. In addition, they allow us to use classical results on Markov processes, see e.g.
\cite{Dyn}. We compute the Fourier transform of radial functions
defined on the ring of finite adeles, see Theorem \ref{LemmaFourierRadial},
and we introduce adelic analogues of the Taibleson operators and Lizorkin spaces of the second kind and prove some basic properties of them. In Section \ref{SectAdelicHeat} we study the heat kernels
on the ring of finite adeles, see Definition \ref{DefAdelicheathkernel} and
Theorem \ref{Theo1}. We give an `explicit formula' for the heat kernel as a
series involving Chebyshev type functions, i.e. products of powers of primes,
some arithmetic operators and exponential functions depending on $t$, see
Proposition \ref{pro2}. We require the Prime Number Theorem to establish the
existence of the adelic heat kernels, see Proposition \ref{pro1}. In Section
\ref{SectMarkov} we show that the adelic heat kernels are the transition
functions of Markov processes, see Theorem \ref{Theo2}. In Sections
\ref{SEctHeatKernelA}, \ref{SEctMarkovA} we study the heat kernels on the
ring of adeles, see Definition \ref{DefHeatKerA} and Theorem \ref{Theo5}, and
show that the heat kernels are the transition functions of Markov processes, see
Theorem \ref{Theo6}. In Sections \ref{SectCauchy}, \ref{SectionCauchyA} we
study Cauchy problems for parabolic type equations involving adelic versions
of the Taibleson operator. We show that these problems are well-posed and find
explicit formulas for the solutions of homogeneous and non-homogeneous
equations, see Proposition \ref{Corr3}, Theorems \ref{Theo4minus}, \ref{Theo3},
\ref{Theo4}, Proposition \ref{propoAdelic} and Theorems \ref{Theo9},
\ref{Theo10}.

Finally we hope that this article will raise interest on studying
pseudodifferential equations and stochastic processes on adeles. We are still at the beginning to develop a complete theory, there are
many open problems and questions, among them, we propose the study of adelic Schrödinger equations and their connection with Feynman and Feynman-Kac integrals.

\section{\label{Section2}Preliminaries}

In this section we fix the notation and collect some basic results on $p$-adic
and adelic analysis that we will use through the article. For a detailed
exposition on $p$-adic and adelic analysis the reader may consult
\cite{A-K-S}, \cite{G-H}, \cite{Koch}, \cite{R-V}, \cite{Taibleson},
\cite{V-V-Z}.

\subsection{Adeles on $\mathbb{Q}$}

\label{SubsectAdeles}

Let $p$ be a fixed prime number, and let $x$ be a nonzero rational number.
Then $x$ may be represented uniquely as $x=p^{k}\frac{a}{b}$ with $p\nmid ab$
and $k\in\mathbb{Z}$. The function
\[
|x|_{p}:=%
\begin{cases}
p^{-k} & \text{ if }x\neq0,\\
0 & \text{ if }x=0
\end{cases}
\]
gives rise to a non-Archimedean absolute value on $\mathbb{Q}$. The field of
$p$-adic numbers $\mathbb{Q}_{p}$ is defined as the completion of $\mathbb{Q}$
with respect to the distance induced by $|\cdot|_{p}$. Any non-zero $p$-adic
number $x_{p}$ has a unique representation of the form%
\begin{equation}
x_{p}=p^{\gamma}\sum_{i=0}^{\infty}a_{i}p^{i}, \label{eq1}%
\end{equation}
where $\gamma=\gamma(x_{p})\in\mathbb{Z},\ a_{i}\in\{0,1,\dots,p-1\},\ a_{0}%
\neq0$. Series (\ref{eq1}) converges in the $p$-adic absolute value. The
integer $\gamma$ is called the $p$\textit{-adic order of} $x_{p}$, and it will
be denoted as $\operatorname{ord}_{p}(x_{p})$, with $\operatorname{ord}%
_{p}(0):=+\infty$. Note that $|x_{p}|_{p}=p^{-\operatorname{ord}_{p}\left(
x_{p}\right)  }$. With the topology induced by $|\cdot|_{p}$, $\mathbb{Q}_{p}$
is a locally compact topological field. The unit ball $\mathbb{Z}_{p}$\ of
$\mathbb{Q}_{p}$ is a compact topological ring. Let $dx_{p}$ denote the Haar
measure of the topological group $\left(  \mathbb{Q}_{p},+\right)  $
normalized by the condition $\operatorname{vol}(\mathbb{Z}_{p})=1$. For a
detailed presentation of the integration theory on $\mathbb{Q}_{p}$ see
\cite{G-H}, \cite{V-V-Z}.

Along the article, the variables $p$, $q$ will denote `primes', including `the
infinite prime', denoted by $\infty$. To each prime $p$ corresponds an
absolute value $| \cdot| _{p}$ on $\mathbb{Q}$, with $| \cdot| _{\infty}$
corresponding to the usual Euclidean norm. In addition, $\mathbb{Q}_{p}$
denotes the completion of $\mathbb{Q}$ with respect to $| \cdot| _{p}$, note
that $\mathbb{Q}_{\infty}=\mathbb{R}$.

The \textit{ring of adeles of} $\mathbb{Q}$, denoted $\mathbb{A}$, is defined
by%
\[
\mathbb{A}=\bigl\{  \left(  x_{\infty},x_{2},x_{3},\ldots\right)  :x_{p}%
\in\mathbb{Q}_{p},\text{ and }x_{p}\in\mathbb{Z}_{p}\text{ for all but
finitely many }p\bigr\}.
\]

Alternatively, we can define $\mathbb{A}$ as \textit{the restricted product}
of the $\mathbb{Q}_{p}$ with respect to the $\mathbb{Z}_{p}$. The
componentwise addition and multiplication give to $\mathbb{A}$ a ring
structure. Furthermore, $\mathbb{A}$ can be made into a locally compact
topological ring by taking as a base for the topology, certainly the\textit{
restricted product topology}, all the sets of the form $U\times%
{\textstyle\prod\nolimits_{p\notin S}}
\mathbb{Z}_{p}$ where $S$ is any finite set of primes containing $\infty$, and
$U$ is any open subset in $%
{\textstyle\prod\nolimits_{p\in S}}
\mathbb{Q}_{p}$.

The restricted product topology is not equal to the product topology. However,
the following relation holds. Take $S$ as before and consider the group%
\[
G_{S}=%
{\displaystyle\prod\limits_{p\in S}}
\mathbb{Q}_{p}\times%
{\displaystyle\prod\limits_{p\notin S}}
\mathbb{Z}_{p}.
\]
Then, the product topology on $G_{S}$ is identical to the one induced by the
restricted product topology on $G_{S}$, thus $G_{S}$ is a locally compact
subgroup of $\mathbb{A}$, and the locally compact topological group $\left(
\mathbb{A},+\right)  $ has a Haar measure, denoted $dx_{\mathbb{A}}$, which
coincides on $G_{S}$ with the product measure $%
{\textstyle\prod\nolimits_{p}}
dx_{p}$, where $dx_{\infty}$ is the Lebesgue measure of $\mathbb{R}$.
We also note that any set of the form
\begin{equation}
\prod_{p\in S}p^{l_{p}}\mathbb{Z}_{p}\times\prod_{p\notin S}\mathbb{Z}_{p},
\label{compactsubset}%
\end{equation}
where $l_{p}$ are arbitrary integers, is a compact subset of $\mathbb{A}$.

The \textit{ring of finite adeles} over $\mathbb{Q}$, denoted $\mathbb{A}_{f}%
$, is defined by
\[
\mathbb{A}_{f}=\bigl\{\left(  x_{2},x_{3},\ldots\right)  :x_{p}\in
\mathbb{Q}_{p},\text{ and }x_{p}\in\mathbb{Z}_{p}\text{ for all but finitely
many }p\bigr\}.
\]
From now on, we consider $\mathbb{A}_{f}$ as a topological ring with respect
to the restricted product topology. Then $\mathbb{A}=\mathbb{R\times A}_{f}$.
Since $\left(  \mathbb{A}_{f},+\right)  $ is a locally compact topological
group, it has a Haar measure, denoted $dx_{\mathbb{A}_{f}}$, which agrees with
the product measure $%
{\textstyle\prod\nolimits_{p<\infty}}
dx_{p}$ on open subgroups of type
\[%
{\displaystyle\prod\limits_{p\leq N}}
\mathbb{Q}_{p}\times%
{\displaystyle\prod\limits_{p>N}}
\mathbb{Z}_{p},\quad\text{with }N\in\mathbb{N}\text{.}%
\]
Furthermore $dx_{\mathbb{A}}=dx_{\infty}dx_{\mathbb{A}_{f}}$. For a detailed
presentation of the integration theory on $\mathbb{A}$ and $\mathbb{A}_{f}$
see \cite[Chapter 1]{G-H}, see also \cite{R-V}, \cite{We}.

In this article we work exclusively with complex valued functions on adeles.
Having complex valued functions defined on a locally compact topological
group, we have the notion of continuous function and may use the functional
spaces $L^{\varrho}(\mathbb{A}_{f})$ and $L^{\varrho}(\mathbb{A})$, $\rho
\geq1$ defined in the standard way.

For studying solutions of parabolic equations we need notations for several
spaces of functions which depend on time and adelic (space) variables. We
denote by:

\begin{itemize}
\item[(i)] $C(I, X)$ the space of continuous functions $u$ on a time interval
$I$ with values in $X$;

\item[(ii)] $C^{1}(I, X)$ the space of continuously differentiable functions
$u$ on a time interval $I$ such that $u^{\prime}\in X$;

\item[(iii)] $L^{1}(I, X)$ the space of measurable functions $u$ on $I$ with
values in $X$ such that $\| u\|$ is integrable;

\item[(iv)] $W^{1,1}(I, X)$ the space of measurable functions $u$ on $I$ with
values in $X$ such that $u^{\prime}\in L^{1}\bigl(I, X\bigr)$.
\end{itemize}

\subsection{Fourier transform on adeles}

\label{SubSectFourier}

Let $p$ be a finite prime, and let $\chi_{p}:\mathbb{Q}_{p}\rightarrow
\mathbb{C}^{\times}$ be the additive character defined by
\[
\chi_{p}\left(  x_{p}\right)  =\exp\left(  -2\pi i\left\{  x_{p}\right\}
\right)  ,
\]
where%
\[
\left\{  x_{p}\right\}  :=%
\begin{cases}
\sum\limits_{i=-k}^{-1}a_{i}p^{i} & \text{if}\ x=\sum\limits_{i=-k}^{\infty
}a_{i}p^{i}\text{ with }k>0\text{ and }0\leq a_{i}\leq p-1,\\
0 & \text{otherwise.}%
\end{cases}
\]

A function $f_{p}:\mathbb{Q}_{p}\rightarrow\mathbb{C}$ which is locally
constant with compact support is called a \textit{Bruhat-Schwartz function}.
The space of such functions is denoted as $\mathcal{S}\left(  \mathbb{Q}%
_{p}\right)  $. Note that in $p$-adic analysis the space $\mathcal{S}\left(
\mathbb{Q}_{p}\right)  $ coincides with the space of test functions
$\mathcal{D}(\mathbb{Q}_{p})$, see \cite{V-V-Z} for details. For $f_{p}%
\in\mathcal{S}\left(  \mathbb{Q}_{p}\right)  $, its Fourier transform
$\widehat{f_{p}}$ is defined by
\[
\widehat{f_{p}}\left(  \xi_{p}\right)  =\int_{\mathbb{Q}_{p}}\chi_{p}\left(
-x_{p}\xi_{p}\right)  f_{p}\left(  x_{p}\right)  dx_{p}.
\]
The Fourier transform induces a linear isomorphism of $\mathcal{S}\left(
\mathbb{Q}_{p}\right)  $\ onto $\mathcal{S}\left(  \mathbb{Q}_{p}\right)  $
satisfying $\widehat{\widehat{f_{p}}}\left(  \xi_{p}\right)  =f_{p}\left(
-\xi_{p}\right)  $.

In the case $p=\infty$ the additive character is defined by $\chi_{\infty
}\left(  x_{\infty}\right)  :=\exp\left(  2\pi ix_{\infty}\right)  $. Let
$\mathcal{S}(\mathbb{R})$ denote the Schwartz space. The Fourier transform of
$f_{\infty}\in\mathcal{S}(\mathbb{R})$, denoted $\widehat{f_{\infty}}$, is
defined by%
\[
\widehat{f_{\infty}}\left(  \xi_{\infty}\right)  =\int_{\mathbb{R}}%
\chi_{\infty}\left(  -x_{\infty}\xi_{\infty}\right)  f_{\infty}\left(
x_{\infty}\right)  dx_{\infty}.
\]
The Fourier transform induces a linear isomorphism of $\mathcal{S}\left(
\mathbb{Q}_{\infty}\right)  $\ onto $\mathcal{S}\left(  \mathbb{Q}_{\infty
}\right)  $ satisfying $\widehat{\widehat{f_{\infty}}}\left(  \xi_{\infty
}\right)  =f_{\infty}\left(  -\xi_{\infty}\right)  $.


The \textit{additive adelic character} $\chi:\mathbb{A\rightarrow C}$ is
defined by%
\[
\chi\left(  x\right)  =\prod_{p} \chi_{p}\left(  x_{p}\right)  \quad\text{ for
}x=\left(  x_{\infty},x_{2},x_{3},\ldots\right)  .
\]
An adelic function is said to be \textit{Bruhat-Schwartz} if it can be
expressed as a finite linear combination, with complex coefficients, of
factorizable functions $f=%
{\textstyle\prod\nolimits_{p\leq\infty}}
f_{p}$, where $f_{p}$ satisfies the following conditions: (A1) $f_{\infty}%
\in\mathcal{S}(\mathbb{R})$; (A2) $f_{p}\in\mathcal{S}\left(  \mathbb{Q}%
_{p}\right)  $ for $p<\infty$; (A3)\ $f_{p}$ is the characteristic function of
$\mathbb{Z}_{p}$ for all but finitely many $p<\infty$. The adelic space of
Bruhat-Schwartz functions is denoted as $\mathcal{S}(\mathbb{A})$. The space
of Bruhat-Schwartz functions $\mathcal{S}(\mathbb{A}_{f})$ is defined in a
similar form except that only conditions A2 and A3 are required.

The Fourier transform of a factorizable adelic Bruhat-Schwartz function is
defined by
\begin{equation}\label{defFTfactorisable}
\widehat{f}\left(  \xi\right)  =\prod_{p\leq\infty}\int_{\mathbb{Q}_{p}}%
f_{p}\left(  x_{p}\right)  \chi\left(  -x_{p}\xi_{p}\right)  dx_{p}.
\end{equation}
This definition may be extended to arbitrary adelic Bruhat-Schwartz functions
by linearity. The Fourier transform gives a linear isomorphism of
$\mathcal{S}\left(  \mathbb{A}\right)  $\ to $\mathcal{S}\left(
\mathbb{A}\right)  $ satisfying $\hat{\hat{f}}\left(  \xi\right)  =f\left(
-\xi\right)  $. Analogous definitions and results hold for the Fourier
transform on $\mathbb{A}_{f}$. The Fourier transform may be extended to the
space $L^{2}(\mathbb{A})$ (or to $L^{2}(\mathbb{A}_{f})$), where it is a
unitary operator and the Steklov--Parseval equality holds.

We will also use the notation $\mathcal{F}\varphi$ for the Fourier transform
and $\mathcal{F}^{-1}\varphi$ for the inverse Fourier transform. We used as a
main reference for this section \cite[Chapter 1]{G-H}, see also \cite{Kh-Ra},
\cite{R-V}, \cite{We}.

Since $\mathbb{A}$ (resp. $\mathbb{A}_{f}$) is a locally compact topological
group, a convolution operation between functions is also defined on
$\mathcal{S}(\mathbb{A})$ and $L^{2}(\mathbb{A})$ (resp. $\mathcal{S}%
(\mathbb{A}_{f})$ and $L^{2}(\mathbb{A}_{f})$). It is connected with the
Fourier transform in the usual way, see e.g. \cite{Rudin} for details.

\section{\label{Sect3}Metric structures, Distributions and Pseudodifferential
Operators on $\mathbb{A}_{f}$}

\subsection{A structure of complete metric space for the finite adeles}

In the previous section the restricted product topology on adeles was
described. By so far authors are unaware of any article introducing metric on
adeles producing the same topology. We show that for the finite adeles the
topology is metrizable and present a non-Archimedean metric on $\mathbb{A}%
_{f}$. Moreover, in this metric each ball is a compact set and the Fourier
transform of a radial function is again a radial function. Hence, despite of
the complicated form of the presented metric we believe that it is natural for
the ring of finite adeles.

Consider the following two functions:
\begin{equation}
\label{Norm0}\|x\|_{1} := \max_{p} | x_{p}|_{p},\qquad x\in\mathbb{A}_{f},
\end{equation}
and
\begin{equation}
\label{Norm1}\|x\|_{0}:=\max_{p}\frac{|x_{p}|_{p}}{p},\qquad x\in
\mathbb{A}_{f}.
\end{equation}
Both functions are well defined and may be used to introduce a metric on
$\mathbb{A}_{f}$. However, the topology induced by the metric $\|x-y\|_{1}$
does not coincide with the restricted product topology which may be easily
seen from the following example. The sequence of adeles $x^{(k)} :=
\bigl( \underbrace{0,0,\ldots,0}_{k-1}, 1, 0,\ldots\bigr)$, $k\in\mathbb{N}$
converges to 0 in the restricted product topology, but does not converge in
the metric generated by $\|\cdot\|_{1}$. The metric $\|x-y\|_{0}$ induces the
same topology as the restricted product topology, however it does not satisfy
the above mentioned properties. For instance, with respect to this metric only
balls of radiuses less than 1 are compact, and the Fourier transform of a radial
function is not necessary a radial function. We left checking of these
statements to reader, all required proofs may be obtained similarly to the
proofs in this article.

To overcome the mentioned problems we define a function
\begin{equation}
\label{Norm}\|x\| :=
\begin{cases}
\|x\|_{0} & \text{if } x\in\prod_{p}\mathbb{Z}_{p},\\
\|x\|_{1} & \text{if } x\not \in \prod_{p}\mathbb{Z}_{p},
\end{cases}
\end{equation}
for arbitrary $x\in\mathbb{A}_{f}$. Note that $\|x\|_{0}\le\|x\|\le\|x\|_{1}$
for any $x\in\mathbb{A}_{f}$. We introduce the function (our metric)
\begin{equation}
\label{MetrAf}\rho(x,y) := \|x-y\|,\qquad x,y\in\mathbb{A}_{f}.
\end{equation}

The function $\Vert\cdot\Vert$ can be also represented as
\[
\Vert x\Vert=\max_{p}p^{-[[\operatorname{ord}_{p}(x_{p})]]},\qquad
x\in\mathbb{A}_{f}\setminus\{0\},
\]
where
\begin{equation}
\lbrack\lbrack t]]:=%
\begin{cases}
\lbrack t] & \text{if }t\geq0\\[0pt]%
\lbrack t]+1 & \text{if }t<0,
\end{cases}
\label{def[[]]}%
\end{equation}
here $[ \cdot] $ denotes the integer part function.

\begin{remark}
\label{Rmk Metric Range} The range of values of the function $\rho$ coincides
with the set $\{0\}\cup\bigl\{ p^{j}: p\text{ is prime, }j\in\mathbb{Z}%
\setminus\{0\}\bigr\}$.
\end{remark}

\begin{remark}
It may seem odd that the proposed metric does not attain the value 1. It is
possible to define another metric, using instead of $\|\cdot\|_{0}$ in
\eqref{Norm} the function $(\|\cdot\|_{0})_{+}$, see \eqref{nplusdef} for the
definition of `$_{+}$' operator. Due to Bertrand's postulate the generated
metrics are equivalent. In some cases like in Corollary \ref{FTBall} such
change simplifies formulas. However, calculations with this metric become more
complicated. For this reason we do not use it in this article.
\end{remark}

\begin{proposition}
\label{Prop1}The restricted product topology on $\mathbb{A}_{f}$ is
metrizable, the metric is given by \eqref{MetrAf}. Furthermore, $\left(
\mathbb{A}_{f},\rho\right)  $ is a complete non-Archimedean metric space.
\end{proposition}

\begin{proof}
The fact that $\rho(x,y)$ is a non-Archimedean metric is a consequence of the
fact that
\[
\|x+y\|\le\max\{ \|x\|, \|y\|\},\qquad x,y\in\mathbb{A}_{f},
\]
which can be checked easily case by case.


We now show that $\left(  \mathbb{A}_{f},\rho\right)  $ is a complete metric
space. Let $x^{\left(  n\right)  }=\bigl(x_{p}^{\left(  n\right)  }\bigr)_{p}$
be a Cauchy sequence in $\mathbb{A}_{f}$ with respect to $\rho$. Since we have
coordinate-wise convergence, we may define $\widetilde{x}_{p}:=\lim
_{n\rightarrow\infty}x_{p}^{\left(  n\right)  }$ in $\mathbb{Q}_{p}$ and
$\widetilde{x}:=(\widetilde{x}_{p})_{p}$. We assert that $\widetilde{x}%
\in\mathbb{A}_{f}$. Indeed, $\rho(x^{(n)}, x^{(m)})<1$ for all $n,m\ge M_{0}$,
hence $\rho(x^{(n)}, x^{(m)})=\|x^{(n)}- x^{(m)}\|_{0}$. Due to the properties
of $p$-adic absolute value it follows from $\frac{|x_{p}-y_{p}|_{p}}{p}<1$
that $|x_{p}-y_{p}|_{p}\leq1$. Therefore $\bigl|x_{p}^{\left(  n_{0}\right)
}-x_{p}^{\left(  m\right)  }\bigr|_{p}\leq1$ for all $p$ and $n_{0}$, $m\geq
M_{0}$. Then $\bigl|x_{p}^{\left(  n_{0}\right)  }-\widetilde
x_{p}\bigr|_{p}\leq1$ for all $p$ and $n_{0}\ge M_{0}$. Since $\big(x_{p}%
^{(n_{0})}\big)_{p}\in\mathbb{A}_{f} $, there exist a constant $N$ such that
$x_{p}^{(n_{0})}\in\mathbb{Z}_{p}$ for $p\geq N$. Then also $\widetilde{x}%
_{p}\in\mathbb{Z}_{p}$ for $p\geq N$. To show that $\lim_{n\rightarrow\infty
}\rho(x^{(n)},\widetilde{x})=0$, consider arbitrary $\epsilon>0$ and take an
integer $N^{\prime}\geq N$ such that $1/N^{\prime}<\epsilon$. Since
$\bigl|x_{p}^{\left(  n\right)  }-\widetilde x_{p}\bigr|_{p}\leq1$ for all $p$
and $n\geq M_{0}$, and $x_{p}^{(n)}\rightarrow\widetilde{x}_{p}$ for any
$p$, we have for $n$ big enough
\begin{multline*}
\rho(x^{\left(  n\right)  },\widetilde{x})=\max\Bigl\{\max_{p<N^{\prime}}%
\frac{|x_{p}^{(n)}-\widetilde{x}_{p}|_{p}}{p},\max_{p\geq N^{\prime}}%
\frac{|x_{p}^{(n)}-\widetilde{x}_{p}|_{p}}{p}\Bigr\}\leq\\
\max\Bigl\{\max_{p<N^{\prime}}\frac{|x_{p}^{(n)}-\widetilde{x}_{p}|_{p}}%
{p},\frac{1}{N^{\prime}}\Bigr\}\leq\max\Bigl\{\max_{p<N^{\prime}}\frac
{|x_{p}^{(n)}-\widetilde{x}_{p}|_{p}}{p},\epsilon\Bigr\}=\epsilon.
\end{multline*}

Let $\tau_{\mathbb{A}_{f}}$ denote the restricted product topology on
$\mathbb{A}_{f}$, and let $\tau_{\rho}$ denote the topology\ induced by $\rho$
on $\mathbb{A}_{f}$. We want to show that $\tau_{\mathbb{A}_{f}}=\tau_{\rho}$.
Set $U:=\prod_{p}\mathbb{Z}_{p}$. Then the family
\[
x+yU,\quad x\in\mathbb{A}_{f}\text{, }y\in\mathbb{A}_{f}\setminus\{0\},
\]
is a base for $\tau_{\mathbb{A}_{f}}$. Note that $\left(  \mathbb{A}%
_{f},+,\cdot\right)  $ is a topological ring with respect to $\tau_{\rho}$,
and that the set $U$ coincides with $\left\{  x\in\mathbb{A}_{f}:\rho\left(
0,x\right)  \leq\frac{1}{2}\right\}  $ which is open in $\tau_{\rho}$ due to
the non-Archimedean nature of the metric. Then $x+yU\in\tau_{\rho}$ for any
$x\in\mathbb{A}_{f}$, $y\in\mathbb{A}_{f}\setminus\{0\}$, i.e. $\tau
_{\mathbb{A}_{f}}\subset\tau_{\rho}$. We now show that $\tau_{\rho}\subset
\tau_{\mathbb{A}_{f}}$. The family of balls
\begin{equation}
B_{\epsilon}\big(x^{\left(  0\right)  }\big)=\left\{  x\in\mathbb{A}_{f}%
:\rho\big(x^{\left(  0\right)  },x\big)\leq\epsilon\right\}  ,\quad x^{\left(
0\right)  }=\big(x_{p}^{\left(  0\right)  }\big)_{p}\in\mathbb{A}_{f}
\label{adelic_ball}%
\end{equation}
is a base for $\tau_{\rho}$. We have
\begin{equation}
B_{\epsilon}\big(x^{\left(  0\right)  }\big)=\prod_{p}\left(  x_{p}^{\left(
0\right)  }+p^{-\alpha_{p}\left(  \epsilon\right)  }\mathbb{Z}_{p}\right)  ,
\label{set}%
\end{equation}
where $\alpha_{p}\left(  \epsilon\right)  =[[\log_{p}\epsilon]]$, here the
function $[[\cdot]]$ is defined by \eqref{def[[]]}. Note that for $p$ big
enough $\alpha_{p}(\epsilon)=0$ and $x_{p}^{(0)}\in\mathbb{Z}_{p}$. Therefore
by (\ref{set}) we have $B_{\epsilon}(x^{\left(  0\right)  })\in\tau
_{\mathbb{A}_{f}}$.
\end{proof}

\begin{corollary}
\label{Cor1}$B_{\epsilon}\big(x^{\left(  0\right)  }\big)$ is a compact subset
for any $\epsilon>0$.
\end{corollary}

\begin{proof}
By (\ref{set}), $B_{\epsilon}(x^{\left(  0\right)  })$ is a translation of a
compact subset $%
{\textstyle\prod\nolimits_{p}}
p^{-\alpha_{p}\left(  \epsilon\right)  }\mathbb{Z}_{p}$, cf.
(\ref{compactsubset}).
\end{proof}

\begin{remark}
\label{nota2}The following properties of the space $\left(  \mathbb{A}%
_{f},\rho\right)  $ hold.

\begin{itemize}
\item[(i)] $\left(  \mathbb{A}_{f},\rho\right)  $ is $\sigma$-compact space.
Indeed, consider
\[
K_{N}:=\prod_{p\leq N}p^{-N}\mathbb{Z}_{p}\times\prod\limits_{p>N}%
\mathbb{Z}_{p}\quad\text{for }N\in\mathbb{N}.
\]
Then $K_{N}$ is a compact subgroup with respect to $\tau_{\mathbb{A}_{f}}$,
see e.g. \cite[Section 5.1]{R-V} and $\mathbb{A}_{f}=\cup_{N}K_{N}$.

\item[(ii)] $\left(  \mathbb{A}_{f},\rho\right)  $ is second-countable
topological space. Indeed, by applying twice the Weak Approximation Theorem,
see e.g. \cite[Theorem 1.4.4]{G-H}, one gets that $\beta+\alpha\prod
_{p}\mathbb{Z}_{p}$, $\beta\in\mathbb{Q}$, $\alpha\in\mathbb{Q\setminus
}\left\{  0\right\}  $ is a countable base for the topology of $\mathbb{A}%
_{f}$.

\item[(iii)] $\left(  \mathbb{A}_{f},\rho\right)  $ is a semi-compact space,
i.e. a locally compact Hausdorff space with a countable base.
\end{itemize}
\end{remark}

Metric $\rho$ allows us to introduce an adelic ball (given by
\eqref{adelic_ball}) and an adelic sphere, given by
\begin{equation}
S_{r}\bigl(x^{\left(  0\right)  }\bigr)=\left\{  x\in\mathbb{A}_{f}%
:\rho\bigl(x^{\left(  0\right)  },x\bigr)=r\right\}  ,\quad x^{\left(
0\right)  }=\bigl(x_{p}^{\left(  0\right)  }\bigr)_{p}\in\mathbb{A}_{f}.
\label{adelic_sphere}%
\end{equation}
Note that by Remark \ref{Rmk Metric Range} the radius $r$ of the adelic sphere
may possess only values equal to any non-zero integer power of prime number.
We now introduce some notations and compute volumes of adelic balls and
adelic spheres.

Given a positive real number $x$, we define
\begin{equation}
\Phi(x)=\prod_{p}p^{[[ \log_{p}x]] }, \label{defphi}%
\end{equation}
where $[[\cdot]]$ is defined by \eqref{def[[]]}, i.e. for $x\ge1$ we take a
product over all prime numbers each taken in the largest power $\alpha_{p}$
such that $p^{\alpha_{p}}\leq x$ and for $x<1$ we take a product over all
prime numbers each taken in the largest power $\alpha_{p}$ such that
$p^{\alpha_{p}}\leq px$, see also \eqref{set}. Note that only a finite number
of terms in this product differs from $1$ and that the function $\Phi(x)$ is
non-decreasing, right-continuous and piecewise constant. Then $\Phi(x)=1$ if
$1/2\leq x<2$. If $x\geq2$, $\Phi(x)$ coincides with the exponential of the
second Chebyshev function $\psi(x)=\sum_{p}[\log_{p}x]\ln p=\sum_{p^{k}\leq
x}\ln p$, where the last sum is taken over all powers of prime numbers not
exceeding $x$.

It is easy to check using \eqref{defphi}, \eqref{def[[]]} and properties of
the entire part function that for any prime number $p$ and any $j\in
\mathbb{Z}\setminus\{0\}$,
\begin{equation}
\label{Phi(1/x)}\Phi(p^{-j}) = \frac{p}{\Phi(p^{j})}.
\end{equation}

\begin{definition}
For $n\in\mathbb{R}$, $n>0$ we define the \textit{next} and \textit{previous
non-zero power of a prime operators} as
\begin{align}
n_{+}  &  =\min\left\{  p^{\beta}:n<p^{\beta},\ p\ \text{prime},\ \beta
\in\mathbb{Z}\setminus\{0\}\right\}  ,\label{nplusdef}\\
n_{-}  &  =\max\left\{  p^{\beta}:p^{\beta}<n,\ p\ \text{prime},\ \beta
\in\mathbb{Z}\setminus\{0\}\right\}  . \label{nminusdef}%
\end{align}

\end{definition}

It is easy to see that the following relations hold for any number $n=p^{j}$,
where $p$ is a prime and $j\in\mathbb{Z} \setminus\{0\}$
\begin{gather}
\begin{aligned} (n_{-})_{+} & =n, & \quad && (n_{+})^{-1} & = (n^{-1})_{-},\\ (n_{+})_{-} & =n, &&& (n_{-})^{-1} & = (n^{-1})_{+}, \end{aligned}\nonumber\\
\Phi\left(  (p^{j})_{-}\right)  =\frac{\Phi(p^{j})}{p}. \label{identity}%
\end{gather}

By using the operators $(\cdot)_{-}$ and $(\cdot)_{+}$ we can completely order
the set of non-zero powers of primes. This total order will be very relevant
in the next sections. To simplify notations we will write $p^{j}_{-}$
instead of $(p^{j})_{-}$ and $p^{j}_{+}$ instead of $(p^{j})_{+}$.

\begin{lemma}
\label{Lemma2A} (i) The the adelic ball $B_{r}:=B_{r}(0)$ is a compact subset
and its volume is given by
\[
\operatorname{vol}\left(  B_{r}\right)  =\Phi(r).
\]
(ii) The the adelic sphere $S_{r}:=S_{r}(0)$ is a compact subset and its
volume is given by
\[
\operatorname{vol}\left(  S_{r}\right)  =\Phi(r)-\Phi(r_{-}).
\]

\end{lemma}

\begin{proof}
The compactness of $B_{r}(0)$ was established in Corollary \ref{Cor1}. Since
$S_{r}(0)$ is a closed subset of $\mathbb{A}_{f}$ and $S_{r}(0)\subset
B_{r}(0)$ we conclude that $S_{r}(0)$ is compact. The formulas for volumes
follows immediately from (\ref{set}), (\ref{defphi}) and (\ref{def[[]]}).
\end{proof}

\subsection{The Fourier transform of radial functions}

\begin{definition}
A function $f:\mathbb{A}_{f}\rightarrow\mathbb{C}$ is said to be radial if its
restriction to any sphere $S_{r}$, $r>0$, is a constant function, i.e.
$\left.  f\right\vert _{S_{r}}=f_{r}\in\mathbb{C}$, $r>0$.
\end{definition}

By abuse of notation we will denote a radial function $f$ in the form
$f=f\left(  \left\Vert \xi\right\Vert \right)  $.

\begin{lemma}
\label{integral_radi_function}Let $f:\mathbb{A}_{f}\rightarrow\mathbb{C}$ be
an integrable function. Then the following assertions hold:

\noindent(i)
\[
\int_{\mathbb{A}_{f}} f\left(  \xi\right)  d\xi_{\mathbb{A}_{f}}=\sum
_{p^{m},\,m\in \mathbb{Z}\setminus\{0\}  }\int_{S_{p^{m}}}
f\left(  \xi\right)  d\xi_{\mathbb{A}_{f}} .
\]
In the particular case in which $f$ is a radial function this formula takes
the form
\[
\int_{\mathbb{A}_{f}} f\left(  \xi\right)  d\xi_{\mathbb{A}_{f}}=\sum
_{p^{m},\,m\in \mathbb{Z}\setminus\{0\} } f\left(  p^{m}\right)
\operatorname{vol}\left(  S_{p^{m}}\right)  .
\]

\noindent(ii) Take $A^{(i)}=\bigsqcup_{m\in J} S_{p^{m}}\subset\mathbb{A}_{f}%
$, where $J$ is a (countable) subset of $\mathbb{Z}\setminus\{0\}$, then
\[
\int_{\mathbb{A}_{f}} f\left(  \xi\right)  \boldsymbol{1}_{A^{(i)}}\left(
\xi\right)  d\xi_{\mathbb{A}_{f}}=\sum_{p^{m},\ m\in J }\int_{S_{p^{m}}}
f\left(  \xi\right)  d\xi_{\mathbb{A}_{f}} .
\]
In the particular case in which $f$ is a radial function this formula takes
the form
\[
\int_{\mathbb{A}_{f}} f\left(  \xi\right)  \boldsymbol{1}_{A^{(i)}}\left(
\xi\right)  d\xi_{\mathbb{A}_{f}}=\sum_{p^{m},\,m\in J} f\left(  p^{m}\right)
\operatorname{vol}\left(  S_{p^{m}}\right)  .
\]

\noindent(iii) Assume that $\mathbb{A}_{f}=\bigsqcup_{i\in\mathbb{N}}A^{(i)}$
with each $A^{(i)}$ is a disjoint union of spheres, then
\[
\int_{\mathbb{A}_{f}} f\left(  \xi\right)  d\xi_{\mathbb{A}_{f}}%
=\sum\limits_{i\in\mathbb{N}} \int_{A^{(i)}} f\left(  \xi\right)
d\xi_{\mathbb{A}_{f}}.
\]

\end{lemma}

\begin{proof}
The proof follows by general techniques in measure theory, the compactness of
the adelic balls and spheres, see Lemma \ref{Lemma2A}, and the
characterization of the adelic integrals for positive functions given in
\cite[p. 21]{G-H}.
\end{proof}

To simplify notations, throughout this subsection the expressions $\|0\|^{-1}$
and $|0|^{-1}_{p}$ in the inequalities mean $\infty$. The following theorem
describes the Fourier transform of a radial function.

\begin{theorem}
\label{LemmaFourierRadial} Let $f=f(\Vert\xi\Vert):\mathbb{A}_{f}%
\rightarrow\mathbb{C}$ be a radial function in $L^{1}(\mathbb{A}_{f})$. Then
the following formula holds:
\begin{equation}
\check{f}(x):=\bigl(\mathcal{F}_{\xi\rightarrow x}^{-1}f\bigr)(x)=\sum
_{q^{j}<\Vert x\Vert^{-1}}\Phi\left(  q^{j}\right)  \bigl(f(q^{j})-f(q_{+}%
^{j})\bigr)\quad\text{ for any }x\in\mathbb{A}_{f}, \label{FTradial}%
\end{equation}
where $q^{j}$ runs through all non-zero powers of prime numbers; the functions
$\Vert x\Vert$, $\Phi(x)$ and $q_{+}^{j}$ are defined by \eqref{Norm},
\eqref{defphi} and \eqref{nplusdef}.
\end{theorem}

\begin{remark}
It follows from \eqref{FTradial} that the Fourier transform of a radial
function is again a radial function.
\end{remark}

\begin{proof}
We represent the ring of finite adeles $\mathbb{A}_{f}$ as a disjoint union of
the following sets
\[
\mathbb{A}_{f}=\{0\}\sqcup\bigsqcup_{q}\mathbb{A}^{\left(  0,q\right)  }%
\sqcup\bigsqcup_{q}\mathbb{A}^{\left(  1,q\right)  },
\]
where
\begin{align*}
\mathbb{A}^{\left(  0,q\right)  }  &  :=\bigsqcup_{j <0} S_{q^{j}}
=\Bigl\{\xi\in\mathbb{A}_{f}:0<\Vert\xi\Vert<1,\ \left\Vert \xi\right\Vert
=\frac{|\xi_{q}|_{q}}{q}\text{ and }\frac{|\xi_{p}|_{p}}{p}<\frac{|\xi
_{q}|_{q}}{q}\text{ for }p\neq q\Bigr\},\\
\mathbb{A}^{\left(  1,q\right)  }  &  :=\bigsqcup_{j >0} S_{q^{j}}
=\bigl\{\xi\in\mathbb{A}_{f}:\Vert\xi\Vert>1,\ \left\Vert \xi\right\Vert
=|\xi_{q}|_{q}\text{ and }|\xi_{p}|_{p}<|\xi_{q}|_{q}\text{ for }p\neq
q\bigr\}.
\end{align*}
Note that on the sets $A^{(0,q)}$ we have $\Vert\xi\Vert=\Vert\xi\Vert_{0}$
and on the sets $A^{(1,q)}$ we have $\Vert\xi\Vert=\Vert\xi\Vert_{1}$. Then
$\check{f}\left(  x\right)  =\sum_{q}\check{f}^{\left(  0,q\right)  }\left(
x\right)  +\sum_{q}\check{f}^{\left(  1,q\right)  }\left(  x\right)  $, where
\[
\check{f}^{\left(  k,q\right)  }\left(  x\right)  :=\int_{\mathbb{A}^{\left(
k,q\right)  }}\chi\left(  \xi\cdot x\right)  f(\left\Vert \xi\right\Vert
)\,d\xi_{\mathbb{A}_{f}},\qquad k=0,1,\ q\ \text{is a prime}.
\]

We set
\begin{equation}\label{eq_beta_q}
\beta_{q}:=\beta_{q}\left(  x\right)  =-[\log_{q}\Vert x\Vert]
\end{equation}
with
convention that $\beta_{q}\left(  0\right)  =+\infty$. We also set
$\delta\left(  t\right)  =1$ if $t=0$ and $\delta\left(  t\right)  =0$ otherwise.

To simplify the proof, we first present the final formulas for the
functions $\check{f}^{\left(  k,q\right)  }\left(  x\right)  $, the proofs are
given later.

\begin{claim}
\label{Claim1}%
\begin{equation}
\sum_{q}\check{f}^{\left(  1,q\right)  }\left(  x\right)  =0\qquad
\text{if}\ \Vert x\Vert>1, \tag{A}\label{form1}
\end{equation}
\begin{equation}
\begin{split}
\sum_{q}\check{f}^{\left(  1,q\right)  }\left(  x\right)  =  &  \sum_{q<\Vert
x\Vert^{-1}}\biggl\{\left(  1-\frac{1}{q}\right)  \sum_{j=1}^{\beta_{q}\left(
x\right)  -1}f(q^{j})\Phi\left(  q^{j}\right)  \biggr\}\\
&  -\sum_{q}\frac{1}{q}f\bigl(\Vert x\Vert^{-1}\bigr)\Phi\left(  \Vert
x\Vert^{-1}\right)  \delta\Bigl(\frac{|x_{q}|_{q}}{q}-\Vert x\Vert
\Bigr)\quad\text{if}\ \Vert x\Vert<1.
\end{split}
\tag{B}\label{form2}%
\end{equation}
\end{claim}

\begin{claim}
\label{Claim2}%
\begin{equation}
\sum_{q}\check{f}^{\left(  0,q\right)  }\left(  x\right)  =\sum_{q}
\biggl\{\left(  1-\frac{1}{q}\right)  \sum_{j=-\infty}^{-1}f(q^{j})\Phi\left(
q^{j}\right)  \biggr\}\qquad\text{if}\ \Vert x\Vert<1, \tag{C}\label{form3}%
\end{equation}
\begin{equation}%
\begin{split}
\sum_{q}\check{f}^{\left(  0,q\right)  }\left(  x\right)  =  &  \sum
_{q}\biggl\{\left(  1-\frac{1}{q}\right)  \sum_{j=-\infty}^{\beta_{q}\left(
x\right)  -1}f(q^{j})\Phi\left(  q^{j}\right)  \biggr\}\\
&  -\sum_{q}\frac{1}{q}f\bigl(\Vert x\Vert^{-1}\bigr)\Phi\left(  \Vert
x\Vert^{-1}\right)  \delta\bigl(|x_{q}|_{q}-\Vert x\Vert\bigr)\quad
\text{if}\ \Vert x\Vert>1.
\end{split}
\tag{D}\label{form4A}%
\end{equation}
\end{claim}

Combining (\ref{form1}), (\ref{form2}), (\ref{form3}), (\ref{form4A}) we
obtain
\begin{equation}%
\begin{split}
\check{f}(x)=  &  \sum_{q}\biggl\{\left(  1-\frac{1}{q}\right)  \sum
_{\substack{j\leq\beta_{q}(x)-1,\\j\neq0}}f(q^{j})\Phi\left(  q^{j}\right)
\biggr\}\\
&  -\sum_{q,\,j}\frac{1}{q}f\bigl(\Vert x\Vert^{-1}\bigr)\Phi\left(  \Vert
x\Vert^{-1}\right)  \delta\bigl(q^{j}-\Vert x\Vert\bigr).
\end{split}
\label{form5}%
\end{equation}
Note that the last sum over $q$ and $j$ involving the function $\delta$ means
that we take the only term corresponding to the prime number $q$ such that
$\Vert x\Vert=q^{j}$ for some $j\in\mathbb{Z}\setminus\{0\}$.

Now the proof of the theorem may be finished as follows. Since $f\in
L^{1}(\mathbb{A}_{f})$ and $\operatorname{vol}(S_{p^{j}}) = \operatorname{vol}%
\bigl(  \{\xi\in\mathbb{A}_{f}:\Vert\xi\Vert=p^{j}\}\bigr)
=\Phi(p^{j})-\Phi(p_{-}^{j})$, see Lemma \ref{Lemma2A}, the series
$\sum_{q^{j}}\bigl(\Phi(q^{j})-\Phi(q_{-}^{j})\bigr)\left\vert f(q^{j}%
)\right\vert $ is convergent. Because of the inequality $\Phi(q^{j}%
)-\Phi(q_{-}^{j})\geq\frac{1}{2}\Phi(q^{j})$ the series $\sum_{q^{j}}%
\Phi\left(  q^{j}\right)  \left\vert f(q^{j})\right\vert $ converges as well,
hence we may arbitrary reorder the terms in (\ref{form5}).

By the properties of the entire part function, the following inequalities hold for the function $\beta_q$ \eqref{eq_beta_q}:
\begin{align}
& q^{j}\leq q^{-[\log_{q}\Vert x\Vert
]-1}<q^{-\log_{q}\Vert x\Vert}=\Vert x\Vert^{-1}, && j<\beta_q,\ j\in\mathbb{Z},\label{eq_beta_q1}\\
& q^{j}\geq q^{-[\log_{q}\Vert x\Vert]}\geq q^{-\log_{q}\Vert x\Vert
}=\Vert x\Vert^{-1}, && j\ge \beta_q,\ j\in\mathbb{Z},\label{eq_beta_q2}
\end{align}
where the equality in the second inequality is possible only when $\Vert x\Vert$ is a power of $q$. Suppose in \eqref{form5}, that $\|x\| = p^k$ for some prime number $p$ and integer $k\ne 0$.
It follows from inequalities \eqref{eq_beta_q1}, \eqref{eq_beta_q2} that the formula \eqref{form5} may be written as
\begin{multline*}
\check{f}(x)=  \sideset{}{'}\sum_{q^j < \|x\|^{-1}}\left(1-\frac{1}{q}\right)  f(q^{j})\Phi(  q^{j}) -\frac{1}{p}f\bigl(\Vert x\Vert^{-1}\bigr)\Phi\left(  \Vert
x\Vert^{-1}\right) \\
=\sideset{}{'}\sum_{q^j < \|x\|^{-1}}  f(q^{j})\Phi(  q^{j}) - \sideset{}{'}\sum_{q^j < \|x\|^{-1}}  f(q^{j})\Phi(  q^{j}_-) - f\bigl(\Vert x\Vert^{-1}\bigr)\Phi\left(  \Vert
x\Vert^{-1}_-\right)\\
=\sideset{}{'}\sum_{q^j < \|x\|^{-1}}  f(q^{j})\Phi(  q^{j}) -
\sideset{}{'}\sum_{q^j \le \|x\|^{-1}}  f(q^{j})\Phi(  q^{j}_-)=\sideset{}{'}\sum
_{q^{j}<\Vert x\Vert^{-1}}\Phi\left(  q^{j}\right)  \bigl(f(q^{j})-f(q_{+}
^{j})\bigr),
\end{multline*}
where we have used
(\ref{nplusdef})--(\ref{identity}) and $\sum'$ means that the value $j=0$ is omitted in the summation.

The checking of the formula \eqref{FTradial} in the case $\|x\|=0$ is left to the reader.
\end{proof}

\begin{proof}[Proof of Claim \ref{Claim1}] We assume that $x\neq0$. The case $x=0$ may be checked directly. With the use of \eqref{defFTfactorisable}, Lemma \ref{integral_radi_function} and the fact that $\mathbf{1}_{S_r}(x)$ is a factorizable function, we may write $\check{f}^{\left(  1,q\right)  }$ as
\[
\check{f}^{\left(  1,q\right)  }\left(  x\right)  =\int\limits_{|\xi_{q}%
|_{q}\geq q}\chi_{q}\left(  x_{q}\xi_{q}\right)  f(|\xi_{q}|_{q}%
)\bigg\{\prod_{p\neq q}\ \int\limits_{|\xi_{p}|_{p}<|\xi_{q}|_{q}}\chi
_{p}\left(  x_{p}\xi_{p}\right)  d\xi_{p}\bigg\}d\xi_{q}.
\]
Denote by $\alpha_{p}\left(  \xi_{q}\right)  $ the largest integer satisfying
$p^{\alpha_{p}\left(  \xi_{q}\right)  }\leq|\xi_{q}|_{q}$ (i.e. $\alpha
_{p}\left(  \xi_{q}\right)  =[\log_{p}|\xi_{q}|_{q}]=[[\log_{p}|\xi_{q}%
|_{q}]]$). Note that the equality $p^{\alpha_{p}\left(  \xi_{q}\right)  }%
=|\xi_{q}|_{q}$ is impossible for $|\xi_{q}|_{q}>1$ and $p\neq q$, hence
$p^{\alpha_{p}\left(  \xi_{q}\right)  }<|\xi_{q}|_{q}$. Recall that
\[
\int_{|\xi_{p}|_{p}\leq p^{\alpha_{p}\left(  \xi_{q}\right)  }}\chi_{p}\left(
x_{p}\xi_{p}\right)  d\xi_{p}=%
\begin{cases}
p^{\alpha_{p}\left(  \xi_{q}\right)  } & \text{if}\ |x_{p}|_{p}\leq
p^{-\alpha_{p}\left(  \xi_{q}\right)  },\\
0 & \text{if}\ |x_{p}|_{p}\geq p^{-\alpha_{p}\left(  \xi_{q}\right)  +1},
\end{cases}
\]
and since $p^{\alpha_{p}(\xi_{q})}<|\xi_{q}|_{q}<p^{\alpha_{p}(\xi_{q})+1}$,
we have
\[
\int_{|\xi_{p}|_{p}\leq p^{\alpha_{p}\left(  \xi_{q}\right)  }}\chi_{p}\left(
x_{p}\xi_{p}\right)  d\xi_{p}=%
\begin{cases}
p^{\alpha_{p}\left(  \xi_{q}\right)  } & \text{if}\ |x_{p}|_{p}< p|\xi
_{q}|_{q}^{-1},\\
0 & \text{if}\ |x_{p}|_{p}>p|\xi_{q}|_{q}^{-1},
\end{cases}
\]
which implies
\begin{equation}
\prod_{p\neq q}\int_{|\xi_{p}|_{p}<|\xi_{q}|_{q}}\chi_{p}\left(  x_{p}\xi
_{p}\right)  d\xi_{p}=\bigg(\prod_{p\neq q}p^{\alpha_{p}\left(  \xi
_{q}\right)  }\bigg)\mathbf{1}_{B}\left(  \xi_{q}\right)  =\frac{\Phi(|\xi
_{q}|_{q})}{|\xi_{q}|_{q}}\mathbf{1}_{B}\left(  \xi_{q}\right)  ,
\label{prod_pneq}%
\end{equation}
where $\mathbf{1}_{B}\left(  \xi_{q}\right)  $ is the characteristic function
of the set
\[
B:=\biggl\{\xi_{q}\in\mathbb{Q}_{q}:\max_{p\neq q}\frac{|x_{p}|_{p}}{p}
<|\xi_{q}|_{q}^{-1}\biggr\}.
\]
Therefore
\begin{equation}
\check{f}^{\left(  1,q\right)  }\left(  x\right)  =\int_{q\leq|\xi_{q}%
|_{q}<\big(\max_{p\neq q}\frac{|x_{p}|_{p}}{p}\big)^{-1}}\chi_{q}\left(
x_{q}\xi_{q}\right)  f(|\xi_{q}|_{q})\frac{\Phi(|\xi_{q}|_{q})}{|\xi_{q}|_{q}%
}d\xi_{q}. \label{eq6}%
\end{equation}
Note that it follows from \eqref{eq6} that $\check{f}^{\left(  1,q\right)
}\left(  x\right)  =0$ if $\max_{p\neq q}\frac{|x_{p}|_{p}}{p}\ge\frac{1}{q}$.

Set $\gamma_{q}$ to be the largest integer satisfying $q^{\gamma_{q}}<\left(  \max_{p\neq q}\frac{|x_{p}|_{p}}{p}\right)  ^{-1}$, then
\begin{equation}
\check{f}^{\left(  1,q\right)  }\left(  x\right)  =\sum_{j=1}^{\gamma_{q}%
}\frac{f(q^{j})\Phi\left(  q^{j}\right)  }{q^{j}}\int_{|\xi_{q}|_{q}=q^{j}%
}\ \chi_{q}\left(  x_{q}\xi_{q}\right)  d\xi_{q}. \label{eq7}%
\end{equation}
We recall that%
\begin{equation}
\int_{|\xi_{q}|_{q}=q^{j}}\ \chi_{q}\left(  x_{q}\xi_{q}\right)  d\xi_{q}=%
\begin{cases}
q^{j}\left(  1-q^{-1}\right)  & \text{if }|x_{q}|_{q}\leq q^{-j},\\
-q^{j-1} & \text{if }|x_{q}|_{q}=q^{-j+1},\\
0 & \text{if }|x_{q}|_{q}\geq q^{-j+2}.
\end{cases}
\label{eq8}%
\end{equation}
Note that the integral \eqref{eq8} is non-zero when $|\xi_{q}|_{q}\leq\frac
{q}{|x_{q}|_{q}}=\Bigl(\frac{|x_{q}|_{q}}{q}\Bigr)^{-1}$. Since for $|\xi_q|_q>1$ 
the equality $|\xi_q|_q = \Big(\max_{p\neq q}\frac{|x_{p}|_{p}
}{p}\Big)^{-1}$ is impossible, the last inequality
may be combined with $|\xi_{q}|_{q}<\Big(\max_{p\neq q}\frac{|x_{p}|_{p}
}{p}\Big)^{-1}$ into the inequality $|\xi_{q}|_{q}\leq\Big(\max_{p}
\frac{|x_{p}|_{p}}{p}\Big)^{-1}=\Vert x\Vert_{0}^{-1}$, cf. \eqref{Norm0}.
Then it follows from (\ref{eq6})--(\ref{eq8}) that
\begin{equation}
\check{f}^{\left(  1,q\right)  }\left(  x\right)  =\int_{q\leq|\xi_{q}
|_{q}\leq\Vert x\Vert_{0}^{-1}}\frac{\chi_{q}\left(  x_{q}\xi_{q}\right)
f(|\xi_{q}|_{q})\Phi(|\xi_{q}|_{q})}{|\xi_{q}|_{q}}d\xi_{q}. \label{eq9}
\end{equation}
Note that $\check{f}^{\left(  1,q\right)  }\left(  x\right)  =0$ if $\Vert
x\Vert_{0}^{-1}<q$ and that
\[
\left\{  x\in\mathbb{A}_{f}:\Vert x\Vert_{0}^{-1}\geq2\right\}  =\left\{
x\in\mathbb{A}_{f}:\max_{p}\frac{|x_{p}|_{p}}{p}\leq\frac{1} {2}\right\}
=\prod_{p}\mathbb{Z}_{p},
\]
hence for any $x$ outside of $\prod_{p}\mathbb{Z}_{p}$ the sum $\sum_{q}
\check{f}^{\left(  1,q\right)  }\left(  x\right)  $ vanishes. Therefore we
have non-zero terms in $\sum_{q}\check{f}^{\left(  1,q\right)  }\left(
x\right)  $ only if $\Vert x\Vert_{0}<1$. In such case $\Vert x\Vert_{0}=\Vert
x\Vert$. We recall definition \eqref{eq_beta_q}
of $\beta_{q}$  and inequalities \eqref{eq_beta_q1}, \eqref{eq_beta_q2}.
Then it follows from \eqref{eq7}--\eqref{eq9} that
\begin{equation}
\check{f}^{\left(  1,q\right)  }\left(  x\right)  =\left(  1-\frac{1}%
{q}\right)  \sum_{j=1}^{\beta_{q}-1}f(q^{j})\Phi(q^{j})\qquad\text{if}%
\ \frac{|x_{q}|_{q}}{q}<\Vert x\Vert\label{f1qpq}%
\end{equation}
and
\begin{equation}
\check{f}^{\left(  1,q\right)  }\left(  x\right)  =\left(  1-\frac{1}%
{q}\right)  \sum_{j=1}^{\beta_{q}-1}f(q^{j})\Phi(q^{j})-\frac{1}%
{q}f\bigl(\Vert x\Vert^{-1}\bigr)\Phi\left(  \Vert x\Vert^{-1}\right)
\quad\text{if}\ \frac{|x_{q}|_{q}}{q}=\Vert x\Vert. \label{f1qq}%
\end{equation}
By combining the formulas (\ref{f1qpq})--(\ref{f1qq}) we obtain formula
(\ref{form2}).
\end{proof}

\begin{proof}
[Proof of Claim \ref{Claim2}]We assume that $x\neq0$. The case $x=0$  may be checked directly. 
 The required calculations are mostly similar to the
previous ones, however there are some subtle variations. We have $\Vert
\xi\Vert=q^{-1}|\xi_{q}|_{q}$ and
\[
\check{f}^{\left(  0,q\right)  }\left(  x\right)  =\int\limits_{q^{-1}|\xi
_{q}|_{q}<1}\chi_{q}\left(  x_{q}\xi_{q}\right)  f(q^{-1}|\xi_{q}%
|_{q})\bigg\{\prod_{p\neq q}\ \int\limits_{p^{-1}|\xi_{p}|_{p}<q^{-1}|\xi
_{q}|_{q}}\chi_{p}\left(  x_{p}\xi_{p}\right)  d\xi_{p}\bigg\}d\xi_{q}.
\]
Let $\alpha_{p}(\xi_{q})$ denote the largest power $p^{\alpha_{p}(\xi_{q})}$
satisfying $p^{-1}p^{\alpha_{p}(\xi_{q})}<q^{-1}|\xi_{q}|_{q}$ which is equal
to $1+\bigl[\log_{p}q^{-1}|\xi_{q}|_{q}\bigr]$, and since $q^{-1}|\xi_{q}%
|_{q}<1$, the last quantity is equal to $\bigl[\bigl[\log_{p}q^{-1}|\xi
_{q}|_{q}\bigr]\bigr]$, cf. \eqref{def[[]]}. Hence similarly to
\eqref{prod_pneq} with the use of \eqref{defphi} we obtain
\[
\prod_{p\neq q}\text{ }\int\limits_{p^{-1}|\xi_{p}|_{p}<q^{-1}|\xi_{q}|_{q}%
}\chi_{p}\left(  x_{p}\xi_{p}\right)  d\xi_{p}=\bigg(\prod_{p\neq q}%
p^{\alpha_{p}\left(  \xi_{q}\right)  }\bigg)\mathbf{1}_{B}\left(  \xi
_{q}\right)  =\frac{\Phi(q^{-1}|\xi_{q}|_{q})}{q^{[[\log_{q}q^{-1}|\xi
_{q}|_{q}]]}}\mathbf{1}_{B}\left(  \xi_{q}\right)  ,
\]
where $\mathbf{1}_{B}\left(  \xi_{q}\right)  $ is the characteristic function
of the set
\begin{equation}
B:=\Bigl\{\xi_{q}\in\mathbb{Q}_{q}:|\xi_{q}|_{q}< q\Bigl(\max_{p\neq
q}{|x_{p}|_{p}}\Bigr)^{-1}\Bigr\}. \label{Bf0}%
\end{equation}
Since $|\xi_{q}|_{q}$ is a power of $q$ and $q^{-1}|\xi_{q}|_{q}<1$, we have
$[[\log_{q}q^{-1}|\xi_{q}|_{q}]]=\log_{q}|\xi_{q}|_{q}$ and $q^{[[\log
_{q}q^{-1}|\xi_{q}|_{q}]]}=q^{\log_{q}|\xi_{q}|_{q}}=|\xi_{q}|_{q}$. Since
$q^{-1}|\xi_{q}|_{q}<1$ is equivalent to $|\xi_{q}|_{q}\leq1$, we obtain
\begin{equation}
\check{f}^{\left(  0,q\right)  }\left(  x\right)  =\int\limits_{\substack{|\xi
_{q}|_{q}\leq1,\\|\xi_{q}|_{q}< q(\max_{p\neq q}{|x_{p}|_{p}})^{-1}}%
}\chi_{q}\left(  x_{q}\xi_{q}\right)  f(q^{-1}|\xi_{q}|_{q})\frac{\Phi
(q^{-1}|\xi_{q}|_{q})}{|\xi_{q}|_{q}}d\xi_{q}. \label{int_for_f0q}%
\end{equation}
It follows from \eqref{eq8} that the last integral is non-zero when
$|x_{q}|_{q}\leq q|\xi_{q}|_{q}^{-1}$, which may be combined with \eqref{Bf0}
into $\max_{p}|x_{p}|_{p}\leq q|\xi_{q}|_{q}^{-1}$ or, equivalently, into
$|\xi_{q}|_{q}\leq q(\max_{p}{|x_{p}|_{p}})^{-1}$. Hence the domain of
integration in the last integral is
\begin{equation}
|\xi_{q}|_{q}\leq\min\Bigl\{1,\frac{q}{\max_{p}{|x_{p}|_{p}}}\Bigr\},
\label{min_for_xi}%
\end{equation}
and we have to consider two cases.

\noindent\textbf{Case 1}: $\Vert x\Vert<1$, i.e. $\max_{p}|x_{p}|_{p}\leq1$.
In this case the minimum in \eqref{min_for_xi} is equal to 1 and the equality
$|x_{q}|_{q}<q|\xi_{q}|_{q}^{-1}$ always holds, hence from
\eqref{int_for_f0q}, \eqref{min_for_xi} and \eqref{eq8} similarly to
\eqref{f1qpq} we obtain
\[
\check{f}^{\left(  0,q\right)  }\left(  x\right)  =\left(  1-\frac{1}%
{q}\right)  \sum_{j=-\infty}^{0}f(q^{j-1})\Phi(q^{j-1})=\left(  1-\frac{1}%
{q}\right)  \sum_{j=-\infty}^{-1}f(q^{j})\Phi(q^{j}).
\]

\noindent\textbf{Case 2}: $\Vert x\Vert>1$, i.e. $\max_{p}|x_{p}|_{p}=\Vert
x\Vert>1$. In this case it is sufficient to determine possible values of
$|\xi_{q}|_{q}$ from the inequality $|\xi_{q}|_{q}\leq\frac{q}{\Vert x\Vert}$,
because they satisfy the inequality $|\xi_{q}|_{q}\leq1$ even in the case $\Vert
x\Vert<q$. Recall definition \eqref{eq_beta_q}
of $\beta_{q}$  and inequalities \eqref{eq_beta_q1}, \eqref{eq_beta_q2}. As a result, from \eqref{int_for_f0q}, \eqref{min_for_xi} and \eqref{eq8}
similarly to \eqref{f1qpq} and \eqref{f1qq} we have
\begin{align*}
\check{f}^{\left(  0,q\right)  }\left(  x\right)   &  =\left(  1-\frac{1}%
{q}\right)  \sum_{j=-\infty}^{\beta_{q}}f(q^{j-1})\Phi(q^{j-1})\\
&  =\left(  1-\frac{1}{q}\right)  \sum_{j=-\infty}^{\beta_{q}-1}f(q^{j}%
)\Phi(q^{j})\qquad\text{if}\ |x_{q}|_{q}<\Vert x\Vert
\end{align*}
and
\[
\check{f}^{\left(  0,q\right)  }\left(  x\right)  =\left(  1-\frac{1}%
{q}\right)  \sum_{j=-\infty}^{\beta_{q}-1}f(q^{j})\Phi(q^{j})-\frac{1}%
{q}f\bigl(\Vert x\Vert^{-1}\bigr)\Phi\left(  \Vert x\Vert^{-1}\right)
\quad\text{if}\ |x_{q}|_{q}=\Vert x\Vert.
\]
The last two equalities give us formulas (\ref{form3})--(\ref{form4A}).
\end{proof}

As a corollary from Theorem \ref{LemmaFourierRadial} we derive a sufficient condition for the Fourier transform of a radial function to be non-negative.
\begin{corollary}\label{FTNonNegative}
Let $f$ be a real-valued non-increasing radial function, i.e. $f=f(\|\xi\|)$ and $f(\xi) \ge f(\zeta)$ for any $\xi,\zeta\in\mathbb{A}_f$ satisfying $\|\xi\|\le \|\zeta\|$. Then
\[
\check{f}(x):=\bigl(\mathcal{F}_{\xi\rightarrow x}^{-1}f\bigr)(x) \ge 0 \qquad \text{for any}\ x\in\mathbb{A}_f.
\]
\end{corollary}

On the base of Theorem \ref{LemmaFourierRadial} we may compute the Fourier
transforms of characteristic functions of balls and spheres.

\begin{corollary}
\label{FTBall} Let $f$ be a characteristic function of a ball, i.e. $f =
\mathbf{1}_{B_{r}}(x)$, $r\in\{p^{j}: p\ \text{is prime},\ j\in\mathbb{Z}%
\setminus\{0\}\}$. Then
\[
\widehat f(\xi) = \Phi(r) \mathbf{1}_{B_{R}}(\xi),\quad R = (r^{-1})_{-}.
\]
Let $g$ be a characteristic function of a sphere, i.e. $g = \mathbf{1}_{S_{r}%
}(x)$, $r\in\{p^{j}: p\ \text{is prime},\ j\in\mathbb{Z}\setminus\{0\}\}$.
Then
\[
\widehat g(\xi) = \Phi(r) \mathbf{1}_{B_{R}}(\xi) - \Phi(r_{-})\mathbf{1}%
_{B_{R_{+}}}(\xi),\quad R = (r^{-1})_{-}.
\]

\end{corollary}

\begin{proof}
As it follows from \eqref{FTradial} we have at most one non-zero term equal to
$\Phi(r)$ in the Fourier transform of the ball $B_{r}$, and this term is
present if and only if $r<\|x\|^{-1}$ which is equivalent to $\|x\|<r^{-1}$ or
$\|x\|\le(r^{-1})_{-}$.

The second statement follows from the presentation $S_{r} = B_{r} \setminus
B_{r_{-}}$ and the properties of the operators `$_{-}$' and `$_{+}$':
$\bigl((r_{-})^{-1}\bigr)_{-} = \bigl( (r^{-1})_{+}\bigr)_{-} = r^{-1} =
R_{+}$.
\end{proof}

\subsection{Distributions on $\mathbb{A}_{f}$}

In this subsection we consider $\mathbb{A}_{f}$ as the complete metric space
$(\mathbb{A}_{f},\rho)$. As it was previously mentioned in Subsection
\ref{SubSectFourier}, the space $\mathcal{S}(\mathbb{A}_{f})$ of
Bruhat-Schwartz functions consists of finite linear combinations of
factorizable functions $f=\prod_{p}f_{p}$, where a finite number of the
functions $f_{p}$ are in $\mathcal{S}(\mathbb{Q}_{p})$ and the rest of the
functions are the characteristic functions of the sets $\mathbb{Z}_{p}$, i.e.
$f=\prod_{p\leq N}f_{p}\times\prod_{p>N}\Omega_{p}(|x_{p}|_{p})$. For the sake
of simplicity, from now on we will use \textit{test function} to mean
\emph{Bruhat-Schwartz function}. The spaces $\mathcal{S} (\mathbb{Q}_{p})$ consist of
compactly supported locally constant functions. We show that the same property
characterizes the space $\mathcal{S} (\mathbb{A}_{f})$. Despite a similar
result was already proved in \cite{Kh-Ra}, the adelic metric $\rho$ allows us
to introduce the notion of `parameter of constancy' for functions in
$\mathcal{S}(\mathbb{A}_{f})$ and to give a construction of a topology for
$\mathcal{S}(\mathbb{A}_{f})$ in a similar way as for the spaces
$\mathcal{S}(\mathbb{Q}_{p})$.

\begin{definition}
We say that a function $f$ is locally constant if for any $x\in\mathbb{A}_{f}$
there exists a constant $\ell(x)>0$ such that $f(x+y)=f(x)$ for any $y\in
B_{\ell(x)}(0)$.
\end{definition}

The same reasoning as in the $p$-adic case, see e.g. \cite{V-V-Z} or
\cite{A-K-S}, shows that for a compactly supported function $f$ the same
constant $\ell$ may be chosen for all points $x\in\mathbb{A}_{f}$.

\begin{definition}
Let $f$ be a non-zero compactly supported function. We define the parameter of
constancy $\ell$ of $f$ as the largest non-zero integer power of a prime
number such that
\begin{equation}
f(x+y)=f(x)\qquad\text{\ for any }x\in\mathbb{A}_{f},\ y\in B_{\ell}(0).
\label{Local Constancy}%
\end{equation}
By definition we set the parameter of constancy of function $0$ to be equal
$+\infty$.
\end{definition}

\begin{lemma}
\label{Lemma3} The function $f\in\mathcal{S}(\mathbb{A}_{f})$ if and only if
it is locally constant with compact support.
\end{lemma}

\begin{proof}
The statement is trivial for $f\equiv0$. Suppose $f\in\mathcal{S}%
(\mathbb{A}_{f})\setminus\left\{  0\right\}  $, and $f=\sum_{m=1}^{M}%
f^{(m)}(x)$, where each function $f^{(m)}$ is factorizable, $f^{(m)}%
=\prod_{p\leq N_{m}}f_{p}^{(m)}\times\prod_{p>N_{m}}\Omega_{p}(|x_{p}|_{p})$.
Since each $f^{(m)}$ is compactly supported, so is $f$.

Let $l_{p}^{(m)}$ denote the parameter of constancy of the function
$f_{p}^{(m)}$, i.e. $f_{p}^{(m)}(x_{p}+y_{p})=f_{p}^{(m)}(x_{p})$ for all
$y_{p}$ such that $|y_{p}|_{p}\leq l_{p}^{(m)}$. Note that our definition of
the parameter of constancy on $\mathbb{Q}_{p}$ is different from the one
presented in \cite{V-V-Z}. Such change of definition is justified by necessity
to make the parameter of constancy independent on $p$. Consider
\[
\ell=\min\Big\{\frac{1}{2},\min_{p,m}\frac{l_{p}^{(m)}}{p}\Big\}.
\]
Since we have only finite number of parameters $l_{p}^{(m)}$, $\ell>0$. It is
easy to check that \eqref{Local Constancy} holds with this parameter $\ell$,
i.e. the function $f$ is locally constant.

Suppose now that $f$ is a locally constant function with compact support. Let
$K=\operatorname{supp}f$. Since $f$ is locally constant, for each $x\in K$
there exists a ball $B_{r(x)}(x)$, which is an open set, such that $f$ is
constant on $B_{r(x)}(x)$. Then there exists a finite number of these balls,
say $B_{r_{1}}(x_{1}),\ldots,B_{r_{n}}(x_{n})$, covering $K$. Since the metric
is non-Archimedean we may assume that these balls are disjoint. Therefore
\begin{equation}
\label{representFS}f(x)=f(x_{1})\cdot\mathbf{1}_{B_{r_{1}}(x_{1})}%
(x)+\ldots+f(x_{n})\cdot\mathbf{1}_{B_{r_{n}}(x_{n})}(x),
\end{equation}
where each characteristic function $\mathbf{1}_{B_{r}}(x)$ is factorizable,
cf. \eqref{set}.
\end{proof}

\begin{remark}
Let $\mathcal{P}(\mathbb{A}_{f})$ denote the set of parameters of constancy of
functions from $\mathcal{S} (\mathbb{A}_{f})$. Then
\[
\mathcal{P}(\mathbb{A}_{f})=\left\{  l\in\mathbb{Q}:l=p^{m},\ p\text{ is a
prime},\ m\in\mathbb{Z}\setminus\{0\}\right\}  \cup\{+\infty\}.
\]
By considering the characteristic functions of the adelic balls $B_{r}$ we
verify that every number in $\mathcal{P}$ is an admissible parameter of
constancy. $\mathcal{P}(\mathbb{A}_{f})$ is a countable and totally ordered set.
\end{remark}

We define by $\mathcal{S}_{R}^{l}(\mathbb{A}_{f})$ the subspace of test
functions with supports contained in the adelic ball $B_{R}$ and parameters of
constancy $\geq l$. Then the following embedding holds: $\mathcal{S}_{R}%
^{l}(\mathbb{A}_{f})\subset\mathcal{S}_{R^{\prime}}^{l^{\prime}}%
(\mathbb{A}_{f})$ whenever $R\leq R^{\prime}$, $l\geq l^{\prime}$. As in the
$p$-adic setting, see e.g. \cite{A-K-S}, \cite{Taibleson}, \cite{V-V-Z}, we
define the convergence in $\mathcal{S}(\mathbb{A}_{f})$ in the following way:
$f_{k}\rightarrow0$, $k\rightarrow\infty$ in $\mathcal{S}(\mathbb{A}_{f})$ if
and only if

\begin{itemize}
\item[(i)] $f_{k}\in\mathcal{S}_{R}^{l}(\mathbb{A}_{f})$ where $R$ and $l$ do
not depend on $k$;

\item[(ii)] $f_{k}\to0$ uniformly as $k\to\infty$.
\end{itemize}

With this notion of convergence $\mathcal{S}(\mathbb{A}_{f})$ becomes a
complete topological vector space. In addition,
\[
\mathcal{S}_{R}(\mathbb{A}_{f})=
\mathop{\operatorname{lim\, ind}}\limits_{l\rightarrow0}\mathcal{S}_{R }%
^{l}(\mathbb{A}_{f}),\qquad\mathcal{S}(\mathbb{A}_{f}%
)=\mathop{\operatorname{lim\, ind}}\limits_{R\rightarrow\infty}\mathcal{S}
_{R}(\mathbb{A}_{f}).
\]
Note that the second inductive limit makes sense because $\mathcal{P}%
(\mathbb{A}_{f})$\ is totally ordered.

The following proposition shows that the spaces $\mathcal{S}_{R}%
^{l}(\mathbb{A}_{f})$ possess similar properties to their $p$-adic analogues.

\begin{proposition}
\label{SlK finite dim}For arbitrary $l\le R$ the space $\mathcal{S}_{R}%
^{l}(\mathbb{A}_{f})$ is non-trivial and finite dimensional, its dimension is
equal to $\Phi(R)/\Phi(l)$, with a basis given by the characteristic
functions of disjoint balls $B_{l}(x^{(n)})\subset B_{R}$. If
$f\in\mathcal{S}_{R}^{l}(\mathbb{A}_{f})$ then $\widehat f(\xi) =
\mathcal{F}_{x\rightarrow\xi}f \in\mathcal{S}_{(1/l)_{-}}^{(1/R)_{-}%
}(\mathbb{A}_{f})$. Moreover, $\mathcal{F}\mathcal{S}_{R}^{l}(\mathbb{A}_{f})
= \mathcal{S}_{(1/l)_{-}}^{(1/R)_{-}}(\mathbb{A}_{f})$.
\end{proposition}

\begin{proof}
Note that $B_{R}$ is a finite disjoint union of balls of type $B_{l}(x_{i})$
and the number of such balls is $\operatorname{vol}(B_{R}) /
\operatorname{vol}(B_{l})$. The first statement follows from this observation
by \eqref{representFS}.

For the second part it is enough to consider the Fourier transform of the
characteristic function of a ball $B_{l}(x^{(n)})\subset B_{R}$. We obtain
from Corollary \ref{FTBall}
\begin{multline*}
\widehat{\mathbf{1}}_{B_{l}(x^{(n)})} (\xi) = \int_{\mathbb{A}_{f}}
\chi(-x\cdot\xi)\mathbf{1}_{B_{l}}(x-x^{(n)})\, dx_{\mathbb{A}_{f}}\\
= \chi(-x^{(n)}\cdot\xi) \widehat{\mathbf{1}}_{B_{l}}(\xi) = \chi
(-x^{(n)}\cdot\xi) \Phi(l) \mathbf{1}_{B_{(1/l)_{-}}}(\xi),
\end{multline*}
hence the Fourier transform is supported in the ball $B_{(1/l)_{-}}$. Since
$x^{(n)}\in B_{R}$, for any $y\in B_{(1/R)_{-}}$ we have $\|x^{(n)} \cdot y\|
< 1$ hence $\chi(x^{(n)} \cdot y)=1$ and the Fourier transform of a ball
$B_{l}(x^{(n)})$ is locally constant with the parameter of constancy
$\ge(1/R)_{-}$. The last part follows from the observation that
$\bigl(\bigl((n^{-1})_{-}\bigr)^{-1}\bigr)_{-} =\bigl(\bigl((n^{-1}%
)^{-1}\bigr)_{+}\bigr)_{-} =(n_{+})_{-}= n$ for any non-zero power of a prime.
\end{proof}

\begin{proposition}
\label{Prop3}(i) Let $K$ be a compact subset of $\mathcal{S}(\mathbb{A}_{f})$.
The space of test functions $\mathcal{S}(\mathbb{A}_{f})$ is dense in the
space $C(K)$ of continuous functions on $K$. (ii) The space of test functions
$\mathcal{S}(\mathbb{A}_{f})$ is dense in $L^{\varrho}\left(  \mathbb{A}%
_{f}\right)  $ for $1\leq\varrho<\infty$.
\end{proposition}

\begin{proof}
The proof follows the classical pattern, see e.g. \cite{A-K-S},
\cite{Taibleson}, \cite{V-V-Z}.
\end{proof}

Denote by $\mathcal{S}^{\prime}(\mathbb{A}_{f})$ the $\mathbb{C}$-vector space
of all (complex-valued) linear continuous functionals on $\mathcal{S}
(\mathbb{A}_{f})$. This space is \textit{the space of Bruhat-Schwartz
distributions} on $\mathbb{A}_{f}$. For the sake of simplicity we will use
\textit{distribution} instead of \textit{Bruhat-Schwartz distribution}. We
equip $\mathcal{S}^{\prime}(\mathbb{A}_{f})$ with the weak topology. The
following proposition allows to simplify checking that a functional belongs to
the space $\mathcal{S}^{\prime}(\mathbb{A}_{f})$ stating that every linear
functional on $\mathcal{S} (\mathbb{A}_{f})$ is continuous.

\begin{proposition}
\label{Prop4}(i) $\mathcal{S}^{\prime}(\mathbb{A}_{f})$ is the $\mathbb{C}%
$-vector space of all (complex-valued) linear functionals on $\mathcal{S}%
(\mathbb{A}_{f})$. (ii) $\mathcal{S}^{\prime}(\mathbb{A}_{f})$ is complete.
\end{proposition}

\begin{proof}
Due to Proposition \ref{SlK finite dim} the proof of this proposition is
completely similar to the proof given for the analogous statement in the
$p$-adic case, see e.g. \cite{A-K-S}, \cite{V-V-Z}.
\end{proof}

\subsection{Pseudodifferential operators and the Lizorkin space on $\mathbb{A}_{f}$}

\label{SubsectDaAf}As it was mentioned in the introduction, the classical
derivative cannot be defined for complex-valued functions on adeles. Instead
we consider pseudodifferential operators. The function $\|\cdot\|$ which
generates the metric $\rho$ allows us to introduce a natural generalization of
the Taibleson operator $D_{\mathbb{A}_{f}}^{\gamma}=:D^{\gamma}$, $\gamma>0$,
defined on $\mathcal{S}(\mathbb{A}_{f})$ by
\begin{equation}
\left(  D^{\gamma}f\right)  \left(  x\right)  =\mathcal{F}_{\xi\rightarrow
x}^{-1}\left(  \left\Vert \xi\right\Vert ^{\gamma}\mathcal{F}_{x\rightarrow
\xi}f\right)  , \qquad f\in\mathcal{S}(\mathbb{A}_{f}). \label{POperator}%
\end{equation}

\begin{lemma}
\label{lemma_Poperator}With the above notation,%
\[
D^{\gamma}: \mathcal{S}(\mathbb{A}_{f}) \to C\left(  \mathbb{A}_{f}\right)
\cap L^{2}\left(  \mathbb{A}_{f}\right)  .
\]

\end{lemma}

\begin{proof}
Since $\mathcal{F}_{x\rightarrow\xi}f$ may be represented as a linear
combination of functions of type $\mathbf{1}_{B_{r}(\xi_{0})}(\xi)$, it is sufficient to 
consider the case $\mathcal{F}_{x\rightarrow\xi}f=\mathbf{1}_{B_{r}(\xi_{0})}(\xi)$. If $0\notin B_{r}(\xi_{0})$ then $\left\Vert \xi\right\Vert ^{\gamma}\mathbf{1}_{B_{r}(\xi
_{0})}(\xi)\in\mathcal{S}(\mathbb{A}_{f})$ because $\left\Vert \xi\right\Vert
^{\gamma}$ is locally constant outside of the origin, and hence
\[
\mathcal{F}_{\xi\rightarrow x}^{-1}\left(  \left\Vert \xi\right\Vert ^{\gamma
}\mathbf{1}_{B_{r}(\xi_{0})}(\xi)\right)  \in\mathcal{S}(\mathbb{A}
_{f})\subset L^{2}\left(  \mathbb{A}_{f}\right)  .
\]
If $0\in B_{r}(\xi_{0})$ then $B_{r}(\xi_{0}) = B_{r}(0)$ and $\left\Vert
\xi\right\Vert ^{\gamma}\leq r^\gamma$ on the $B_{r}(0)$, hence $\left\Vert
\xi\right\Vert ^{\gamma} \mathbf{1}_{B_{r}(0)}(\xi)\in L^{1}\left(
\mathbb{A}_{f}\right)  \cap L^{2}\left(  \mathbb{A}_{f}\right)  $. Thus
$\mathcal{F} _{\xi\rightarrow x}^{-1}\left(  \left\Vert \xi\right\Vert
^{\gamma} \mathbf{1}_{B_{r}(\xi_{0})}(\xi)\right)  \in C\left(  \mathbb{A}%
_{f}\right)  \cap L^{2}\left(  \mathbb{A}_{f}\right)  $.
\end{proof}

The space $\mathcal{S}(\mathbb{A}_{f})$ is not invariant under the action of the
operator $D^{\gamma}$. To overcome such an inconvenience, we introduce the
following space
\[
\mathcal{L}_{0}(\mathbb{A}_{f}):= \mathcal{L}_{0} = \bigl\{ f\in
\mathcal{S}(\mathbb{A}_{f}): \widehat f(0)=0\bigr\}.
\]
The space $\mathcal{L}_{0}$ can be equipped with the topology of the space
$\mathcal{S}(\mathbb{A}_{f})$, which makes $\mathcal{L}_{0}$ a complete space.
Note that
\begin{equation}
\label{eq_FL}\mathcal{L}_{0} = \mathcal{F} \bigl\{ h\in\mathcal{S}%
(\mathbb{A}_{f}): h(0) = 0\bigr\}.
\end{equation}
This space is an adelic analogue of the Lizorkin space of the second kind. We refer the reader to \cite{A-K-S} for the theory of the $p$-adic Lizorkin spaces. Recently in \cite{Kh-K-Sh} an adelic version of the Lizorkin space of the first kind was introduced.

\begin{lemma}
\label{LemmaDL0=L0} With the above notation the following assertions hold:

\begin{itemize}
\item[(i)] $D^{\gamma}\mathcal{L}_{0} = \mathcal{L}_{0}$ for $\gamma>0$.

\item[(ii)] $f\in\mathcal{L}_{0}$ if and only if $f\in\mathcal{S}$ and
$\int_{\mathbb{A}_{f}}f(x)\, dx_{\mathbb{A}_{f}} = 0$.

\item[(iii)] $\mathcal{L}_{0}$ is dense in $\mathcal{S}$ with respect to the
$L^{2}$-norm.

\item[(iv)] $\mathcal{L}_{0}$ is dense in $L^{2}(\mathbb{A}_{f})$.
\end{itemize}
\end{lemma}

\begin{proof}
(i) Take $f\in\mathcal{L}_{0}$, then $\|\xi\|^{\gamma}\widehat f(\xi)
\in\mathcal{S}(\mathbb{A}_{f})$ because $\widehat f$ is equal to 0 in some
neighborhood of 0 and $\|\xi\|^{\gamma}$ is locally constant outside of the
origin. Therefore $D^{\gamma}f\in\mathcal{L}_{0}$, i.e. $D^{\gamma}%
\mathcal{L}_{0}\subset\mathcal{L}_{0}$. The converse inclusion follows from
the fact that $\frac{\widehat h}{\|\xi\|^{\gamma}}\in\mathcal{S}%
(\mathbb{A}_{f})$ for any $h\in\mathcal{L}_{0}$.

(ii) The statement follows from $\widehat f(0) = \int_{\mathbb{A}_{f}}f(x)\,
dx_{\mathbb{A}_{f}}$ which is the consequence of $\mathcal{F}\mathcal{S}%
(\mathbb{A}_{f}) = \mathcal{S}(\mathbb{A}_{f})$.

(iii) By \eqref{eq_FL} and the fact that the Fourier transform preserves the
$L^{2}$ norm, it is sufficient to show that $\mathbf{1}_{B_{r}} (x)$ can be
arbitrarily closely approximated by functions from $\mathcal{L}_{0}$ in the
$L^{2}$ norm. As such approximating functions we may use $\mathbf{1}_{B_{r}}
(x) \setminus\mathbf{1}_{B_{p^{-m}}} (x)$ for $m$ big enough.

(iv) The statement follows from (iii) since $\mathcal{S}(\mathbb{A}_{f})$ is
dense $L^{2}(\mathbb{A}_{f})$, see Proposition \ref{Prop3}.
\end{proof}

We will use the notation $\mathcal{D}(D^{\gamma})$ to denote the domain of the
operator $D^{\gamma}$. We refer reader to \cite{R-S} for notions of
essentially self-adjoint operators and to \cite{Aren} and \cite{E-N} for the
definition of strongly continuous ($C_{0}$) semigroups and related notions.

Since $\mathcal{L}_{0}(\mathbb{A}_{f})\subset L^{2}(\mathbb{A}_{f})$, we may
consider the operator $D^{\gamma}$ as an operator acting on $L^{2}%
(\mathbb{A}_{f})$. It is easy to see that the operator $D^{\gamma}$ with the
domain $\mathcal{D}(D^{\gamma}) = \mathcal{L}_{0}(\mathbb{A}_{f})$ is
symmetric. Moreover, similarly to the proof of Lemma \ref{LemmaDL0=L0} (i) we may check that $(D^\gamma \pm i)\mathcal{L}_{0}(\mathbb{A}_{f}) = \mathcal{L}_{0}(\mathbb{A}_{f})$, i.e. the ranges of the operators $D^\gamma \pm i$ are dense in $L^2(\mathbb{A}_f)$, hence the operator $D^\gamma$ is essentially self-adjoint, see \cite[Corollary to
Theorem VIII.3]{R-S} for details. The following description of the
self-adjoint closure holds.

\begin{lemma}
\label{DgammaC0semigroup} The closure of the operator $D^{\gamma}$, $\gamma>0$
(let us denote it by $D^{\gamma}$ again) with domain
\begin{equation}
\label{DomDa}\mathcal{D}\left(  D^{\gamma}\right)  :=\left\{  f\in
L^{2}\left(  \mathbb{A}_{f}\right)  :\| \xi\|^{\gamma}\widehat{f}\in
L^{2}\left(  \mathbb{A}_{f}\right)  \right\}
\end{equation}
is a self-adjoint operator. Moreover, the following assertions hold:

\begin{itemize}
\item[(i)] $D^{\gamma}$ is a positive operator;

\item[(ii)] $D^{\gamma}$ is $m$-accretive, i.e. $-D^{\gamma}$ is an
$m$-dissipative operator;

\item[(iii)] the spectrum $\sigma(D^{\gamma}) = \big\{ p^{\gamma j}:
p\ \text{is a prime},\ j\in\mathbb{Z}\setminus\{0\}\big\}\cup\{0\}$;

\item[(iv)] $-D^{\gamma}$ is the infinitesimal generator of a contraction
$C_{0}$ semigroup $\bigl(\mathcal{T}(t)\bigr)_{t\ge 0}$. Moreover, the semigroup $\bigl(\mathcal{T}(t)\bigr)_{t\ge 0}$ is bounded holomorphic (or analytic) with angle $\pi/2$.
\end{itemize}
\end{lemma}

\begin{proof}
(i) It follows from the Steklov--Parseval equality that for any $f\in
L^{2}(\mathbb{A}_{f})$
\[
(D^{\gamma}f, f) = (\| \xi\|^{\gamma}\mathcal{F}f, \mathcal{F}f) =
\int_{\mathbb{A}_{f}} \| \xi\|^{\gamma}\big|\mathcal{F}f\big|^{2}
d\xi_{\mathbb{A}_{f}}\ge0.
\]

(ii), (iv) The result follows from the well-known corollary from the
Lumer-Phillips theorem, see e.g. \cite[Chapter 2, Section 3]{E-N} or \cite{C-H}. For the property of the semigroup of being holomorphic, see e.g. \cite[3.7]{Aren} or \cite[Chapter 2, Section 4.7]{E-N}.

(iii) Since $D^{\gamma}$ is self-adjoint and positive, $\sigma(D^{\gamma
})\subset[0,\infty)$. Consider the eigenvalue problem $D^{\gamma}f=\lambda f$,
$f\in\mathcal{D} (D^{\gamma})$, $\lambda>0$. By applying the Fourier transform
we obtain the equivalent equation
\begin{equation}
(\Vert\xi\Vert^{\gamma}-\lambda)\widehat{f}=0. \label{DaSpectr}%
\end{equation}
If $\lambda=p^{\gamma j}$ for some prime $p$ and $j\in\mathbb{Z}%
\setminus\{0\}$ then the inverse Fourier transform of the characteristic
function of $S_{p^{j}}=\{\xi\in\mathbb{A}_{f}:\Vert\xi\Vert=p^{j}\}$ is a
solution of \eqref{DaSpectr}. If $\lambda\not \in \big\{p^{\gamma
j}:p\ \text{is a prime},\ j\in\mathbb{Z}\setminus\{0\}\big\}$, then the
functions $\bigl|\frac{1}{\Vert\xi\Vert^{\gamma}-\lambda}\bigr|$ and $\bigl|\frac{\Vert\xi\Vert^{\gamma}}{\Vert\xi\Vert^{\gamma}-\lambda}\bigr|$ are bounded,
hence the equation $D^{\gamma}f-\lambda f=h$ is uniquely solvable for any
$h\in L^{2}(\mathbb{A}_{f})$ and $\lambda\in\rho(D^{\gamma})$. The point $0$
belongs to $\sigma(D^{\gamma})$ as a limit point.
\end{proof}

The representation of the generated semigroup $\bigl(\mathcal{T}(t)\bigr)_{t\ge 0}$ is presented in
detail in Theorem \ref{Theo4minus}.

\section{\label{Sect4}Metric structures, Distributions and Pseudodifferential
Operators on $\mathbb{A}$}

\subsection{A structure of complete metric space for the adeles}

We recall that $\mathbb{A}=\mathbb{R}\times\mathbb{A}_{f}$. Then any
$x\in\mathbb{A}$ can be written uniquely as $x=\left(  x_{\infty}
,x_{f}\right)  \in\mathbb{R}\times\mathbb{A}_{f}=\mathbb{A}$. Set for
$x,y\in\mathbb{A}$
\[
\rho_{\mathbb{A}}\left(  x,y\right)  :=| x_{\infty}-y_{\infty}|_{\infty}%
+\rho\left(  x_{f},y_{f}\right)  ,
\]
where $\rho\left(  x,y\right)  $ was defined in \eqref{MetrAf}. Then $\left(
\mathbb{A},\rho_{\mathbb{A}}\right)  $ is a complete metric space, see
Proposition \ref{Prop1}. Note that $\rho_{\mathbb{A}}\left(  x,y\right)  $ is
topologically equivalent to
\[
\tilde\rho(x,y) := \max\big\{  | x_{\infty}-y_{\infty}| _{\infty},\rho\left(
x,y\right)  \big\},\qquad x,y\in\mathbb{A},
\]
which induces on $\mathbb{A}$ the product topology. The topology of the
restricted product on $\mathbb{A}$ is equal to the product topology on
$\mathbb{R}\times\mathbb{A}_{f}$, where $\mathbb{R}$ is equipped with the
usual topology and $\mathbb{A}_{f}$ with the restricted product topology.
Hence the following result holds.

\begin{proposition}
\label{pro1a} The restricted product topology on $\mathbb{A}$ is metrizable, a
metric is given by $\rho_{\mathbb{A}}$. Furthermore, $\left(  \mathbb{A}%
,\rho_{\mathbb{A}}\right)  $ is a complete metric space and $\left(
\mathbb{A},\rho_{\mathbb{A}}\right)  $ as a topological space is homeomorphic
to $\left(  \mathbb{R},| \cdot| _{\infty}\right)  \times\left(  \mathbb{A}%
_{f},\rho\right)  $.
\end{proposition}

\begin{remark}
\label{nota3}$\left(  \mathbb{A},\rho_{\mathbb{A}}\right)  $ is a
second-countable topological space. Indeed, $\big(  B_{\infty}^{\left(
i\right)  }\times B_{f}^{\left(  j\right)  }\big)_{i,j\in\mathbb{N}}$ is a
countable base, where $\big(  B_{\infty}^{\left(  i\right)  }\big)
_{i\in\mathbb{N}}$ is a countable base of $\big(  \mathbb{R},| \cdot|
_{\infty}\big)  $ and $\big(  B_{f}^{\left(  j\right)  }\big)  _{j\in
\mathbb{N}}$ is a countable base of $\left(  \mathbb{A}_{f},\rho\right)  $,
see Remark \ref{nota2} (ii). Therefore $\left(  \mathbb{A},\rho_{\mathbb{A}%
}\right)  $ is a semi-compact space.
\end{remark}

\subsection{Distributions on $\mathbb{A}$}

The\textit{ space of Bruhat-Schwartz functions}, denoted $\mathcal{S}%
(\mathbb{A})$, consists of finite linear combinations of functions of type
$h\left(  x\right)  =h_{\infty}\left(  x_{\infty}\right)  h_{f}\left(
x_{f}\right)  $ with $h_{\infty}\in\mathcal{S}(\mathbb{R})$, Schwartz space on
$\mathbb{R}$, and $h_{f}\in\mathcal{S}(\mathbb{A}_{f})$. The space
$\mathcal{S}(\mathbb{A})$ is dense in $L^{\varrho}\left(  \mathbb{A}%
,dx_{\mathbb{A}}\right)  $ for $1\leq\varrho<+\infty$, see e.g. \cite[Theorem
2.9]{D-R-K}. The \textit{space of distributions on} $\mathcal{S}(\mathbb{A})$
is the strong dual space of $\mathcal{S}(\mathbb{A})$.

\subsection{Pseudodifferential operators and the Lizorkin space on $\mathbb{A}$}
We consider the pseudodifferential operator $D_{\mathbb{R}}^{\beta
}=:D^{\beta}$, $\beta>0$ on $\mathcal{S}(\mathbb{R})$ defined
by
\begin{equation}
\left(  D^{\beta}h\right)  \left(  x_\infty\right)  =\mathcal{F}
_{\xi_\infty\rightarrow x_\infty}^{-1}\left(  |\xi_{\infty}|_{\infty}^{\beta
}\mathcal{F}_{x_\infty\rightarrow\xi_\infty
}h\right)  ,\qquad h\in\mathcal{S}(\mathbb{R}). \label{POperatorInf}%
\end{equation}
Recall that the operator $D^\beta$ is the real Riesz fractional operator and represents a fractional power of the Laplacian, see e.g. \cite[\S8]{Samko}, \cite[\S25]{S-K-M}.

We introduce the pseudodifferential operator $D_{\mathbb{A}}^{\alpha,\beta
}=:D^{\alpha,\beta}$, $\alpha,\beta>0$ on $\mathcal{S}(\mathbb{A})$ defined
by
\begin{equation}
\left(  D^{\alpha,\beta}h\right)  \left(  x\right)  =\mathcal{F}%
_{\xi\rightarrow x}^{-1}\left(  \big(|\xi_{\infty}|_{\infty}^{\beta
}+\left\Vert \xi_{f}\right\Vert ^{\alpha}\big)\mathcal{F}_{x\rightarrow\xi
}h\right)  ,\qquad h\in\mathcal{S}(\mathbb{A}). \label{POperator2}%
\end{equation}

\begin{lemma}
With the above notation
\[
D^{\alpha,\beta}: \mathcal{S}(\mathbb{A}) \to C\left(  \mathbb{A}%
,\mathbb{C}\right)  \cap L^{2}\left(  \mathbb{A}\right)
\]

\end{lemma}

\begin{proof}
It is sufficient to show the result for a factorizable function $h=h_{\infty
}h_{f}$, $h_{\infty}\in\mathcal{S}\left(  \mathbb{R}\right)  $, $h_{f}%
\in\mathcal{S}(\mathbb{A}_{f})$. Since $\widehat{h}\left(  \xi\right)
=\widehat{h_{\infty}}\left(  \xi_{\infty}\right)  \widehat{h_{f}}\left(
\xi_{f}\right)  $,
\begin{equation}
\left(  D^{\alpha,\beta}h\right)  \left(  x\right)  =h_{f}\left(
x_{f}\right)  \left(  D^{\beta}h_{\infty}\right)  \left(  x_{\infty}\right)
+h_{\infty}\left(  x_{\infty}\right)  \left(  D^{\alpha}h_{f}\right)  \left(
x_{f}\right)  . \label{eq12}%
\end{equation}
Note that $D^{\beta}h_{\infty}\in L^{2}(\mathbb{R})\cap C\left(  \mathbb{R},\mathbb{C}\right)$, $D^{\alpha}
h_{f}\in L^{2}(\mathbb{A}_{f})\cap C\left(  \mathbb{A}_{f},\mathbb{C}\right)
$, cf. Lemma \ref{lemma_Poperator}, and since $dx_{\mathbb{A}}=dx_{\infty
}dx_{\mathbb{A}_{f}}$ we conclude $h_{f}D^{\beta}h_{\infty}$, $h_{\infty
}D^{\alpha}h_{f}\in C\left(  \mathbb{A},\mathbb{C}\right)  \cap
L^{2}\left(  \mathbb{A}\right)  $.
\end{proof}

The space $\mathcal{S}(\mathbb{A})$ is not invariant under the action of  the
operator $D^{\alpha,\beta}$. To overcome such an inconvenience, we introduce an adelic version of the Lizorkin space of the second kind. First we recall that the real Lizorkin space of test functions, see e.g. \cite[\S2]{Samko} or \cite[\S25]{S-K-M}, is defined by
\begin{equation*}
    \mathcal{L}_0(\mathbb{R}) = \biggl\{ f_\infty\in \mathcal{S}(\mathbb{R}): \int_{\mathbb{R}} x^n_\infty f_\infty(x_\infty)\,dx_\infty = 0,\ \text{for}\  n\in\mathbb{N}\biggr\}.
\end{equation*}
The real Lizorkin space can be equipped with the topology of the space $\mathcal{S}(\mathbb{R})$, which makes $\mathcal{L}_0(\mathbb{R})$ a complete space. The real Lizorkin space is invariant with respect to $D^\beta$, is dense in $L^p(\mathbb{R})$, $1<p<\infty$, and admits the following characterization: $f_\infty\in\mathcal{L}_0(\mathbb{R})$ if and only if $f_\infty\in\mathcal{S}(\mathbb{R})$ and
\[
\left.\frac{d^n}{d\xi_\infty^n} \mathcal{F} f(\xi_\infty) \right |_{\xi_\infty = 0} = 0,\qquad \text{for}\ n\in\mathbb{N}.
\]
We introduce an \emph{adelic Lizorkin space of the second kind} $\mathcal{L}_{0}:=\mathcal{L}_{0}(\mathbb{A})$ as
\[
\mathcal{L}_{0}(\mathbb{A}) = \mathcal{L}_0(\mathbb{R}) \otimes\mathcal{L}_{0}(\mathbb{A}_{f}).
\]
The space $\mathcal{L}_{0}(\mathbb{A})$ consists of finite linear combinations of factorizable functions $h(x) = h_\infty(x_\infty) h_f(x_f)$ with $h_{\infty}\in\mathcal{L}_0\left(  \mathbb{R}\right)  $, $h_{f}
\in\mathcal{L}_0(\mathbb{A}_{f})$. Note that $\mathcal{L}_{0}(\mathbb{A})$ is a subspace of $\mathcal{S}(\mathbb{A})$ and it may be equipped with the topology of $\mathcal{S}(\mathbb{A})$.

\begin{lemma}\label{LemmaDabL0=L0}
With the above notation the following assertions hold:

\begin{itemize}
\item[(i)] $D^{\alpha,\beta} \mathcal{L}_{0} = \mathcal{L}_{0}$ for
$\alpha,\beta>0$;

\item[(ii)] $\mathcal{L}_{0}$ is dense in $L^{2}(\mathbb{A})$.
\end{itemize}
\end{lemma}

\begin{proof}(i) It is sufficient to consider a factorizable function $h=h_{\infty
}h_{f}$, $h_{\infty}\in\mathcal{L}_0\left(  \mathbb{R}\right)  $, $h_{f}
\in\mathcal{L}_0(\mathbb{A}_{f})$. Since $D^\alpha h_f \in\mathcal{L}_0(\mathbb{A}_f)$ and $D^\beta h_\infty \in\mathcal{L}_0(\mathbb{R})$, see Lemma \ref{LemmaDL0=L0} and \cite[(9.1)]{Samko}, we conclude from \eqref{eq12} that $D^{\alpha,\beta}\mathcal{L}_0(\mathbb{A}) \subset \mathcal{L}_0(\mathbb{A})$.

Conversely, take $h\in\mathcal{L}_0(\mathbb{A})$. We want to show that the equation $D^{\alpha,\beta}g = h$ has a solution $g\in\mathcal{L}_0(\mathbb{A})$. We may assume without loss of generality that $h=h_{\infty
}h_{f}$, $h_{\infty}\in\mathcal{L}_0\left(  \mathbb{R}\right)  $, $h_{f}
\in\mathcal{L}_0(\mathbb{A}_{f})$. Applying the Fourier transform we obtain
\begin{equation}\label{eq_FTg}
\widehat g(\xi) = \frac{\widehat{h}_\infty(\xi_\infty) \widehat{h}_f(\xi_f)}{|\xi_\infty|_\infty^\beta + \|\xi_f\|^\alpha}.
\end{equation}
Since $h_f\in\mathcal{L}_0(\mathbb{A}_{f})$, it follows from \eqref{representFS} that
\[
\widehat{h}_f(\xi_f) = \sum_{i=1}^N c_i \mathbf{1}_{B_R(\xi_i)}(\xi_f),
\]
where the balls $B_R(\xi_i)$ are disjoint and $0\not\in B_R(\xi_i)$ for $i=1,\ldots,N$. It follows from non-Archimedean property that the function $\|\xi_f\|$ is constant on each of the balls $B_R(\xi_i)$, hence we may rewrite \eqref{eq_FTg} as
\begin{equation*}
\widehat g(\xi) = \sum_{i=1}^N \frac{c_i \widehat{h}_\infty(\xi_\infty) }{|\xi_\infty|_\infty^\beta + d_i} \mathbf{1}_{B_R(\xi_i)}(\xi_f),
\end{equation*}
where $d_i:=\|\xi_i\|^\alpha>0$ are constants. It may be easily checked that the functions $\frac{c_i \widehat{h}_\infty(\xi_\infty) }{|\xi_\infty|_\infty^\beta + d_i}$ are Fourier transforms of real Lizorkin functions and the functions $\mathbf{1}_{B_R(\xi_i)}(\xi_f)$ are Fourier transforms of Lizorkin functions on $\mathbb{A}_f$, thus $g\in \mathcal{L}_0(\mathbb{A})$.

(ii) Since $\mathcal{L}_0(\mathbb{R})$ is dense in $L^2(\mathbb{R})$, see \cite[Thm. 3.2]{Samko}, and  $\mathcal{L}_0(\mathbb{A}_f)$ is dense in $L^2(\mathbb{A}_f)$ by  Lemma \ref{LemmaDL0=L0}, the tensor product $L_0(\mathbb{A}) = \mathcal{L}_0(\mathbb{R}) \otimes\mathcal{L}_{0}(\mathbb{A}_{f})$ is dense in the tensor product $L^2(\mathbb{R})\otimes L^2(\mathbb{A}_f)$ which is isomorphic to the space $L^2(\mathbb{A})$, see e.g. \cite[Theorem II.10]{R-S}.
\end{proof}

Similarly to Subsection \ref{SubsectDaAf} we may consider the operator
$D^{\alpha,\beta}$ as an operator acting on $L^{2}(\mathbb{A})$. It is easy to
see that the operator $D^{\alpha,\beta}$ with the domain $\mathcal{D}%
(D^{\alpha,\beta}) = \mathcal{S}(\mathbb{A})$ is symmetric. Moreover, similarly to the proof of Lemma \ref{LemmaDabL0=L0} (i) we may check that $(D^{\alpha, \beta} \pm i)\mathcal{L}_{0}(\mathbb{A}) = \mathcal{L}_{0}(\mathbb{A})$, i.e. the ranges of the operators $D^{\alpha, \beta} \pm i$ are dense in $L^2(\mathbb{A}_f)$, hence the operator $D^{\alpha, \beta}$ is essentially
self-adjoint. The following description of the closure holds.

\begin{lemma}
\label{Lemma5} The closure of the operator $D^{\alpha,\beta}$, $\alpha
,\beta>0$ (let us denote it by $D^{\alpha,\beta}$ again) with domain
\begin{equation}
\label{DabDom}\mathcal{D}\left(  D^{\alpha,\beta}\right)  :=\left\{  f\in
L^{2}\left(  \mathbb{A}\right)  : \big( | \xi_{\infty}| _{\infty}^{\beta
}+\left\Vert \xi\right\Vert ^{\alpha}\big)\widehat{f}\in L^{2}\left(
\mathbb{A}\right)  \right\}
\end{equation}
is a self-adjoint operator. Moreover, the following assertions hold:

\begin{itemize}
\item[(i)] $D^{\alpha,\beta}\ge0$;

\item[(ii)] $D^{\alpha,\beta}$ is $m$-accretive, i.e. $-D^{\alpha,\beta}$ is
an $m$-dissipative operator;

\item[(iii)] the spectrum $\sigma(D^{\alpha,\beta}) = [0,\infty)$;

\item[(iv)] $-D^{\alpha,\beta}$ is the infinitesimal generator of a
contraction $C_{0}$ semigroup $\bigl(\mathcal{T_{\alpha,\beta}}(t)\bigr)_{t\ge 0}$. Moreover, the semigroup $\bigl(\mathcal{T_{\alpha,\beta}}(t)\bigr)_{t\ge 0}$ is bounded holomorphic (or analytic) with angle $\pi/2$.
\end{itemize}
\end{lemma}

\section{The Adelic Heat Kernel on $\mathbb{A}_{f}$}

\label{SectAdelicHeat}

In this section we introduce the adelic heat kernel on $\mathbb{A}_{f}$ as the
inverse Fourier transform of $e^{-t\Vert y\Vert^{\alpha}}$ with $y\in
\mathbb{A}_{f}$, $\Vert y\Vert$ defined by \eqref{Norm}, $\alpha>1$ and $t>0$.

In Sections \ref{SectAdelicHeat}, \ref{SectMarkov} and \ref{SectCauchy} we
work only with finite adeles, for this reason in the variables we omit the
subindex `$_{f}$'.

\begin{proposition}
\label{pro1} Consider the function $\left\Vert y\right\Vert ^{\beta
}e^{-t\left\Vert y\right\Vert ^{\alpha}}$ for fixed $t>0$, $\beta\geq0$ and
$\alpha>1$. Then
\[
\left\Vert y\right\Vert ^{\beta}e^{-t\left\Vert y\right\Vert ^{\alpha}}\in
L^{\varrho}\left(  \mathbb{A}_{f},dy_{\mathbb{A}_{f}}\right)
\]
for any $1\leq\varrho<+\infty$.
\end{proposition}

\begin{proof}
It is sufficient to show that for any $t>0$ and $\beta\geq0$
\[
I(t):=\int_{\mathbb{A}_{f}}\left\Vert y\right\Vert ^{\beta}e^{-t\left\Vert
y\right\Vert ^{\alpha}}dy_{\mathbb{A}_{f}}<+\infty.
\]

According to Lemmas \ref{integral_radi_function} and \ref{Lemma2A}
\[
\int_{\mathbb{A}_{f}} \left\Vert y\right\Vert ^{\beta}e^{-t\left\Vert
y\right\Vert ^{\alpha}} \,dy_{\mathbb{A}_{f}} = \sum_{p^{m},\ m\ne0}
p^{m\beta} e^{-tp^{m\alpha}} \bigl(\Phi(p^{m}) - \Phi(p^{m}_{-})\bigr),
\]
thus we have to prove the convergence of the latter series. We consider two cases:
$m<0$ and $m>0$.

If $m<0$, then $p^{m\beta} e^{-tp^{m\alpha}}\le1$ and
\begin{multline*}
S_{-}(t):=\sum_{p^{m},\ m< 0} p^{m\beta} e^{-tp^{m\alpha}} \bigl(\Phi(p^{m}) -
\Phi(p^{m}_{-})\bigr)\\
\le\sum_{p^{m},\ m< 0} \bigl(\Phi(p^{m}) - \Phi(p^{m}_{-})\bigr) = \Phi(1/2).
\end{multline*}

Let $m>0$. We have
\[
S_{+}(t) := \sum_{p^{m},\ m> 0} p^{m\beta} e^{-tp^{m\alpha}} \bigl(\Phi(p^{m})
- \Phi(p^{m}_{-})\bigr) \le\sum_{p^{m},\ m> 0} p^{m\beta} e^{-tp^{m\alpha}}
\Phi(p^{m}).
\]
We recall that the Prime Number Theorem is equivalent to
\[
\ln\Phi(x) = \psi\left(  x\right)  \sim x,\quad x\to\infty,
\]
see e.g. \cite{D}, hence there exists a constant $C$ such that
\[
S_{+}(t)\le\sum_{p}\sum_{m=1}^{\infty}p^{m\beta} e^{-tp^{m\alpha} +Cp^{m}}.
\]
We want to show the existence of a positive constant $M=M\left(  \beta\right)
$ such that
\[
p^{m \beta}e^{-tp^{m\alpha}+Cp^{m}}\leq Mp^{-1-m}\quad\text{ for all }%
m\ge1\ \text{and prime }p,
\]
or equivalently that $p^{1+m\left(  \beta+1\right)  }e^{-tp^{m\alpha}+Cp^{m}%
}\leq M$. Since $p^{1+m\left(  \beta+1\right)  }\leq e^{(\beta+1) p^{m}}$ for
all $m\ge1$ and $p\geq2$, consider $e^{\left(  C+\beta+1\right)
p^{m}-tp^{\alpha m}}$. This expression is less than or equal to $1$ when
$p^{m}\geq\left(  \frac{C+\beta+1}{t}\right)  ^{\frac{1}{\alpha-1}}$ and hence
there exist only a finite number of pairs $\left(  p,m\right)  $ for which it
can be greater than $1$, so the announced constant exists. Therefore
\[
S_{+}(t)\le M\sum_{p} \sum_{m=1}^{\infty} p^{-1-m}\le2M \sum_{p}
p^{-2}<+\infty.\qedhere
\]

\end{proof}

\begin{definition}
\label{DefAdelicheathkernel}We define \textit{the adelic heat kernel} on
$\mathbb{A}_{f}$ as%
\begin{equation}
Z\left(  x,t;\alpha\right)  :=Z\left(  x,t\right)  =\int_{\mathbb{A}_{f}}%
\chi\left(  \xi\cdot x\right)  e^{-t\left\Vert \xi\right\Vert ^{\alpha}}%
d\xi_{\mathbb{A}_{f}},\quad x\in\mathbb{A}_{f},\ t>0,\ \alpha>1.
\label{AdelicHeat}%
\end{equation}

\end{definition}

By Proposition \ref{pro1} the integral is convergent.
When considering $Z\left(  x,t\right)  $ as a function of $x$ for $t$ fixed we
will write $Z_{t}\left(  x\right)  $. By applying Theorem
\ref{LemmaFourierRadial} to the function $e^{-t\Vert\xi\Vert^{\alpha}}$ we obtain
the following result.

\begin{proposition}
\label{pro2}The following representation holds for the heat kernel:
\begin{equation}
Z\left(  x,t\right)  =\sum_{q^{j}<\|x\|^{-1},\ j\ne0 }\Phi\left(
q^{j}\right)  \Big(e^{-tq^{j\alpha}}-e^{-t(q_{+}^{j})^{\alpha}}\Big)\quad
\text{ for}\ t>0,\ x\in\mathbb{A}_{f}, \label{formula}%
\end{equation}
where $q^{j}$ runs through all non-zero powers of prime numbers; functions
$\|x\|$, $\Phi(x)$ and $q_{+}^{j}$ are defined by \eqref{Norm}, \eqref{defphi}
and \eqref{nplusdef}. For $x=0$ the expression $\|0\|^{-1}$ in the
representation means $\infty$.
\end{proposition}

\begin{lemma}
\label{HeatKernelEstimate} The following estimate holds for the heat kernel:
\begin{equation}
\label{EqHKestimate}Z(x,t) \le2t \|x\|^{-\alpha} \Phi\bigl(\|x\|^{-1}%
_{-}\bigr),\quad x\in\mathbb{A}_{f}\setminus\{0\},\ t>0.
\end{equation}

\end{lemma}

\begin{proof}
From the inequality $1-e^{-x}\le x$ valid for $x\ge0$ we obtain
\[
e^{-tq^{j\alpha}}-e^{-t(q_{+}^{j})^{\alpha}} \le1-e^{-t(q_{+}^{j})^{\alpha}%
}\le t (q_{+}^{j})^{\alpha}.
\]
Then with the use of the inequality $\frac12 \Phi(q^{j}) \le\Phi(q^{j})
-\Phi(q^{j}_{-})$ we have
\begin{multline*}
Z(x,t) = \sum_{q^{j}<\|x\|^{-1}}\Phi( q^{j}) \Big(e^{-tq^{j\alpha}%
}-e^{-t(q_{+}^{j})^{\alpha}}\Big) \le t\sum_{q^{j}<\|x\|^{-1}}\Phi( q^{j})
(q_{+}^{j})^{\alpha}\\
\le2 t\|x\|^{-\alpha} \sum_{q^{j}<\|x\|^{-1}}\bigl(\Phi( q^{j}) - \Phi
(q^{j}_{-})\bigr) = 2t \|x\|^{-\alpha} \Phi(\|x\|^{-1}_{-}).
\end{multline*}
\vskip-1.5em
\end{proof}

\begin{corollary}
\label{cor1}With the above notation the following assertions hold:

\begin{itemize}
\item[(i)] $Z\left(  x,t\right)  \geq0$ for $t>0$;

\item[(ii)] $\lim_{t\rightarrow0+}Z\left(  x,t\right)  =0$ for any
$x\in\mathbb{A}_{f}\setminus\left\{  0\right\}  $;

\item[(iii)] For any $\epsilon>0$ there exists a constant $C=C(\epsilon)$ such
that for any $t>0$
\begin{equation}
\label{EqHKintEstimate}\int_{\|y\|>\epsilon} Z_{t}(y)\, dy_{\mathbb{A}_{f}}
\le C t <+\infty.
\end{equation}

\end{itemize}
\end{corollary}

\begin{proof}
The statements (i) and (ii) immediately follows from formulas (\ref{formula})
and (\ref{EqHKestimate}), respectively.

(iii) By Proposition \ref{pro2}, Lemma \ref{integral_radi_function} and
\eqref{EqHKestimate} we have
\begin{multline*}
\int_{\|y\|>\epsilon} Z_{t}(y)dy_{\mathbb{A}_{f}} = \sum_{p^{k}>\epsilon
}\operatorname{vol}(S_{p^{k}})\cdot\sum_{q^{j}<p^{-k}}\Phi( q^{j})
\bigl(e^{-tq^{j\alpha}}-e^{-t(q_{+}^{j})^{\alpha}}\bigr)\\
\le\sum_{p^{k}>\epsilon} \Phi(p^{k}) \cdot2t p^{-k\alpha}\Phi(p^{-k}%
_{-})=2t\cdot\sum_{p^{k}>\epsilon}p^{-k\alpha}<+\infty,
\end{multline*}
where we have used \eqref{Phi(1/x)} and \eqref{identity}.
\end{proof}

\begin{theorem}
\label{Theo1} The adelic heat kernel on $\mathbb{A}_{f}$ satisfies the following:

\begin{itemize}
\item[(i)] $Z\left(  x,t\right)  \geq0$ for any $t>0$;

\item[(ii)] $\int_{\mathbb{A}_{f}} Z_{t}\left(  x\right)  dx_{\mathbb{A}_{f}%
}=1$ for any $t>0$;

\item[(iii)] $Z_{t}(x)\in L^{1}(\mathbb{A}_{f})$ for any $t>0$;

\item[(iv)] $Z_{t}\left(  x\right)  \ast Z_{t^{\prime}}\left(  x\right)
=Z_{t+t^{\prime}}\left(  x\right)  $ for any $t$, $t^{\prime}>0$;

\item[(v)] $\lim_{t\rightarrow0+}Z_{t}\left(  x\right)  =\delta\left(
x\right)  $ in $S^{\prime}\left(  \mathbb{A}_{f}\right)  $;

\item[(vi)] $Z_{t}\left(  x\right)  $ is a uniformly continuous function for
any fixed $t>0$;

\item[(vii)] $Z(x,t)$ is uniformly continuous in $t$, i.e. $Z(x,t)\in
C((0,\infty), C(\mathbb{A}_{f}))$ or $\lim_{t^{\prime}\to t}\max
_{x\in\mathbb{A}_{f}}|Z(x,t) - Z(x,t^{\prime})|=0$ for any $t>0$.
\end{itemize}
\end{theorem}

\begin{proof}
(i) It follows from Corollary \ref{cor1}.

(ii) For any $t>0$ the function $e^{-t\left\Vert \xi\right\Vert ^{\alpha}}$ is
continuous at $\xi=0$ and by Proposition \ref{pro1} we have $e^{-t\left\Vert
\xi\right\Vert ^{\alpha}}\in L^{1}\left(  \mathbb{A}_{f}\right)  \cap
L^{2}\left(  \mathbb{A}_{f}\right)  $. Then $Z_{t}\left(  x\right)  \in
C\left(  \mathbb{A}_{f},\mathbb{R}\right)  \cap L^{2}\left(  \mathbb{A}%
_{f}\right)  $. Now the statement follows from the inversion formula for the
Fourier transform on $\mathbb{A}_{f}$.

(iii) The statement follows from (i) and (ii).

(iv) By the previous property $Z_{t}\left(  x\right)  \in L^{1}\left(
\mathbb{A}_{f}\right)  $ for any $t>0$. Then%
\[
Z_{t}\left(  x\right)  \ast Z_{t^{\prime}}\left(  x\right)  =\mathcal{F}%
_{\xi\rightarrow x}^{-1}\left(  e^{-t\left\Vert \xi\right\Vert ^{\alpha}%
}e^{-t^{\prime}\left\Vert \xi\right\Vert ^{\alpha}}\right)  =\mathcal{F}%
_{\xi\rightarrow x}^{-1}\left(  e^{-\left(  t+t^{\prime}\right)  \left\Vert
\xi\right\Vert ^{\alpha}}\right)  =Z_{t+t^{\prime}}\left(  x\right)  .
\]

(v) Since $e^{-t\left\Vert \xi\right\Vert ^{\alpha}}\in C\left(
\mathbb{A}_{f},\mathbb{R}\right)  \cap L^{1}\left(  \mathbb{A}_{f}\right)  $,
cf. Proposition \ref{pro1}, the scalar product
\[
\big(e^{-t\left\Vert \xi\right\Vert ^{\alpha}},f\left(  \xi\right)
\big)=\int_{\mathbb{A}_{f}}e^{-t\left\Vert \xi\right\Vert ^{\alpha}}%
\overline{f\left(  \xi\right)  }d\xi_{\mathbb{A}_{f}}\quad\text{ with }%
f\in\mathcal{S}\left(  \mathbb{A}_{f}\right)
\]
defines a distribution on $\mathbb{A}_{f}$. Since the support of $f$ is
compact, cf. Lemma \ref{Lemma3}, and $e^{-t\left\Vert \xi\right\Vert ^{\alpha
}}\in L^{1}\left(  \mathbb{A}_{f}\right)  $, the Dominated Convergence Lemma
with the characteristic function of the support of $f$ as a dominant function
implies
\[
\lim_{t\rightarrow0{+}}\left(  e^{-t\left\Vert \xi\right\Vert ^{\alpha}%
},f\left(  \xi\right)  \right)  =\left(  1,f\right)
\]
and then, as $\mathcal{F}\left(  \mathcal{S}\left(  \mathbb{A}_{f}\right)
\right)  =\mathcal{S}\left(  \mathbb{A}_{f}\right)  $, we have
\[
\lim_{t\rightarrow0{+}}\big(Z(x,t),f\big)=\lim_{t\rightarrow0{+}}\left(
\mathcal{F}_{\xi\rightarrow x}^{-1}\big(e^{-t\left\Vert \xi\right\Vert
^{\alpha}}\big),f\left(  x\right)  \right)  =\left(  1,\mathcal{F}%
_{\xi\rightarrow x}^{-1}f\right)  =\left(  \delta,f\right)  .
\]

(vi) Since $Z_{t}(x)=\mathcal{F}_{\xi\rightarrow x}^{-1}\left(
e^{-t\left\Vert \xi\right\Vert ^{\alpha}}\right)  $ and $e^{-t\left\Vert
\xi\right\Vert ^{\alpha}}\in L^{1}\left(  \mathbb{A}_{f}\right)  $ for $t>0$,
$Z_{t}(x)$ is uniformly continuous in $x$ for any fixed $t>0$.

(vii) Suppose that $t<t^{\prime}$. By the Mean value theorem $e^{-t\left\Vert
\xi\right\Vert ^{\alpha}}-e^{-t^{\prime}\left\Vert \xi\right\Vert ^{\alpha}} =
(t^{\prime}-t)\|\xi\|^{\alpha}e^{-t(\|\xi\|)\left\Vert \xi\right\Vert
^{\alpha}}$, where $t<t(\|\xi\|)<t^{\prime}$. Hence
\begin{multline*}
|Z(x,t) - Z(x,t^{\prime})| = \biggl|\int_{\mathbb{A}_{f}} \chi\left(  \xi\cdot
x\right)  \Bigl(e^{-t\left\Vert \xi\right\Vert ^{\alpha}}-e^{-t^{\prime
}\left\Vert \xi\right\Vert ^{\alpha}}\Bigr) d\xi_{\mathbb{A}_{f}}\biggr|\\
= |t-t^{\prime}| \biggl|\int_{\mathbb{A}_{f}} \chi\left(  \xi\cdot x\right)
\|\xi\|^{\alpha}e^{-t(\|\xi\|)\left\Vert \xi\right\Vert ^{\alpha}}
d\xi_{\mathbb{A}_{f}}\biggr| \le|t-t^{\prime}| \int_{\mathbb{A}_{f}}
\|\xi\|^{\alpha}e^{-t_{0}\left\Vert \xi\right\Vert ^{\alpha}} d\xi
_{\mathbb{A}_{f}},
\end{multline*}
for some $0<t_{0}<t,t^{\prime}$. Now the statement follows from Proposition
\ref{pro1}.
\end{proof}

\section{Markov Processes on $\mathbb{A}_{f}$}

\label{SectMarkov} Along this section we consider $\left(  \mathbb{A}_{f}%
,\rho\right)  $ as the complete non-Archimedean metric space and use the
terminology, notation and results of \cite[Chapters Two, Three]{Dyn}. Let
$\mathcal{B}$ denote the $\sigma$-algebra of the Borel sets of $\mathbb{A}%
_{f}$. Then $\left(  \mathbb{A}_{f},\mathcal{B},dx_{\mathbb{A}_{f}}\right)  $
is a measure space. Let $\mathbf{1}_{B}\left(  x\right)  $ denote the
characteristic function of a set $B\in\mathcal{B}$.

We assume along this section that $\alpha>1$ and set
\[
p\left(  t,x,y\right)  :=Z\left(  x-y,t\right)  \qquad\text{for }%
t>0,\,x,y\in\mathbb{A}_{f},
\]
and
\[
P\left(  t,x,B\right)  :=%
\begin{cases}
\int_{B} p\left(  t,x,y\right)  dy_{\mathbb{A}_{f}}, & \text{for }
t>0,\,x\in\mathbb{A}_{f},B\in\mathcal{B}\\
\mathbf{1}_{B}\left(  x\right)  , & \text{for } t=0.
\end{cases}
\]

\begin{lemma}
\label{lema1}With the above notation the following assertions hold:

\begin{itemize}
\item[(i)] $p\left(  t,x,y\right)  $ is a normal transition density;

\item[(ii)] $P\left(  t,x,B\right)  $ is a normal transition function.
\end{itemize}
\end{lemma}

\begin{proof}
The result follows from Theorem \ref{Theo1}, see \cite[Sec.2.1]{Dyn} for
further details.
\end{proof}

\begin{lemma}
\label{lema22}The transition function $P\left(  t,y,B\right)  $ satisfies the
following two conditions:

\begin{itemize}
\item[(i)] for each $u\geq0$ and a compact $B$
\[
\lim_{x\rightarrow\infty}\sup_{t\leq u}P\left(  t,x,B\right)  =0;\qquad
[\text{Condition } L(B)]
\]

\item[(ii)] for each $\epsilon>0$ and a compact $B$
\[
\lim_{t\rightarrow0{+}}\sup_{x\in B}P\left(  t,x,\mathbb{A}_{f}\setminus
B_{\epsilon}\left(  x\right)  \right)  =0.\qquad[\text{Condition } M(B)]
\]

\end{itemize}
\end{lemma}

\begin{proof}
Since $B$ is a compact, $\operatorname{dist} (x, B)=: d(x)\to\infty$ as
$x\to\infty$. Since the function $\Phi(x)$ is non-decreasing, we obtain from
\eqref{EqHKestimate} that $Z(x-y,t)\le2u\bigl(d(x)\bigr)^{-\alpha}
\Phi\bigl((d(x))^{-1}\bigr)$ for any $y\in B$ and $t\le u$. Hence $P(t, x, B)
\le2u\bigl(d(x)\bigr)^{-\alpha}\cdot\Phi\bigl((d(x))^{-1}\bigr) \cdot
\operatorname{vol} (B) \to0$ as $x\to\infty$.

To verify Condition $M(B)$ we proceed as follows: for $y\in\mathbb{A}%
_{f}\setminus B_{\epsilon}\left(  x\right)  $ we have $\|x-y\|>\epsilon$. The
statement follows from \eqref{EqHKintEstimate}:
\[
P\left(  t,x,\mathbb{A}_{f}\setminus B_{\epsilon}\left(  x\right)  \right)
\le C(\epsilon)t \to0,\quad t\to0+.\qedhere
\]

\end{proof}

\begin{theorem}
\label{Theo2}$Z(x,t)$ is the transition density of a time- and space
homogenous Markov process which is bounded, right-continuous and has no
discontinuities other than jumps.
\end{theorem}

\begin{proof}
The result follows from \cite[Theorem 3.6]{Dyn}, Remark \ref{nota2} (ii) and
Lemmas \ref{lema1}, \ref{lema22}.
\end{proof}

\begin{remark}
The more strict version of Condition $M(B)$ which is sufficient for the
continuity of a Markov process, namely, that for each $\epsilon>0$ and a
compact $B$
\[
\lim_{t\rightarrow0{+}}\frac1t\sup_{x\in B}P\left(  t,x,\mathbb{A}%
_{f}\setminus B_{\epsilon}\left(  x\right)  \right)  =0,\qquad[\text{Condition
}N(B)]
\]
does not hold for the function $Z(x,t)$. This may be easily seen if we take
$\epsilon= 1/4$. In such case by Proposition \ref{pro2} and Lemma
\ref{integral_radi_function} we have
\begin{multline*}
\int_{\mathbb{A}_{f}\setminus B_{1/4}\left(  x\right)  } Z_{t}\left(
x-y\right)  dy_{\mathbb{A}_{f}} \ge\int_{S_{1/3}\left(  x\right)  }
Z_{t}\left(  x-y\right)  dy_{\mathbb{A}_{f}}\\
=\operatorname{vol} S_{1/3}\left(  x\right)  \cdot\sum_{q^{j}<3, j\ne0}
\Phi\left(  q^{j}\right)  \left(  e^{-tq^{j\alpha}}-e^{-t\left(  q^{j}
_{+}\right)  ^{\alpha}}\right)  \ge\frac13 \left(  e^{-2^{\alpha}%
t}-e^{-3^{\alpha}t}\right)  ,
\end{multline*}
hence
\[
\lim_{t\rightarrow0{+}}\frac1t\sup_{x\in B}P\left(  t,x,\mathbb{A}%
_{f}\setminus B_{1/4}\left(  x\right)  \right)  \ge\frac{3^{\alpha}-
2^{\alpha}}3 \ne0.
\]

\end{remark}

\section{Cauchy problem for parabolic type equations on $\mathbb{A}_{f}$}
\label{SectCauchy}

Consider the following Cauchy problem
\begin{equation}
\left\{
\begin{aligned} &\frac{\partial u(x,t)}{\partial t}+D^{\alpha}u(x,t)=0, && x\in \mathbb{A}_{f},\ t\in\left[ 0,+\infty\right),\\ &u(x,0)=u_{0}(x), && u_{0}(x)\in\mathcal{D}(D^\alpha), \end{aligned}\right.
\label{CauchyProb}%
\end{equation}
where $\alpha>1$, $D^{\alpha}$ is the pseudodifferential operator defined by
(\ref{POperator}) with the domain given by \eqref{DomDa} and $u:\mathbb{A}%
_{f}\times\lbrack0,\infty)\rightarrow\mathbb{C}$ is an unknown function.

We say that a function $u(x,t)$ is a \textit{solution of} (\ref{CauchyProb})
if \linebreak$u\in C\bigl([0,\infty),\mathcal{D}(D^{\alpha})\bigr)\cap
C^{1}\bigl([0,\infty),L^{2}(\mathbb{A}_{f})\bigr)$ and $u$ satisfies equation
\eqref{CauchyProb} for all $t\geq0$.

We understand the notions of continuity in $t$, differentiability in $t$ and
equalities in the $L^{2}(\mathbb{A}_{f})$ sense, as it is customary in the
semigroup theory. More precisely, we say that a function $u(x,t)$ is
continuous in $t$ at $t_{0}$ if $\lim_{t\rightarrow t_{0}}\Vert
u(x,t)-u(x,t_{0})\Vert_{L^{2}(\mathbb{A}_{f})}=0$; the function $u_{t}%
^{\prime}(x,t)$ is the time derivative of function $u(x,t)$ at $t_{0}$ if
$\lim_{t\rightarrow t_{0}}\bigl\|\frac{u(x,t)-u(x,t_{0})}{t-t^{\prime}}%
-u_{t}^{\prime}(x,t_{0})\bigr\|_{L^{2}(\mathbb{A}_{f})}=0$; two functions
$f(x,t)$ and $g(x,t)$ are equal at $t_{0}$ if $\Vert f(x,t_{0})-g(x,t_{0}%
)\Vert_{L^{2}(\mathbb{A}_{f})}=0$.

We know from Lemma \ref{DgammaC0semigroup} that the operator $-D^{\alpha}$
generates a $C_{0}$ semigroup. Therefore Cauchy problem \eqref{CauchyProb} is
well-posed, i.e. it is uniquely solvable with the solution continuously
dependent on the initial data, and its solution is given by
$u(x,t)=\mathcal{T}(t)u_{0}(x)$, $t\geq0$, see e.g. \cite{Aren}, \cite{C-H}, \cite{E-N}.
However the general theory does not give an explicit formula for the semigroup
$\bigl(\mathcal{T}(t)\bigr)_{t\ge 0}$. We show that the operator $\mathcal{T}(t)$ for $t>0$
coincides with the operator of convolution with the heat kernel $Z_{t}%
\ast\cdot$. In order to prove this, we first construct a solution of Cauchy
problem \eqref{CauchyProb} with the initial value from $\mathcal{S}%
(\mathbb{A}_{f})$ without using the semigroup theory. Then we extend the
result to all initial values from $\mathcal{D}(D^{\alpha})$, see Proposition
\ref{Corr3} and Theorem \ref{Theo4minus}.

We show in Theorem \ref{Theo3} that in the case $u_{0}\in\mathcal{L}%
_{0}(\mathbb{A}_{f})$, the function $u(x,t)$ is the solution of Cauchy problem
\eqref{CauchyProb} in a stricter sense, i.e. $u(x,t)\in C^{1}\bigl([0,\infty
),\mathcal{L}_{0}(\mathbb{A}_{f})\bigr)$ and all limits and equalities are
understood pointwise.

\subsection{\label{initialValSA}Homogeneous equations with initial values in
$\mathcal{S}(\mathbb{A}_{f})$}

We first consider Cauchy problem \eqref{CauchyProb} with the initial value
from the space $\mathcal{S}(\mathbb{A}_{f})$. To simplify notations, set
$Z_{0}\ast u_{0}=\left.  \big(Z_{t}\ast u_{0}\big)\right\vert _{t=0}:=u_{0}$.
Note that such definition is consistent with Theorem \ref{Theo1} (v). We
define the function
\begin{equation}
u(x,t)=Z_{t}(x)\ast u_{0}(x),\qquad t\geq0. \label{defn_u(x,t)}%
\end{equation}
Since $Z_{t}\left(  x\right)  \in L^{1}\left(  \mathbb{A}_{f}\right)  $ for
$t>0$ and $u_{0}\left(  x\right)  \in\mathcal{S}(\mathbb{A}_{f})\subset
L^{\infty}\left(  \mathbb{A}_{f}\right)  $, the convolution exists and is a
continuous function, see \cite[Theorem 1.1.6]{Rudin}.


\begin{lemma}
\label{LemmaCP1} Let $u_{0}\in\mathcal{S}(\mathbb{A}_{f})$ and $u(x,t)$,
$t\ge0$ is defined by \eqref{defn_u(x,t)}. Then $u(x,t)$ is continuously
differentiable in time for $t\ge0$ and the derivative is given by
\begin{equation}
\label{eq_dt_u}\frac{\partial u}{\partial t}(x,t) = -\mathcal{F}%
_{\xi\rightarrow x}^{-1}\bigl( \|\xi\|^{\alpha}e^{-t\|\xi\|^{\alpha}}
\cdot\mathbf{1}_{B_{R}}(\xi)\bigr) \ast u_{0}(x),
\end{equation}
where $\mathbf{1}_{B_{R}}(\cdot)$ is the characteristic function of the ball
$B_{R}$, $R= (1/\ell)_{-}$ and $\ell$ is the parameter of constancy of the
function $u_{0}$, see \eqref{Local Constancy} and Proposition
\ref{SlK finite dim}.
\end{lemma}

\begin{proof}
Let $h_{t}(x)$ be a function defined by the right-hand side of \eqref{eq_dt_u}.
Since $\|\xi\|^{\alpha}e^{-t\|\xi\|^{\alpha}} \cdot\mathbf{1}_{B_{R}}(\xi)\in
L^{1}(\mathbb{A}_{f}) \cap L^{2}(\mathbb{A}_{f})$ for any $t\ge0$, the
function $h_{t}(x)$ is well-defined and belongs to $C(\mathbb{A}_{f})\cap
L^{2}(\mathbb{A}_{f})$.

Let $t_{0}\ge0$. Consider a limit
\begin{multline*}
\lim_{t\to t_{0}} \Bigl\|\frac{u(x,t)-u(x,t_{0})}{t-t_{0}} - h_{t}%
(x,t_{0})\Bigr\|_{L^{2}(\mathbb{A}_{f})}\\
=\lim_{t\to t_{0}}\Bigl\| \frac{e^{-t\|\xi\|^{\alpha}} - e^{-t_{0}%
\|\xi\|^{\alpha}}}{t-t_{0}} \widehat u_{0}(\xi) + \|\xi\|^{\alpha}%
e^{-t_{0}\|\xi\|^{\alpha}}\mathbf{1}_{B_{R}}(\xi)\cdot\widehat u_{0}(\xi)
\Bigr\|_{L^{2}(\mathbb{A}_{f})}\\
=\lim_{t\to t_{0}}\Bigl\| \Bigl(\frac{e^{-t\|\xi\|^{\alpha}} - e^{-t_{0}%
\|\xi\|^{\alpha}}}{t-t_{0}} + \|\xi\|^{\alpha}e^{-t_{0}\|\xi\|^{\alpha}%
}\Bigr) \mathbf{1}_{B_{R}}(\xi)\cdot\widehat u_{0}(\xi) \Bigr\|_{L^{2}%
(\mathbb{A}_{f})},
\end{multline*}
where we have applied Steklov-Parseval equality and the fact that
$\operatorname{supp} \widehat u_{0} \subset B_{R}$ which follows from
Proposition \ref{SlK finite dim}. By applying the Mean-Value Theorem twice we
obtain
\[%
\begin{split}
\frac{e^{-t\|\xi\|^{\alpha}} - e^{-t_{0}\|\xi\|^{\alpha}}}{t-t_{0}} +
\|\xi\|^{\alpha}e^{-t_{0}\|\xi\|^{\alpha}}  &  = -\|\xi\|^{\alpha
}e^{-t^{\prime}\|\xi\|^{\alpha}} + \|\xi\|^{\alpha}e^{-t_{0}\|\xi\|^{\alpha}%
}\\
&  = (t^{\prime}-t_{0})\|\xi\|^{2\alpha} e^{-t^{\prime\prime}\|\xi\|^{\alpha}%
},
\end{split}
\]
where $t^{\prime}= t^{\prime}(\|\xi\|)$ is a point between $t_{0}$ and $t$ and
$t^{\prime\prime}= t^{\prime\prime}(\|\xi\|)$ is a point between $t_{0}$ and
$t^{\prime}$ (and thus between $t_{0}$ and $t$). Hence
\begin{multline*}
\Bigl\| \Bigl(\frac{e^{-t\|\xi\|^{\alpha}} - e^{-t_{0}\|\xi\|^{\alpha}}%
}{t-t_{0}} + \|\xi\|^{\alpha}e^{-t_{0}\|\xi\|^{\alpha}}\Bigr) \mathbf{1}%
_{B_{R}}(\xi)\cdot\widehat u_{0}(\xi) \Bigr\|_{L^{2}(\mathbb{A}_{f})}\\
\le|t-t_{0}| R^{2\alpha} \| \widehat u_{0}(\xi)\|_{L^{2}(\mathbb{A}_{f})}
\to0,\qquad t\to t_{0},
\end{multline*}
i.e. $h_{t}(x)$ is the time derivative of the function $u(x,t)$ for any
$t\ge0$.

The proof of the continuous differentiability in time of $u(x,t)$ follows from
the time continuity of $h_{t}(x)$ which can be checked similarly.
\end{proof}

\begin{lemma}
\label{LemmaCP2} Let $u_{0}\in\mathcal{S}(\mathbb{A}_{f})$ and $u(x,t)$,
$t\ge0$ is defined by \eqref{defn_u(x,t)}. Then $u(x,t)\in\mathcal{D}%
(D^{\alpha})$ for any $t\ge0$ and
\begin{equation}
\label{eq_dalpha_u}D^{\alpha}u(x,t) = \mathcal{F}_{\xi\rightarrow x}%
^{-1}\bigl( \|\xi\|^{\alpha}e^{-t\|\xi\|^{\alpha}} \cdot\mathbf{1}_{B_{R}}%
(\xi)\bigr) \ast u_{0}(x),
\end{equation}
where $\mathbf{1}_{B_{R}}(\cdot)$ is the characteristic function of the ball
$B_{R}$, $R= (1/\ell)_{-}$ and $\ell$ is the parameter of constancy of the
function $u_{0}$, see \eqref{Local Constancy} and Proposition
\ref{SlK finite dim}.
\end{lemma}

\begin{proof}
Note that $\widehat u_{0}\in\mathcal{S}(\mathbb{A}_{f})$ which implies that
$e^{-t\Vert\xi\Vert^{\alpha}}\cdot\widehat u_{0}(\xi)\in L^{1}\left(
\mathbb{A}_{f}\right)  \cap L^{2}\left(  \mathbb{A}_{f}\right)  $ and
$\Vert\xi\Vert^{\alpha}e^{-t\Vert\xi\Vert^{\alpha}} \cdot\widehat u_{0}%
(\xi)\in L^{1}\left(  \mathbb{A}_{f}\right)  \cap L^{2}\left(  \mathbb{A}%
_{f}\right)  $ for any $t\ge0$. Hence we may calculate $D^{\alpha}u(x,t)$ by
formula (\ref{POperator}). For $t>0$ we obtain
\begin{multline*}
D^{\alpha}u(x,t)=\mathcal{F}_{\xi\rightarrow x}^{-1}\bigl(\|\xi\|^{\alpha
}\cdot\widehat u(\xi,t)\bigr)=\mathcal{F}_{\xi\rightarrow x}^{-1}\big(\Vert
\xi\Vert^{\alpha}\widehat{Z}_{t}(\xi)\cdot\widehat{u}_{0}(\xi)\big)\\
=\mathcal{F}_{\xi\rightarrow x}^{-1}\big(\Vert\xi\Vert^{\alpha}e^{-t\left\Vert
\xi\right\Vert ^{\alpha}}\mathbf{1}_{B_{R}}(\xi)\cdot\widehat{u}_{0}%
(\xi)\big)=\mathcal{F}_{\xi\rightarrow x}^{-1}\bigl( \|\xi\|^{\alpha
}e^{-t\|\xi\|^{\alpha}} \cdot\mathbf{1}_{B_{R}}(\xi)\bigr) \ast u_{0}(x),
\end{multline*}
where we have used the fact that $\operatorname{supp} \widehat u_{0} \subset
B_{R}$.

For $t=0$ we obtain
\begin{multline*}
D^{\alpha}u(x,0)=D^{\alpha}u_{0}(x)=\mathcal{F}_{\xi\rightarrow x}%
^{-1}\bigl(\|\xi\|^{\alpha}\cdot\widehat u_{0}(\xi)\bigr)=\\
=\mathcal{F}_{\xi\rightarrow x}^{-1}\big(\Vert\xi\Vert^{\alpha}e^{-0\left\Vert
\xi\right\Vert ^{\alpha}}\mathbf{1}_{B_{R}}(\xi)\cdot\widehat{u}_{0}%
(\xi)\big)=\mathcal{F}_{\xi\rightarrow x}^{-1}\bigl( \|\xi\|^{\alpha
}e^{-0\|\xi\|^{\alpha}} \cdot\mathbf{1}_{B_{R}}(\xi)\bigr) \ast u_{0}(x).
\end{multline*}
\vskip-1.2em
\end{proof}

As an immediate consequence from Lemmas \ref{LemmaCP1} and \ref{LemmaCP2} we obtain

\begin{proposition}
\label{Corr3}Let the function $u_{0}\in\mathcal{S}(\mathbb{A}_{f})$. Then the
function $u(x,t)$ defined by \eqref{defn_u(x,t)} is a solution of Cauchy
problem (\ref{CauchyProb}).
\end{proposition}

\subsection{\label{initialValL2}Homogeneous equations with initial values in
$L^{2}(\mathbb{A}_{f})$}

Consider the operator $T(t)$, $t\ge0$ of convolution with the heat kernel,
i.e.
\begin{equation}
\label{ConvZu}T(t)u = Z_{t} \ast u.
\end{equation}
Since $Z_{t} \in L^{2}(\mathbb{A}_{f})$, the convolution $Z_{t} \ast u$ is a
continuous function of $x$ for $t>0$ and any $u\in L^{2}(\mathbb{A}_{f})$, see
\cite[Theorem 1.1.6]{Rudin}.

\begin{lemma}
\label{LemmaConvZBdd} The operator $T(t):L^{2}(\mathbb{A}_{f})\to
L^{2}(\mathbb{A}_{f})$ is bounded.
\end{lemma}

\begin{proof}
Consider a function $u\in L^{2}(\mathbb{A}_{f})$. Since $Z_{t} \in
L^{1}(\mathbb{A}_{f})$, see Theorem \ref{Theo1} (iii), by the Young inequality
and Theorem \ref{Theo1} (ii)
\[
\|Z_{t} \ast u\|_{L^{2}}\le\|Z_{t}\|_{L^{1}} \cdot\| u\|_{L^{2}} = \|
u\|_{L^{2}}.
\]
Hence $T(t)u = Z_{t} \ast u\in L^{2}(\mathbb{A}_{f})$ and $\|T(t)\|\le1$.
\end{proof}

\begin{theorem}
\label{Theo4minus} Let $\alpha>1$. Then the following assertions hold.

\begin{itemize}
\item[(i)] The operator $-D^{\alpha}$ generates a $C_0$ semigroup $\bigl(\mathcal{T}(t)\bigr)_{t\ge 0}$. The operator $\mathcal{T}(t)$ coincides for each $t\ge0$ with the operator $T(t)$ given by \eqref{ConvZu}.

\item[(ii)] Cauchy problem \eqref{CauchyProb} is well-posed and its solution
is given by $u(x,t)=Z_{t}\ast u_{0}$, $t\geq0$.
\end{itemize}
\end{theorem}

\begin{proof}According to Lemma \ref{DgammaC0semigroup}, the operator $-D^{\alpha}$ generates a $C_0$ semigroup $\bigl(\mathcal{T}(t)\bigr)_{t\ge 0}$.
Hence  Cauchy problem \eqref{CauchyProb} is well-posed, see e.g.
\cite[Theorem 3.1.1]{C-H}. By Proposition \ref{Corr3}, $\left.  \mathcal{T}
(t)\right\vert _{\mathcal{S}(\mathbb{A}_{f})}=\left.  T(t)\right\vert
_{\mathcal{S}(\mathbb{A}_{f})}$ and both operators $\mathcal{T}(t)$ and $T(t)$
are defined on the whole $L^{2}(\mathbb{A}_{f})$ and bounded. By the
continuity we conclude that $\mathcal{T}(t)=T(t)$ on $L^{2}\left(
\mathbb{A}_{f}\right)  $. Now the statements follow from well-known results of
the semigroup theory, see e.g. \cite[Proposition 3.1.9.]{Aren}, \cite[Theorem 3.1.1]{C-H}, \cite[Ch. 2, Proposition 6.2]{E-N}.
\end{proof}

\begin{remark}
Since the semigroup $\bigl(\mathcal{T}(t)\bigr)_{t\ge 0}$ is holomorphic, Cauchy problem \eqref{CauchyProb} possesses \emph{smoothing effect}, see e.g. \cite[Corollary 3.7.21]{Aren}. More precisely, consider Cauchy problem \eqref{CauchyProb} with weaker requirement on the initial value, namely, let $u_{0}\in L^{2}(\mathbb{A}_{f})$. Then there exists a unique function
\[
u(x,t)\in C\bigl([0,\infty),L^{2}(\mathbb{A}_{f})\bigr)\cap C\bigl((0,\infty
),\mathcal{D}(D^{\alpha})\bigr)\cap C^{\infty}\bigl((0,\infty),L^{2}(\mathbb{A}
_{f})\bigr)
\]
satisfying the equation for $t>0$ and satisfying the initial condition. That is, this weaker Cauchy problem is solvable for arbitrary initial data and the solution is infinitely differentiable in $t$ for $t>0$.
\end{remark}

\subsection{Homogeneous equations with initial values in $\mathcal{L}%
_{0}(\mathbb{A}_{f})$}

We now consider the Cauchy problem
\begin{equation}
\left\{
\begin{aligned} &\frac{\partial u(x,t)}{\partial t}+D^{\alpha}u(x,t)=0, && x\in \mathbb{A}_{f},\ t\in\left[ 0,+\infty\right),\\ &u(x,0)=u_{0}(x), && u_{0}(x)\in\mathcal{L}_0(\mathbb{A}_f), \end{aligned}\right.
\label{CauchyProbS}%
\end{equation}
with the initial value from the space $\mathcal{L}_{0}(\mathbb{A}_{f})$ and
the pseudodifferential operator $D^{\alpha}$ with the smaller domain
$\mathcal{D}(D^{\alpha})=\mathcal{L}_{0}(\mathbb{A}_{f})$.

We say that a function $u(x,t)$ is \textit{a classical} \textit{solution of}
(\ref{CauchyProbS}), if $u\in C^{1}\bigl([0,\infty),\mathcal{L}_{0}(\mathbb{A}_{f})\bigr)$ and $u$ satisfies equation \eqref{CauchyProb} for all
$t\geq0$, with the understanding that all the involved limits are taken in the
topology of $\mathcal{L}_{0}(\mathbb{A}_{f})$.

\begin{lemma}
\label{Lemma41} Let $u_{0}\in\mathcal{L}_{0}(\mathbb{A}_{f})$ and the function
$u(x,t)$ is defined by \eqref{defn_u(x,t)}. Then $u(x,t)\in\mathcal{L}%
_{0}(\mathbb{A}_{f})$ for any $t>0$.
\end{lemma}

\begin{proof}
Since $e^{-t\|\xi\|^{\alpha}}$ is locally constant outside of the origin, the
function $h_{t}(\xi) = e^{-t\|\xi\|^{\alpha}} \cdot\widehat u_{0}(\xi
)\in\mathcal{S}(\mathbb{A}_{f})$ with $h_{t}(0)=0$. Then $\mathcal{F}%
_{\xi\rightarrow x}^{-1}\bigl(e^{-t\|\xi\|^{\alpha}} \cdot\widehat u_{0}%
(\xi)\bigr) = Z_{t}(x)\ast u_{0}(x) \in\mathcal{L}_{0}(\mathbb{A}_{f})$ for
$t>0$.
\end{proof}

\begin{theorem}
\label{Theo3}Let the function $u_{0}\in\mathcal{L}_{0}(\mathbb{A}_{f})$. Then
the function $u(x,t)$ defined by \eqref{defn_u(x,t)} is the classical solution
of Cauchy problem (\ref{CauchyProbS}).
\end{theorem}

\begin{proof}
By Lemma \ref{Lemma41} the function $u(x,t)$ is correctly defined and
$u(x,t)\in\mathcal{D}(D^{\alpha}) = \mathcal{L}_{0}(\mathbb{A}_{f})$ for all
$t\ge0$.

We assert that there exist constants $\ell$ and $R$ not dependent on $t$ such
that $u(x,t)\in\mathcal{S}_{R}^{\ell}$ and $D^{\alpha}u(x,t)\in\mathcal{S}%
_{R}^{\ell}$ for all $t\geq0$. Consider the function $h_{t}(\xi)=e^{-t\Vert
\xi\Vert^{\alpha}}\cdot\widehat{u}_{0}(\xi)$, $t\geq0$. Since $u_{0}%
\in\mathcal{L}_{0}$, the function $e^{-t\Vert\xi\Vert^{\alpha}}$ is locally
constant on the support of $\widehat{u}_{0}$. Moreover, the parameter of
constancy of $e^{-t\Vert\xi\Vert^{\alpha}}$ on the support of $\widehat{u}%
_{0}$ does not depend on $t$. Hence there exist parameters $\ell^{\prime}$ and
$R^{\prime}$ such that $h_{t}\in\mathcal{S}_{R^{\prime}}^{\ell^{\prime}}$ for
any $t\geq0$. By Proposition \ref{SlK finite dim} we have $\bigl(\mathcal{F}%
^{-1}h_{t}\bigr)(x)=u(x,t)\in\mathcal{S}_{(1/\ell^{\prime})_{-}}%
^{(1/R^{\prime})_{-}}$ for any $t\geq0$. Similar proof works for the function
$g_{t}(\xi):=\mathcal{F}(D^{\alpha}u(x,t))=\Vert\xi\Vert^{\alpha}e^{-t\Vert
\xi\Vert^{\alpha}}\cdot\widehat{u}_{0}(\xi)$.

We recall that for finite dimensional spaces, the uniform convergence is
equivalent to the $L^{2}$-convergence. Since $\mathcal{S}_{R}^{\ell}$ is a
finite dimensional space, cf. Proposition \ref{SlK finite dim}, by\ applying
Lemmas \ref{LemmaCP1} and \ref{LemmaCP2}, we have $\frac{\partial u}{\partial
t}(x,t)\in\mathcal{L}_{0}(\mathbb{A}_{f})$,  $u(x,t)$ is a solution of
Cauchy problem (\ref{CauchyProbS}) and $u\in C^{1}\bigl([0,\infty
),\mathcal{L}_{0}(\mathbb{A}_{f})\bigr)$.
\end{proof}

\subsection{Non homogeneous equations}\label{SubsectNHE}

Consider the following Cauchy problem
\begin{equation}
\left\{
\begin{aligned} &\frac{\partial u(x,t)}{\partial t}+D^{\alpha}u(x,t)=f(x,t), && x\in \mathbb{A}_{f},\ t\in[0,T],\ T>0,\\ &u(x,0)=u_{0}(x), && u_{0}(x)\in\mathcal{D}(D^\alpha). \end{aligned}\right.
\label{CauchyProb2}%
\end{equation}
We say that a function $u(x,t)$ is a \textit{solution of} (\ref{CauchyProb2}),
if $u\in C\bigl([0,T],\mathcal{D}(D^{\alpha})\bigr)\cap C^{1}%
\bigl([0,T],L^{2}(\mathbb{A}_{f})\bigr)$ and if $u$ satisfies equation
\eqref{CauchyProb2} for $t\in[0,T]$.

\begin{theorem}
\label{Theo4}Let $\alpha>1$ and let $f\in C\bigl([0,T], L^{2}(\mathbb{A}%
_{f})\bigr)$. Assume that at least one of the following conditions is satisfied:

\begin{itemize}
\item[(i)] $f\in L^{1}\bigl((0,T), \mathcal{D}(D^{\alpha})\bigr)$;

\item[(ii)] $f\in W^{1,1}\bigl((0,T), L^{2}(\mathbb{A}_{f})\bigr)$.
\end{itemize}

Then Cauchy problem \eqref{CauchyProb2} has a unique solution given by
\[
u(x,t)=\int_{\mathbb{A}_{f}} Z\left(  x-y,t\right)  u_{0}\left(  y\right)
dy_{\mathbb{A}_{f}}+\int_{0}^{t} \biggl\{
\int_{\mathbb{A}_{f}} Z\left(  x-y,t-\tau\right)  f\left(  y,\tau\right)
dy_{\mathbb{A}_{f}}\biggr\}  d\tau.
\]

\end{theorem}

\begin{proof}
With the use of Theorem \ref{Theo4minus} the proof follows from well-known
results of the semigroup theory, see e.g. \cite[Proposition 3.1.16]{Aren}, \cite[Proposition 4.1.6]{C-H}.
\end{proof}

\section{\label{SEctHeatKernelA}The Adelic Heat Kernel on $\mathbb{A}$}

We recall that the \textit{Archimedean heat kernel} is defined as%
\[
Z(x_{\infty},t;\beta)=\int_{\mathbb{R}}\chi_{\infty}\left(  \xi_{\infty
}x_{\infty}\right)  e^{-t\left\vert \xi_{\infty}\right\vert _{\infty}^{\beta}%
}d\xi_{\infty},\quad t>0,\ \beta\in\left(  0,2\right]  .
\]
This heat kernel is a solution of the pseudodifferential equation
\[
\frac{\partial u\left(  x_{\infty},t\right)  }{\partial t}+\mathcal{F}%
_{\xi_{\infty}\rightarrow x_{\infty}}^{-1}\left(  \left\vert \xi_{\infty
}\right\vert _{\infty}^{\beta}\mathcal{F}_{x_{\infty}\rightarrow\xi_{\infty}%
}^{-1}u\left(  x_{\infty},t\right)  \right)  .
\]
For a more detailed discussion of the Archimedean heat kernel and its
properties the reader may consult \cite[Section 2]{D-G-V} and references therein.

From now on we will denote heat kernel \eqref{AdelicHeat} as $Z(x_{f}%
,t;\alpha)$.

\begin{definition}
\label{DefHeatKerA}For fixed $\alpha>1$, $\beta\in\left(  0,2\right]  $ we
define \textit{the heat kernel on} $\mathbb{A}$ as
\[
Z\left(  x,t;\alpha,\beta\right)  :=\int_{\mathbb{A}}\chi\left(  -\xi\cdot
x\right)  e^{-t\left(  \left\vert \xi_{\infty}\right\vert _{\infty}^{\beta
}+\left\Vert \xi_{f}\right\Vert ^{\alpha}\right)  }d\xi_{\mathbb{A}},\quad
x\in\mathbb{A},\ t>0.
\]
\end{definition}
Since $e^{-t\left\vert \xi_{\infty}\right\vert _{\infty}^{\beta}}\in
L^{1}\left(  \mathbb{R},d\xi_{\infty}\right)  $, $e^{-t\left\Vert \xi
_{f}\right\Vert ^{\alpha}}\in L^{1}\left(  \mathbb{A}_{f},d\xi_{\mathbb{A}%
_{f}}\right)  $, cf. \cite[Property 2.2]{D-G-V} and Proposition \ref{pro1},
and $d\xi_{\mathbb{A}}=d\xi_{\infty}d\xi_{\mathbb{A}_{f}}$, we have
\begin{equation}
Z\left(  x,t;\alpha,\beta\right)  =\mathcal{F}^{-1}\bigl(e^{-t\left\vert
\xi_{\infty}\right\vert _{\infty}^{\beta}}\bigr)\mathcal{F}^{-1}%
\bigl(e^{-t\left\Vert \xi_{f}\right\Vert ^{\alpha}}\bigr)=Z(x_{\infty}%
,t;\beta)Z\left(  x_{f},t;\alpha\right)  . \label{heat_kernels}%
\end{equation}
For $t>0$ fixed, we use the notation $Z_{t}\left(  x;\alpha,\beta\right)  $
instead of $Z\left(  x,t;\alpha,\beta\right)  $.

\begin{theorem}
\label{Theo5}The adelic heat kernel on $\mathbb{A}$ possesses the following properties

\begin{itemize}
\item[(i)] $Z\left(  x,t;\alpha,\beta\right)  \geq0$ for any $t>0$;

\item[(ii)] $\int_{\mathbb{A}}Z\left(  x,t;\alpha,\beta\right)  dx_{\mathbb{A}%
}=1$ for any $t>0$;

\item[(iii)] $Z_{t}\left(  x;\alpha,\beta\right)  \in L^{1}(\mathbb{A})$ for
any $t>0$;

\item[(iv)] $Z_{t}\left(  x;\alpha,\beta\right)  \ast Z_{t^{\prime}}\left(
x;\alpha,\beta\right)  =Z_{t+t^{\prime}}\left(  x;\alpha,\beta\right)  $ for
any $t,t^{\prime}>0$;

\item[(v)] $\lim_{t\rightarrow0{+}}Z\left(  x,t;\alpha,\beta\right)
=\delta\left(  x_{\infty}\right)  \times\delta\left(  x_{f}\right)
=\delta\left(  x\right)  $ in $S^{\prime}(\mathbb{A})$;

\item[(vi)] $Z_{t}\left(  x;\alpha,\beta\right)  $ is a uniformly
continuous function for any fixed $t>0$;

\item[(vii)] $Z(x,t;\alpha,\beta)$ is uniformly continuous in $t$, i.e. $Z(x,t;\alpha,\beta)\in
C((0,\infty), C(\mathbb{A}))$ or $\lim_{t^{\prime}\to t}\max
_{x\in\mathbb{A}}|Z(x,t;\alpha,\beta) - Z(x,t^{\prime};\alpha,\beta)|=0$ for any $t>0$.
\end{itemize}
\end{theorem}

\begin{proof}
The statement follows from (\ref{heat_kernels}) and the corresponding
properties for \linebreak$Z(x_{\infty},t;\beta)$ and $Z\left(  x_{f}%
,t;\alpha\right)  $, see \cite[Section 2]{D-G-V} and Theorem \ref{Theo1}.
\end{proof}

\section{\label{SEctMarkovA}Markov Processes on $\mathbb{A}$}

Let $\mathcal{B}(\mathbb{A})$ denote the $\sigma$-algebra of the Borel sets of
$\left(  \mathbb{A},\rho_{\mathbb{A}}\right)  $. Along this section we suppose
that $\alpha>1$ and $\beta\in\left(  0,2\right]  $ are fixed parameters. We
set
\[
p\left(  t,x,y;\alpha,\beta\right)  :=Z(x-y,t;\alpha,\beta)\qquad\text{for }
t>0,x,y\in\mathbb{A}.
\]
Note that
\begin{multline*}
p\left(  t,x,y;\alpha,\beta\right)  =Z(x_{\infty}-y_{\infty} ,t;\beta)Z\left(
x_{f}-y_{f},t;\alpha\right) \\
=:p\left(  t,x_{\infty},y_{\infty};\beta\right)  p\left(  t,x_{f}
,y_{f};\alpha\right)  ,
\end{multline*}
where $p\left(  t,x_{\infty},y_{\infty};\beta\right)  =Z(x_{\infty}-y_{\infty
},t;\beta)$ and $p\left(  t,x_{f},y_{f};\alpha\right)  :=p\left(
t,x_{f},y_{f}\right)  = Z(x_{f}-y_{f}, t; \alpha)$. We also define for
$x_{\infty},y_{\infty}\in\mathbb{R}$ and $B_{\infty}\in\mathcal{B}\left(
\mathbb{R}\right)  $
\[
P\left(  t,x_{\infty},B_{\infty};\beta\right)  :=
\begin{cases}
\int_{B_{\infty}}p\left(  t,x_{\infty},y_{\infty};\beta\right)  dy_{\infty}, &
\text{for }t>0,\\
\mathbf{1}_{B_{\infty}}\left(  x_{\infty}\right)  , & \text{for }t=0
\end{cases}
\]
and for $x,y\in\mathbb{A}$ and $B\in\mathcal{B}\left(  \mathbb{A}\right)  $
\[
P\left(  t,x,B;\alpha,\beta\right)  :=
\begin{cases}
\int_{B}p\left(  t,x_{\infty},y_{\infty};\beta\right)  p\left(  t,x_{f}%
,y_{f};\alpha\right)  dy_{\infty}dy_{\mathbb{A}_{f}}, & \text{for }t>0,\\
\mathbf{1}_{B}\left(  x\right)  , & \text{for }t=0.
\end{cases}
\]

\begin{lemma}
\label{lema5}With the above notation the following assertions hold:

\begin{itemize}
\item[(i)] $p\left(  t,x,y;\alpha,\beta\right)  $ is a normal transition density;

\item[(ii)] $P\left(  t,x,B;\alpha,\beta\right)  $ is a normal transition function.
\end{itemize}
\end{lemma}

\begin{proof}
The statement follows from the corresponding properties for the functions
$p\left(  t,x_{\infty},y_{\infty};\beta\right)  $ and $p\left(  t,x_{f}%
,y_{f}\right)  $, see Lemma \ref{lema1}.
\end{proof}

\begin{lemma}
\label{lema6}The transition function $P\left(  t,x,B;\alpha,\beta\right)  $
satisfies the following two conditions:

\begin{itemize}
\item[(i)] for each $u\geq0$ and a compact $B$
\[
\lim_{x\rightarrow\infty}\sup_{t\leq u}P\left(  t,x,B;\alpha,\beta\right)
=0;\qquad[\text{Condition }L(B)]
\]

\item[(ii)] for each $\epsilon>0$ and a compact $B$
\[
\lim_{t\rightarrow0{+}}\sup_{x\in B}P\Bigl(  t,x,\mathbb{A} \setminus
\overset{\circ}{B}_{\epsilon}\left(  x\right)  ; \alpha,\beta\Bigr)
=0,\qquad[\text{Condition }M(B)]
\]
where $\overset{\circ}{B}_{\epsilon}\left(  x\right)  :=\left\{
y\in\mathbb{A}:\rho_{\mathbb{A}}\left(  x,y\right)  <\epsilon\right\}  $.
\end{itemize}
\end{lemma}

\begin{proof}
(i) Note that there exist compact subsets $K_{\infty}\subset\mathbb{R}$ and
$K_{f}\subset\mathbb{A}_{f}$ such that $B\subset K_{\infty}\times K_{f}$.
Then
\[
P(t,x,B; \alpha,\beta)\leq P(t,x_{\infty},K_{\infty}; \beta)P(t,x_{f},K_{f}).
\]

Since $\rho_{\mathbb{A}}\left(  0,x\right)  \rightarrow\infty$ we have either
$\rho\left(  0,x_{f}\right)  \rightarrow\infty$ or $| x_{\infty}| _{\infty
}\rightarrow\infty$. Therefore it is sufficient to show that
\begin{equation}
\lim_{x_{f}\rightarrow\infty}\sup_{t\leq u}P(t,x_{f},K_{f})=0 \label{eq10}%
\end{equation}
and
\begin{equation}
\lim_{x_{\infty}\rightarrow\infty}\sup_{t\leq u}P(t,x_{\infty},K_{\infty
};\beta)=0. \label{eq11}%
\end{equation}

The equality (\ref{eq10}) follows from Lemma \ref{lema22}. By \cite[(2.2)]%
{D-G-V}%
\begin{equation}
Z(t,x_{\infty};\beta)\leq\frac{Ct^{\frac{1}{\beta}}}{t^{\frac{2}{\beta}%
}+x_{\infty}^{2}}\qquad\text{for }t>0,x_{\infty}\in\mathbb{R}.
\label{heatkernerlint}%
\end{equation}
Then
\[
P(t,x_{\infty},K_{\infty};\beta)=\int_{K_{\infty}} Z(t,x_{\infty}-y_{\infty
};\beta)dy_{\infty}\leq Ct^{\frac{1}{\beta}}\int_{K_{\infty}} \frac
{1}{t^{\frac{2}{\beta}}+\left(  x_{\infty}-y_{\infty}\right)  ^{2}}dy_{\infty
}.
\]
As $x_{\infty}\rightarrow\infty$ we have $\operatorname{dist}\left(
x_{\infty},K_{\infty}\right)  \to\infty$ and $|x_{\infty}-y_{\infty}%
|\ge\operatorname{dist}(x_{\infty}; K_{\infty})$ for any $y_{\infty}\in
K_{\infty}$ and
\[
\frac{1}{t^{\frac{2}{\beta}}+\left(  x_{\infty}-y_{\infty}\right)  ^{2} }%
\le\frac1 {\operatorname{dist}^{2}(x_{\infty}; K_{\infty})}.
\]
Hence
\[
\lim_{x_{\infty}\rightarrow\infty}\sup_{t\leq u}P(t,x_{\infty},K_{\infty};
\beta) \leq\lim_{x_{\infty}\rightarrow\infty}Cu^{\frac{1}{\beta}}%
\int_{K_{\infty}} \frac{1}{\operatorname{dist}^{2}(x_{\infty}; K_{\infty}%
)}dy_{\infty}=0.
\]

(ii) Since $\overset{\circ}{B}_{\epsilon}\left(  x\right)  \supseteq
\overset{\circ}{B}_{\frac{\epsilon}{2}}\left(  x_{\infty}\right)  \times
{B}_{\frac{\epsilon}{2}}\left(  x_{f}\right)  $, where
\[
\overset{\circ}{B}_{\frac{\epsilon}{2}}\left(  x_{\infty}\right)  =\left\{
y_{\infty}\in\mathbb{R}:\left|  x_{\infty}-y_{\infty}\right|  <\frac{\epsilon
}{2}\right\}
\]
and $B_{\frac{\epsilon}{2}}\left(  x_{f}\right)  $ is given by
\eqref{adelic_ball}, we have $\mathbb{A}\setminus\overset{\circ}{B}_{\epsilon
}\left(  x\right)  \subseteq\bigl(\mathbb{R}\times\mathbb{A}_{f}%
\bigr)\setminus\bigl(\overset{\circ}{B}_{\frac{\epsilon}{2}}\left(  x_{\infty
}\right)  \times{B}_{\frac{\epsilon}{2}}\left(  x_{f}\right)  \bigr)\subseteq
\Bigl(\bigl(\mathbb{R}\setminus\overset{\circ}{B}_{\frac{\epsilon}{2}}\left(
x_{\infty}\right)  \bigr)\times\mathbb{A}_{f}\Bigr) \cup\Bigl(\mathbb{R}%
\times\bigl(\mathbb{A}_{f}\setminus{B}_{\frac{\epsilon}{2}}\left(
x_{f}\right)  \bigr)\Bigr)$ and with the use of \cite[(2.1)]{D-G-V} and
Theorem \ref{Theo1} (ii) we obtain
\begin{multline*}
P\bigl( t,x,\mathbb{A}\setminus\overset{\circ}{B} _{\epsilon}\left(  x\right)
; \alpha,\beta\bigr) \leq\biggl( \int_{| x_{\infty}-y_{\infty}| _{\infty}%
\geq\frac{\epsilon}{2}} p\left(  t,x_{\infty},y_{\infty};\beta\right)
dy_{\infty}\biggr)\\
+\biggl(\int_{\rho\left(  x_{f},y_{f}\right)  >\frac{\epsilon}{2}} p\left(
t,x_{f},y_{f}\right)  dy_{\mathbb{A}_{f}}\biggr)\\
\le P\bigl(t, x_{\infty}, \mathbb{R}\setminus\overset{\circ}{B} _{\frac
\epsilon2}\left(  x\right)  ; \beta\bigr) + P\bigl( t,x_{f},\mathbb{A}%
_{f}\setminus B_{\frac{\epsilon}{2}}\left(  x_{f}\right)  \bigr).
\end{multline*}
Now the result follows from Lemma \ref{lema22} and the inequality
\begin{multline*}
P\bigl(t, x_{\infty}, \mathbb{R}\setminus\overset{\circ}{B} _{\frac\epsilon
2}\left(  x\right)  \bigr) = \int_{|y_{\infty}|\ge\frac\epsilon2}
Z(t,y_{\infty};\beta)\,dy_{\infty}\le C\int_{|y_{\infty}|\ge\frac\epsilon
2}\frac{t^{\frac1\beta}}{t^{\frac2\beta}+y_{\infty}^{2}}\,dy_{\infty}\\
= C\int_{|z_{\infty}|\ge\frac\epsilon2 t^{-1/\beta}}\frac1{1+z_{\infty}^{2}
}dz_{\infty}\to0\qquad\text{as}\ t\to+0.
\end{multline*}
\vskip-1.2em
\end{proof}

\begin{theorem}
\label{Theo6}$Z\left(  x,t;\alpha,\beta\right)  $ with $\alpha>1$ and
$\beta\in\left(  0,2\right]  $ is the transition density of a time- and space
homogenous Markov process which is bounded, right-continuous and has no
discontinuities other than jumps.
\end{theorem}

\begin{proof}
The result follows from \cite[Theorem 3.6]{Dyn}, Remark \ref{nota3} (ii) and
Lemmas \ref{lema5}, \ref{lema6}.
\end{proof}

\section{\label{SectionCauchyA}Cauchy problem for parabolic type equations on
$\mathbb{A}$}
In this section we study Cauchy problems for parabolic type equations on $\mathbb{A}$ and present analogues of the results of
Subsections \ref{initialValSA}, \ref{initialValL2} and \ref{SubsectNHE}.

\subsection{Homogeneous equations}
Consider the following Cauchy problem
\begin{equation}
\left\{
\begin{aligned}
& \frac{\partial u\left(  x,t\right)  }{\partial t}+D^{\alpha,\beta}u\left(
x,t\right)  =0, && x\in\mathbb{A},\ t\in\left[  0,+\infty\right)  ,\\
& u(x,0)=u_{0}(x), && u_{0}(x)\in\mathcal{D}\left(  D^{\alpha,\beta}\right)  ,
\end{aligned}
\right.  \label{CauchyProb3}%
\end{equation}
where $\alpha>1$, $\beta\in\left(  0,2\right]  $, $D^{\alpha,\beta}$ is the
pseudodifferential operator defined by (\ref{POperator2}) with the domain given by
\eqref{DabDom} and $u:\mathbb{A}\times\lbrack0,\infty)\rightarrow\mathbb{C}$
is an unknown function. We say that a function $u(x,t)$\ is a \textit{solution
of} (\ref{CauchyProb3}), if $u\in C\bigl([0,\infty),\mathcal{D}(D^{\alpha
,\beta})\bigr)\cap C^{1}\bigl([0,\infty),L^{2}(\mathbb{A})\bigr)$ and if
$u$ satisfies equation (\ref{CauchyProb3}) for all $t\geq0$.

As in Section \ref{SectCauchy} we understand
the notions of continuity, differentiability and equalities in the sense of  $L^{2}(\mathbb{A})$.

We first consider Cauchy problem (\ref{CauchyProb3}) with the initial value from
$\mathcal{S}(\mathbb{A})$. We define the function
\begin{equation}
u(x,t):=u(x,t;\alpha,\beta)=Z_{t}(x;\alpha,\beta)\ast u_{0}(x)=Z_{t}(x)\ast
u_{0}(x),\qquad t\geq0, \label{solAdelic}%
\end{equation}
where $Z_{0}\ast u_{0}=\left.  \big(Z_{t}\ast u_{0}\big)\right\vert
_{t=0}:=u_{0}$. Note that such definition is consistent with Theorem
\ref{Theo5} (v). Since $Z_{t}\left(  x\right)  \in L^{1}\left(  \mathbb{A}%
\right)  $ for $t>0$ and $u_{0}\left(  x\right)  \in\mathcal{S}(\mathbb{A}%
)\subset L^{\infty}\left(  \mathbb{A}\right)  $, the convolution exists and is
a continuous function, see Theorem \ref{Theo5} (ii), \cite[Theorem
1.1.6]{Rudin}.

\begin{lemma}
\label{lemmaAdeli1}Let $u_{0}\in\mathcal{S}(\mathbb{A})$ and $u$ is defined by (\ref{solAdelic}). Then $u\in C\bigl([0,\infty),\mathcal{D}(D^{\alpha,\beta})\bigr)$ and
\begin{equation}
D^{\alpha,\beta}u=\mathcal{F}_{\xi\rightarrow x}^{-1}\left(  \left(
|\xi_{\infty}|_{\infty}^{\beta}+\left\Vert \xi_{f}\right\Vert ^{\alpha
}\right)  e^{-t\left(  |\xi_{\infty}|_{\infty}^{\beta}+\left\Vert \xi
_{f}\right\Vert ^{\alpha}\right)  }\mathcal{F}_{x\rightarrow\xi}u_{0}\right)
\label{formDAdelic}
\end{equation}
for $t\geq0$.
\end{lemma}

\begin{proof}
We first verify that $u(\cdot,t)\in\mathcal{D}(D^{\alpha,\beta})$ for $t\geq
0$. Without loss of generality we may assume that $u_{0}(x)=u_{\infty
}(x_{\infty})u_{f}(x_{f})$ with $u_{\infty}\in\mathcal{S}(\mathbb{R})$ and
$u_{f}\in\mathcal{S}(\mathbb{A}_{f})$. Since
\[
\mathcal{F}_{x\rightarrow \xi}u(x,t)=e^{-t\left(  |\xi_{\infty}|_{\infty
}^{\beta}+\left\Vert \xi_{f}\right\Vert ^{\alpha}\right)  }\widehat{u}%
_{\infty}\left(  \xi_{\infty}\right)  \widehat{u}_{f}\left(  \xi_{f}\right)
,
\]
we have
\begin{multline*}\displaybreak[2]
  \left\Vert \left(  |\xi_{\infty}|_{\infty}^{\beta}+\left\Vert \xi
_{f}\right\Vert ^{\alpha}\right)  \mathcal{F}_{x\rightarrow \xi}u\right\Vert
_{L^{2}(\mathbb{A})}\\
  \leq\bigl\Vert   \left\Vert \xi_{f}\right\Vert ^{\alpha
}e^{-t\left\Vert \xi_{f}\right\Vert ^{\alpha}}\widehat{u}_{f}\left(  \xi
_{f}\right) \bigr\Vert _{L^{2}
(\mathbb{A}_f)} \cdot
 \bigl\Vert   e^{-t|\xi_{\infty}|_{\infty}^{\beta}}\widehat
{u}_{\infty}\left(  \xi_{\infty}\right)  \bigr\Vert _{L^{2}
(\mathbb{R})}\\
  +\bigl\Vert   e^{-t\left\Vert \xi_{f}\right\Vert ^{\alpha}}
\widehat{u}_{f}\left(  \xi_{f}\right)  \bigr\Vert_{L^{2}
(\mathbb{A}_f)}
 \cdot \bigl\Vert    |\xi_{\infty}
|_{\infty}^{\beta}e^{-t|\xi_{\infty}|_{\infty}^{\beta}}\widehat{u}_{\infty
}\left(  \xi_{\infty}\right) \bigr\Vert _{L^{2}(\mathbb{R})}\\
\le \bigl\Vert   \left\Vert \xi_{f}\right\Vert ^{\alpha
}\widehat{u}_{f}\left(  \xi
_{f}\right) \bigr\Vert _{L^{2}
(\mathbb{A}_f)} \cdot
 \bigl\Vert   \widehat
{u}_{\infty}  \bigr\Vert _{L^{2}
(\mathbb{R})}
  +\bigl\Vert
\widehat{u}_{f}  \bigr\Vert_{L^{2}
(\mathbb{A}_f)}
 \cdot \bigl\Vert    |\xi_{\infty}
|_{\infty}^{\beta}\widehat{u}_{\infty
}\left(  \xi_{\infty}\right) \bigr\Vert _{L^{2}(\mathbb{R})}\\
= \left\Vert D^{\alpha}u_{f}\right\Vert _{L^{2}\left(  \mathbb{A}
_{f}\right)  }\left\Vert u_{\infty}\right\Vert _{L^{2}(\mathbb{R}
)} + \left\Vert u_{f}\right\Vert _{L^{2}\left(  \mathbb{A}
_{f}\right)  }\left\Vert D^\beta u_{\infty}\right\Vert _{L^{2}(\mathbb{R})},
\end{multline*}
where we used the Parseval-Steklov equality and the equality $d\xi_{\mathbb{A}}=d\xi_{\mathbb{A}_{f}}d\xi_{\infty}$.
Therefore $u(x,t)\in\mathcal{D}(D^{\alpha,\beta})$ for $t\geq0$ and formula
(\ref{formDAdelic}) holds.

To verify the continuity, assume again that $u_{0}(x)=u_{\infty}(x_{\infty
})u_{f}(x_{f})$ with $u_{\infty}\in\mathcal{S}(\mathbb{R})$, $u_{f}%
\in\mathcal{S}(\mathbb{A}_{f})$. With the use of the Parseval-Steklov equality and the Mean Value Theorem we obtain
\begin{multline*}
 \lim_{t^{\prime}\rightarrow t}\left\Vert u\left(  x,t\right)  -u\left(
x,t^{\prime}\right)  \right\Vert _{L^{2}(\mathbb{A})}\\
  =\lim_{t^{\prime}\rightarrow t}\bigl\Vert \bigl(  e^{-t(  |\xi
_{\infty}|_{\infty}^{\beta}+\Vert \xi_{f}\Vert ^{\alpha})
}-e^{-t^{\prime}(  |\xi_{\infty}|_{\infty}^{\beta}+\Vert \xi
_{f}\Vert ^{\alpha})  }\bigr)  \widehat{u}_{\infty}\left(
\xi_{\infty}\right)  \widehat{u}_{f}\left(  \xi_{f}\right)  \bigr\Vert
_{L^{2}(\mathbb{A})}\\
 =\lim_{t^{\prime}\rightarrow t}\bigl\Vert \left(  t-t^{\prime}\right)
\left(  |\xi_{\infty}|_{\infty}^{\beta}+\left\Vert \xi_{f}\right\Vert
^{\alpha}\right)  e^{-\tilde{t}\cdot(  |\xi_{\infty}|_{\infty}^{\beta
}+\left\Vert \xi_{f}\right\Vert ^{\alpha})  }\widehat{u}_{\infty}\left(
\xi_{\infty}\right)  \widehat{u}_{f}\left(  \xi_{f}\right)  \bigr\Vert
_{L^{2}(\mathbb{A})}\\
 \le \lim_{t^{\prime}\rightarrow t}|t-t'|\cdot\bigl\Vert
\left(  |\xi_{\infty}|_{\infty}^{\beta}+\left\Vert \xi_{f}\right\Vert
^{\alpha}\right)  \widehat{u}_{\infty}\left(
\xi_{\infty}\right)  \widehat{u}_{f}\left(  \xi_{f}\right)  \bigr\Vert
_{L^{2}(\mathbb{A})}\\
 \leq\bigl(  \left\Vert D^{\alpha}u_{f}\right\Vert _{L^{2}\left(
\mathbb{A}_{f}\right)  }\left\Vert u_{\infty}\right\Vert _{L^{2}(\mathbb{R}
)}+\left\Vert u_{f}\right\Vert _{L^{2}\left(  \mathbb{A}_{f}\right)
} \bigl\Vert D^\beta u_{\infty}\bigr\Vert _{L^{2}(\mathbb{R}
)} \bigr)  \lim_{t^{\prime}\rightarrow t}\left\vert t-t^{\prime}\right\vert =0,
\end{multline*}
where $\tilde t = \tilde t\bigl(|\xi_{\infty}|_{\infty}^{\beta}+\left\Vert \xi_{f}\right\Vert
^{\alpha}\bigr)$ is a point between $t$ and $t'$.
\end{proof}

\begin{lemma}
\label{lemmaAdeli2}Let $u_{0}\in\mathcal{S}(\mathbb{A})$ and $u(x,t)$,
$t\geq0$ is defined by (\ref{solAdelic}). Then $u(x,t)$ is continuously
differentiable in time for $t\geq0$ and the derivative is given by%
\begin{equation}
\frac{\partial u}{\partial t}(x,t)=-\mathcal{F}_{\xi\rightarrow x}%
^{-1}\bigl(\left(  |\xi_{\infty}|_{\infty}^{\beta}+\left\Vert \xi
_{f}\right\Vert ^{\alpha}\right)  e^{-t(  |\xi_{\infty}|_{\infty}^{\beta
}+\left\Vert \xi_{f}\right\Vert ^{\alpha})  }\mathcal{F}_{x\rightarrow
\xi}u_{0}\bigr). \label{DeriAdelic}%
\end{equation}
\end{lemma}

\begin{proof}Assume
that $u_{0}(x)=u_{\infty}(x_{\infty})u_{f}(x_{f})$ with $u_{\infty}%
\in\mathcal{S}(\mathbb{R})$, $u_{f}\in\mathcal{S}(\mathbb{A}_{f})$.
By reasoning as in the proofs of Lemmas \ref{LemmaCP1} and \ref{lemmaAdeli1}, we have
\begin{multline*}
\lim_{t\rightarrow t_{0}}\Bigl\Vert \frac{\widehat{u}(\xi,t)-\widehat{u}%
(\xi,t_{0})}{t-t_{0}}+\left(  |\xi_{\infty}|_{\infty}^{\beta}+\left\Vert
\xi_{f}\right\Vert ^{\alpha}\right)  e^{-t\left(  |\xi_{\infty}|_{\infty
}^{\beta}+\left\Vert \xi_{f}\right\Vert ^{\alpha}\right)  }\mathcal{F}
_{x\rightarrow\xi}u_{0}\Bigr\Vert _{L^{2}(\mathbb{A})}\\
\lim_{t\rightarrow t_{0}} |t-t_0| \cdot\bigl\Vert \left(  |\xi_{\infty}|_{\infty}^{\beta}+\left\Vert
\xi_{f}\right\Vert ^{\alpha}\right)^2  e^{-\tilde t\cdot\left(  |\xi_{\infty}|_{\infty
}^{\beta}+\left\Vert \xi_{f}\right\Vert ^{\alpha}\right)  }\mathcal{F}
_{x\rightarrow\xi}u_{0}\bigr\Vert _{L^{2}(\mathbb{A})}\\
\le \lim_{t\rightarrow t_{0}} |t-t_0| \cdot\bigl\Vert (  |\xi_{\infty}|_{\infty}^{\beta}+\left\Vert
\xi_{f}\right\Vert ^{\alpha})^2 \mathcal{F}
_{x\rightarrow\xi}u_{0}\bigr\Vert _{L^{2}(\mathbb{A})}\\
\leq\bigl(  \Vert D^{2\alpha}u_{f}\Vert \cdot\Vert u_{\infty}\Vert+ 2\Vert D^{\alpha}u_{f}\Vert \cdot\Vert D^{\beta}u_{\infty}\Vert+\Vert u_{f}\Vert\cdot \Vert D^{2\beta} u_{\infty}\Vert \bigr)  \lim_{t^{\prime}\rightarrow t}\left\vert t-t^{\prime}\right\vert =0,
\end{multline*}
where we have used the fact that $\mathcal{S}(\mathbb{R})\subset\mathcal{D}(D^\beta)$ for any $\beta>0$ and $\mathcal{S}(\mathbb{A}_f)\subset\mathcal{D}(D^\alpha)$ for any $\alpha>0$.

To verify the continuity of $\frac{\partial u}{\partial t}(x,t)$, we proceed similarly:
\begin{multline*}
  \lim_{t\rightarrow t_{0}}\Bigl\Vert \frac{\partial u}{\partial
t}(x,t)-\frac{\partial u}{\partial t}(x,t_{0})\Bigr\Vert _{L^{2}(\mathbb{A})}\\
  =\lim_{t\rightarrow t_{0}}\left\vert t_{0}-t\right\vert \bigl\Vert(
|\xi_{\infty}|_{\infty}^{\beta}+\left\Vert \xi_{f}\right\Vert ^{\alpha
})^{2}e^{-\widetilde{t}\left(  |\xi_{\infty}|_{\infty}^{\beta
}+\left\Vert \xi_{f}\right\Vert ^{\alpha}\right)  }\widehat{u}_{\infty}\left(
\xi_{\infty}\right)  \widehat{u}_{f}\left(  \xi_{f}\right)  \bigr\Vert
_{L^{2}(\mathbb{A})}\\
  \leq\lim_{t\rightarrow t_{0}}\left\vert t_{0}-t\right\vert \bigl\Vert
(|\xi_{\infty}|_{\infty}^{\beta}+\left\Vert \xi_{f}\right\Vert ^{\alpha
})^{2}\widehat{u}_{\infty
}\left(  \xi_{\infty}\right)\widehat{u}_{f}\left(  \xi_{f}\right)  \bigr\Vert _{L^{2}(\mathbb{A})} = 0.
\end{multline*}
where we used the Mean Value Theorem with a point $\widetilde{t}$ between $t$ and $t_{0}$.
\end{proof}

As an immediate consequence from Lemmas \ref{lemmaAdeli1} and
\ref{lemmaAdeli2} we obtain

\begin{proposition}
\label{propoAdelic}Let the function $u_{0}\in\mathcal{S}(\mathbb{A})$. Then
the function $u(x,t)$ defined by (\ref{solAdelic}) is a solution of Cauchy
problem (\ref{CauchyProb3}).
\end{proposition}

Consider the operator $T(t;\alpha,\beta)$, $t\geq0$ of convolution with the
adelic heat kernel
\begin{equation}
T(t;\alpha,\beta)u=Z_{t}\ast u. \label{Tadelic}%
\end{equation}
As in Section \ref{SectCauchy}, the convolution $Z_{t}\ast u$ is a continuous
function of $x$ for $t>0$ and any $u\in L^{2}(\mathbb{A})$ and the operator
$T(t;\alpha,\beta):L^{2}(\mathbb{A})\rightarrow L^{2}(\mathbb{A})$ is bounded.

By reasoning as in the proof of Theorem \ref{Theo4minus}, we obtain
\begin{theorem}
\label{Theo9} Let $\alpha>1$ and $\beta\in\left(  0,2\right]$. Then the
following assertions hold.

\begin{itemize}
\item[(i)] The
operator $-D^{\alpha,\beta}$ generates a $C_0$ semigroup $\bigl(\mathcal{T}(t;\alpha,\beta)\bigr)_{t\ge 0}$. The operator $\mathcal{T}(t;\alpha,\beta)$ coincides for each $t\geq0$ with the operator
$T(t;\alpha,\beta)$ given by (\ref{Tadelic}).

\item[(ii)] Cauchy problem (\ref{CauchyProb3}) is well-posed and its solution
is given by $u(x,t)=Z_{t}\ast u_{0}$, $t\geq0$.
\end{itemize}
\end{theorem}

\subsection{Non homogeneous equations}

Consider the following Cauchy problem
\begin{equation}
\left\{
\begin{aligned}
&\frac{\partial u\left(  x,t\right)  }{\partial t}+D^{\alpha,\beta}u\left(
x,t\right)  =f\left(  x,t\right)  , && x\in\mathbb{A},\ t\in\left[  0,T\right]
,\ T>0,\\
& u(x,0)=u_{0}(x), && u_{0}(x)\in\mathcal{D}\left(  D^{\alpha,\beta}\right)  .
\end{aligned}
\right.  \label{CauchyProb5}%
\end{equation}

We say that a function $u(x,t)$ is a \textit{solution of} (\ref{CauchyProb5}),
if $u\in C\bigl([0,T],\mathcal{D}(D^{\alpha,\beta})\bigr)\cap
C^{1}\bigl([0,T],L^{2}(\mathbb{A})\bigr)$ and if $u$ satisfies
equation (\ref{CauchyProb5}) for $t\in[0,T]$.

\begin{theorem}
\label{Theo10}Let $\alpha>1$, $\beta\in\left(  0,2\right]  $ and let $f\in
C\bigl([0,T],L^{2}(\mathbb{A})\bigr)$. Assume that at least one of the
following conditions is satisfied:

\begin{itemize}
\item[(i)] $f\in L^{1}\bigl((0,T),\mathcal{D}(D^{\alpha,\beta})\bigr)$;

\item[(ii)] $f\in W^{1,1}\bigl((0,T),L^{2}(\mathbb{A})\bigr)$.
\end{itemize}
Then Cauchy problem (\ref{CauchyProb5}) has a unique solution given by
\[
u(x,t)=\int_{\mathbb{A}}Z\left(  x-y,t;\alpha,\beta\right)
u_{0}\left(  y\right)  dy_{\mathbb{A}}+\int_{0}^{t}\biggl\{
\int_{\mathbb{A}}Z\left(  x-y,t-\tau;\alpha,\beta\right)  f\left(
y,\tau\right)  dy_{\mathbb{A}}\biggr\}  d\tau.
\]
\end{theorem}

\begin{proof}
With the use of Theorem \ref{Theo9} the proof follows from well-known results
of the semigroup theory, see e.g. \cite[Proposition 3.1.16]{Aren}, \cite[Proposition 4.1.6]{C-H}.
\end{proof}


\begin{thebibliography}{99}                                                                                               %

\bibitem{A-B}S. Albeverio, Y. Belopolskaya, Stochastic processes in
$\mathbb{Q}_{p}$ associated with systems of nonlinear PIDEs, $p$-Adic Numbers
Ultrametric Anal. Appl. 1, no. 2, 105--117 (2009)

\bibitem{A-K-S}S. Albeverio, A. Y. Khrennikov, V. M. Shelkovich, Theory of
$p$-adic distributions: linear and nonlinear models, Cambridge University
Press, 2010

\bibitem{Aren}W. Arendt, C. J. K. Batty, M. Hieber, F. Neubrander,
Vector-valued Laplace transforms and Cauchy problems, Birkh\"{a}user/Springer, 2011

\bibitem{Av-1}V. A. Avetisov, A. Kh. Bikulov, On the ultrametricity of the
fluctuation dynamic mobility of protein molecules, Proc. Steklov Inst. Math.
265, no. 1, 75--81 (2009)

\bibitem{Av-2}V. A. Avetisov, A. Kh. Bikulov, A. P. Zubarev, First passage
time distribution and the number of returns for ultrametric random walks, J.
Phys. A 42, no. 8, 085003, 18 pp. (2009)

\bibitem{Av-3}V. A. Avetisov, A. Kh. Bikulov, V. A. Osipov, $p$-adic models of
ultrametric diffusion in the conformational dynamics of macromolecules, Proc.
Steklov Inst. Math. no. 2 (245), 48--57 (2004)

\bibitem{Av-4}V. A. Avetisov, A. Kh. Bikulov, V. A. Osipov, $p$-adic
description of characteristic relaxation in complex systems, J. Phys. A 36,
no. 15, 4239--4246 (2003)

\bibitem{Av-5}V. A. Avetisov, A. H. Bikulov, S. V. Kozyrev, V. A Osipov,
$p$-adic models of ultrametric diffusion constrained by hierarchical energy
landscapes. J. Phys. A 35, no. 2, 177--189 (2002)

\bibitem{Av-6}V. A. Avetisov, A. H. Bikulov, S. V. Kozyrev, Application of
$p$-adic analysis to models of breaking of replica symmetry, J. Phys. A 32,
no. 50, 8785--8791 (1999)

\bibitem{Av-7}V. A. Avetisov, A. Kh. Bikulov, S. V. Kozyrev, Description of
logarithmic relaxation by a model of a hierarchical random walk. (Russian)
Dokl. Akad. Nauk 368, no. 2, 164--167 (1999)

\bibitem{B}O. Beloshapka, Feynman formulas for an infinite-dimensional
$p$-adic heat type equation, Infin. Dimens. Anal. Quantum Probab. Relat. Top.
14, no. 1, 137--148 (2011)

\bibitem{Blair}A. D. Blair, Ad\`{e}lic path space integrals, Rev. Math. Phys. 7, no. 1, 21--49 (1995)

\bibitem{C-H}T. Cazenave, A. Haraux, An introduction to semilinear evolution
equations, Oxford University Press, 1998

\bibitem{C}A. Connes, Trace formula in noncommutative geometry and the zeros of the Riemann zeta
function, Selecta Math. (N.S.), 5, 29--106 (1999)

\bibitem{D}H. Diamond, Elementary methods in the study of the distribution of
prime numbers, Bull. Amer. Math. Soc. (N.S.) 7, no. 3, 553--589 (1982)

\bibitem{Dra}B. Dragovich, $p$-adic and adelic quantum mechanics, Proc.
Steklov Inst. Math. no. 2 (245), 64--77 (2004)

\bibitem{Dra-Kh-K-V}B. Dragovich, A. Y. Khrennikov, S. V. Kozyrev, I. V.
Volovich, On $p$-adic mathematical physics, $p$-Adic Numbers Ultrametric Anal.
Appl. 1, no. 1, 1--17 (2009)

\bibitem{D-R-K}B. Dragovich, Y. Radyno, A. Khrennikov, Generalized functions
on adeles, J. Math. Sci. (N. Y.) 142, no. 3, 2105--2112 (2007)

\bibitem{D-G-V}J. Droniou, T. Gallouet, J. Vovelle, Global solution and
smoothing effect for a non-local regularization of a hyperbolic equation, J.
Evol. Equ. 3, no. 3, 499--521 (2003)

\bibitem{Dyn}E. B. Dynkin, Markov processes, Vol. I, Springer-Verlag, 1965

\bibitem{E-N}K.-J. Engel, R. Nagel, One-Parameter Semigroups for Linear
Evolution Equations, Springer-Verlag, 2000

\bibitem{Ga-Zu}J. Galeano-Peñaloza, W. A. Z\'{u}ñiga-Galindo, Pseudo-differential operators with semi-quasielliptic symbols over $p$-adic fields, J. Math. Anal. Appl. 386, no. 1, 32--49 (2012)

\bibitem{G-H}D. Goldfeld, J. Hundley, Automorphic representations and
L-functions for the general linear group. Volume I, Cambridge University
Press, 2011

\bibitem{Haran}S. Haran, Riesz potentials and explicit sums in arithmetic,
Invent. Math. 101, no. 3, 697--703 (1990)

\bibitem{H-S-S-S}D. Harlow, S. Shenker, D. Stanford, L. Susskind, Tree-like
structure of eternal inflation: A solvable model, Phys. Rev. D 85, no. 6,
Article Number: 063516 (2012)

\bibitem{K-M}W. Karwowski, R. V. Mendes, Hierarchical structures and
asymmetric stochastic processes on p-adics and adèles, J. Math. Phys. 35, no.
9, 4637--4650 (1994)

\bibitem{Koch}A. N. Kochubei, Pseudo-differential equations and stochastics
over non-Archimedean fields, Marcel Dekker, 2001

\bibitem{K-A-1}A. N. Kochubei, M. R. Sait-Ametov, Construction of interaction
measures on the space of distributions over the field of $p$-adic numbers,
Proc. Steklov Inst. Math. no. 2 (245), 135--142 (2004)

\bibitem{K-A-2}A. N. Kochubei, M. R. Sait-Ametov, Interaction measures on the
space of distributions over the field of $p$-adic numbers, Infin. Dimens.
Anal. Quantum Probab. Relat. Top. 6, no. 3, 389--411 (2003)

\bibitem{Kh-K-Sh}A. Yu. Khrennikov, A. V. Kosyak, V. M. Shelkovich, Wavelet
analysis on adeles and pseudo-differential operators, to appear in J. Fourier Anal. Appl., available at arXiv:1107.1700

\bibitem{K-M-1}A. Khrennikov, F. Mukhamedov, On uniqueness of Gibbs measure
for $p$-adic countable state Potts model on the Cayley tree, Nonlinear Anal.
71, no. 11, 5327--5331 (2009)

\bibitem{K-M-2}A. Y. Khrennikov, F. M. Mukhamedov, J. F. F. Mendes, On
$p$-adic Gibbs measures of the countable state Potts model on the Cayley tree,
Nonlinearity 20, no. 12, 2923--2937 (2007)

\bibitem{Kh-Ra}A. Y. Khrennikov, Y. V. Radyno, On adelic analogue of
Laplacian, Proc. Jangjeon Math. Soc. 6, no. 1, 1--18 (2003)

\bibitem{Manin}Y. I. Manin, Reflections on artithmetical physics. Conformal
Invariance and String Theory, 293-303, Academic Press, 1989.

\bibitem{P-S}G. Parisi, N. Sourlas, $p$-adic numbers and replica symmetry
breaking, Eur. Phys. J. B Condens. Matter Phys. 14, no. 3, 535--542 (2000)

\bibitem{R-R}Y. V. Radyno, Y. M. Radyna, Generalized functions on adeles.
Linear and non-linear theories, Linear and non-linear theory of generalized
functions and its applications, 243--250, Banach Center Publ., 88, Polish
Acad. Sci. Inst. Math., Warsaw, 2010

\bibitem{R-S}M. Reed, B. Simon, Methods of Modern Mathematical Physics:
Functional Analysis I, Academic Press, 1980.

\bibitem{R-V}D. Ramakrishnan, R. J. Valenza, Fourier analysis on number
fields, Springer-Verlag, 1999

\bibitem{R-Zu}J. J. Rodríguez-Vega, W. A. Zúñiga-Galindo, Taibleson
operators, $p$-adic parabolic equations and ultrametric diffusion, Pacific J.
Math. 237, no. 2, 327--347 (2008)

\bibitem{Rudin}W. Rudin, Fourier analysis on groups, Interscience Publishers, 1962

\bibitem{Samko}S. G. Samko, Hypersingular integrals and their applications, Taylor \& Francis, 2002

\bibitem{S-K-M}S. G. Samko, A. A. Kilbas, and O. I. Marichev, Fractional Integrals and Derivatives and Some of Their Applications, Nauka i Tekhnika, 1987 (in Russian)

\bibitem{Taibleson}M. H. Taibleson, Fourier analysis on local fields,
Princeton University Press, 1975

\bibitem{Va1}V. S. Varadarajan, Path integrals for a class of $p$-adic
Schrödinger equations, Lett. Math. Phys. 39, no. 2, 97--106 (1997)

\bibitem{Va2}V. S. Varadarajan, Arithmetic quantum physics: why, what, and
whither, Proc. Steklov Inst. Math. no. 2 (245), 258--265 (2004)

\bibitem{V-V-Z}V. S. Vladimirov, I. V. Volovich, E. I. Zelenov, $p$-adic
analysis and mathematical physics, World Scientific, 1994

\bibitem{Vo}I. V. Volovich, Number theory as the ultimate physical theory.
$p$-Adic Numbers Ultrametric Anal. Appl. 2, no. 1, 77--87 (2010)

\bibitem{Vol2}I. V. Volovich, $p$-adic string, Classical Quantum Gravity 4
(1987), no. 4, L83--L87.

\bibitem{We}A. Weil, Basic number theory, Springer-Verlag, 1967

\bibitem{Ya}K. Yasuda, Markov processes on the adeles and representations of
Euler products, J. Theoret. Probab. 23, no. 3, 748--769 (2010)

\bibitem{Zu}W. A. Zúñiga-Galindo, Parabolic equations and Markov processes
over $p$-adic fields, Potential Anal. 28, no. 2, 185--200 (2008)
\end{thebibliography}
\end{document}